\documentclass[sigconf, nonacm]{acmart}

\usepackage{graphicx}
\graphicspath{{Images/}}

\usepackage{balance}
\usepackage{caption}
\usepackage{subcaption}   

\usepackage{multirow}
\usepackage{amsthm}

\usepackage{soul}
\usepackage{enumitem}
\usepackage[procnumbered,ruled,linesnumbered,vlined]{algorithm2e}
\usepackage{tcolorbox}
\usepackage[normalem]{ulem}
\usepackage{xspace}

\SetAlFnt{\small}

\newcommand{\nop}[1]{}

\newtheorem{lemma}{Lemma}

\newtheorem{definition}{Definition}
\newtheorem{example}{Example}

\usepackage{mathtools} 

\newcommand\vldbdoi{XX.XX/XXX.XX}
\newcommand\vldbpages{XXX-XXX}
\newcommand\vldbvolume{14}
\newcommand\vldbissue{1}
\newcommand\vldbyear{2020}
\newcommand\vldbauthors{\authors}
\newcommand\vldbtitle{\shorttitle} 
\newcommand\vldbavailabilityurl{URL_TO_YOUR_ARTIFACTS}
\newcommand\vldbpagestyle{plain} 

\begin{document}
\title{Keyword-based Community Search in Bipartite Spatial-Social Networks (Technical Report)}

\author{Kovan A. Bavi}
\affiliation{%
  \institution{Kent State University}
  \city{Kent}
  \state{Ohio}
  \country{USA}
}
\affiliation{%
  \institution{University of Zakho}
  \city{Zakho}
  \country{Kurdistan Region, Iraq}
}
\email{kmali@kent.edu}
\email{kovan.m.ali@uoz.edu.krd}

\author{Xiang Lian}
\affiliation{%
  \institution{Kent State University}
  \city{Kent}
  \state{Ohio}
  \country{USA}
}
\email{xlian@kent.edu}

\begin{abstract}
Several approaches have been recently proposed for community search in bipartite graphs. These methods have shown promising results in identifying communities in real-world bipartite networks, such as social and biological networks. Given a query user $q$, community search in bipartite graphs involves identifying a group of users containing $q$, with common characteristics or functions within a given bipartite graph. These problems are particularly challenging because bipartite graphs have two distinct sets of nodes, and community search algorithms must account for this structure. However, finding communities in keyword-based bipartite spatial-social networks has yet to be investigated enough. The spatial-social networks are naturally structured as bipartite graphs. Thus, this paper proposes a new community search problem in Bipartite spatial-social networks with a novel $(\omega, \pi)\mbox{-}keyword\mbox{-}core$, named \textit{Keyword-based Community Search in Bipartite Spatial-Social Networks} ($KCS\mbox{-}BSSN$). The $KCS\mbox{-}BSSN$ returns a tightly-knit community, significant social influence, minimal travel distance, and includes a $(\omega, \pi)\mbox{-}keyword\mbox{-}core$. To address the  $KCS\mbox{-}BSSN$ problem, we have developed pruning methods that effectively filter out irrelevant users and points of interest. To improve query-answering efficiency, we have also proposed an indexing technique named the bipartite-spatial-social index. Our pruning techniques, and indexing approach, have proven effective and efficient through experiments with real and artificial data sets.
\end{abstract}

\maketitle

\pagestyle{\vldbpagestyle}
\begingroup\small\noindent\raggedright\textbf{PVLDB Reference Format:}\\
\vldbauthors. \vldbtitle. PVLDB, \vldbvolume(\vldbissue): \vldbpages, \vldbyear.\\
\href{https://doi.org/\vldbdoi}{doi:\vldbdoi}
\endgroup
\begingroup
\renewcommand\thefootnote{}\footnote{\noindent
This work is licensed under the Creative Commons BY-NC-ND 4.0 International License. Visit \url{https://creativecommons.org/licenses/by-nc-nd/4.0/} to view a copy of this license. For any use beyond those covered by this license, obtain permission by emailing \href{mailto:info@vldb.org}{info@vldb.org}. Copyright is held by the owner/author(s). Publication rights licensed to the VLDB Endowment. \\
\raggedright Proceedings of the VLDB Endowment, Vol. \vldbvolume, No. \vldbissue\ %
ISSN 2150-8097. \\
\href{https://doi.org/\vldbdoi}{doi:\vldbdoi} \\
}\addtocounter{footnote}{-1}\endgroup

\ifdefempty{\vldbavailabilityurl}{}{
\begingroup\small\noindent\raggedright\textbf{PVLDB Artifact Availability:}\\
The source code, data, and/or other artifacts have been made available at \url{https://github.com/KBavi-Personal/KCS-BSSN_WITH_DYNAMIC}.
\endgroup
}

\section{Introduction}
\label{sec:intro}

The community search problem has recently attracted significant attention due to its wide range of practical applications in areas such as marketing \cite{Wang2006}, recommendation systems \cite{kpcore-wang-2023, ee-want-2021}, and team formation \cite{zhang2025topr, ee-want-2021}. With the rapid growth of location-based social networks (LBSNs), an enormous volume of social and spatial data has become available. These networks provide a comprehensive representation of human interactions in physical space by capturing both social relationships and mobility behaviors.

Spatial-social networks inherently consist of two distinct but interrelated components: users and spatial locations. Users are connected to their checked-in points of interest (POIs), which are embedded within a road network. Consequently, such networks naturally form a bipartite structure, where one partition represents users and the other represents POIs, and edges exist only between these two sets. Bipartite graphs thus model the interactions between upper-level entities (users) and lower-level entities (POIs), reflecting users’ spatial activities and preferences.

Recently, several studies have investigated community search in bipartite graphs under various bipartite-core models \cite{10.1016/j.ins.2023.119511, LI2024111961, zhang2025topr, Temp-bipartite-Li-2024, CSABiGXu2023, Scalable-CS-wang-2024}. While these works focus on identifying cohesive structures within bipartite networks, research on spatial-social networks often overlooks their inherent bipartite nature. In many existing approaches, social and spatial data are analyzed separately when performing community search, treating the social graph and the road network as independent components. Such separation fails to capture the complex interactions between users and points of interest. To effectively form a meaningful community in spatial-social networks, several interrelated factors must be considered simultaneously:
(i) structural cohesiveness among users,
(ii) significant social influence within the group,
(iii) minimal travel distance to relevant POIs, and
(iv) the inherent bipartite relationships between users and POIs. Ignoring any of these aspects may lead to communities that are structurally valid but practically ineffective.

Social networks play a crucial role in shaping users’ opinions and decision-making processes on specific topics. Influence propagation within tightly connected communities can significantly affect behavioral adoption. Meanwhile, users’ checked-in locations correspond to POIs described by keywords, representing shared interests and preferences. By jointly considering social influence and POI-related keyword information, it becomes possible to identify communities that are both socially cohesive and semantically aligned with specific interests. Analyzing these dimensions independently cannot adequately model the complex interplay of social relationships and spatial behaviors. Therefore, it is essential to design a unified framework that fully exploits the bipartite nature of spatial-social networks.

To address this gap, we introduce a new community search problem, namely \textit{Keyword-based Community Search in Bipartite Spatial-Social Networks} ($KCS\mbox{-}BSSN$). The objective of $KCS\mbox{-}BSSN$ is to retrieve a cohesive community that:
(1) contains a query user $q$,
(2) satisfies strong structural cohesiveness in the social network,
(3) exhibits significant social influence,
(4) minimizes travel distance between community members and relevant POIs, and
(5) satisfies a $(\omega, \pi)\mbox{-}keyword\mbox{-}core$ constraint in the bipartite network. This integrated formulation enables the discovery of communities that are structurally strong, influential, spatially convenient, and semantically meaningful.

We further extend our $KCS\mbox{-}BSSN$ framework to support Temporal Bipartite Spatial-Social Networks ($TBSSN$). By integrating temporal constraints into the core discovery process, we ensure that the identified communities are not only socially and spatially cohesive but also reflect the latest user behaviors (within the most recent sliding window). Our proposed approach can effectively filter out historical or obsolete interactions, providing results with high semantic relevance to current activity patterns.

Our proposed framework has a wide range of real-world applications, including urban planning, location-aware recommendation, and targeted marketing. Example~\ref{ex:resturant-ads} illustrates how our model can be applied for marketing purposes.
\setlength{\textfloatsep}{0pt}
\begin{figure}
  \centering
    \includegraphics[width=3.3in,height=1.9in]{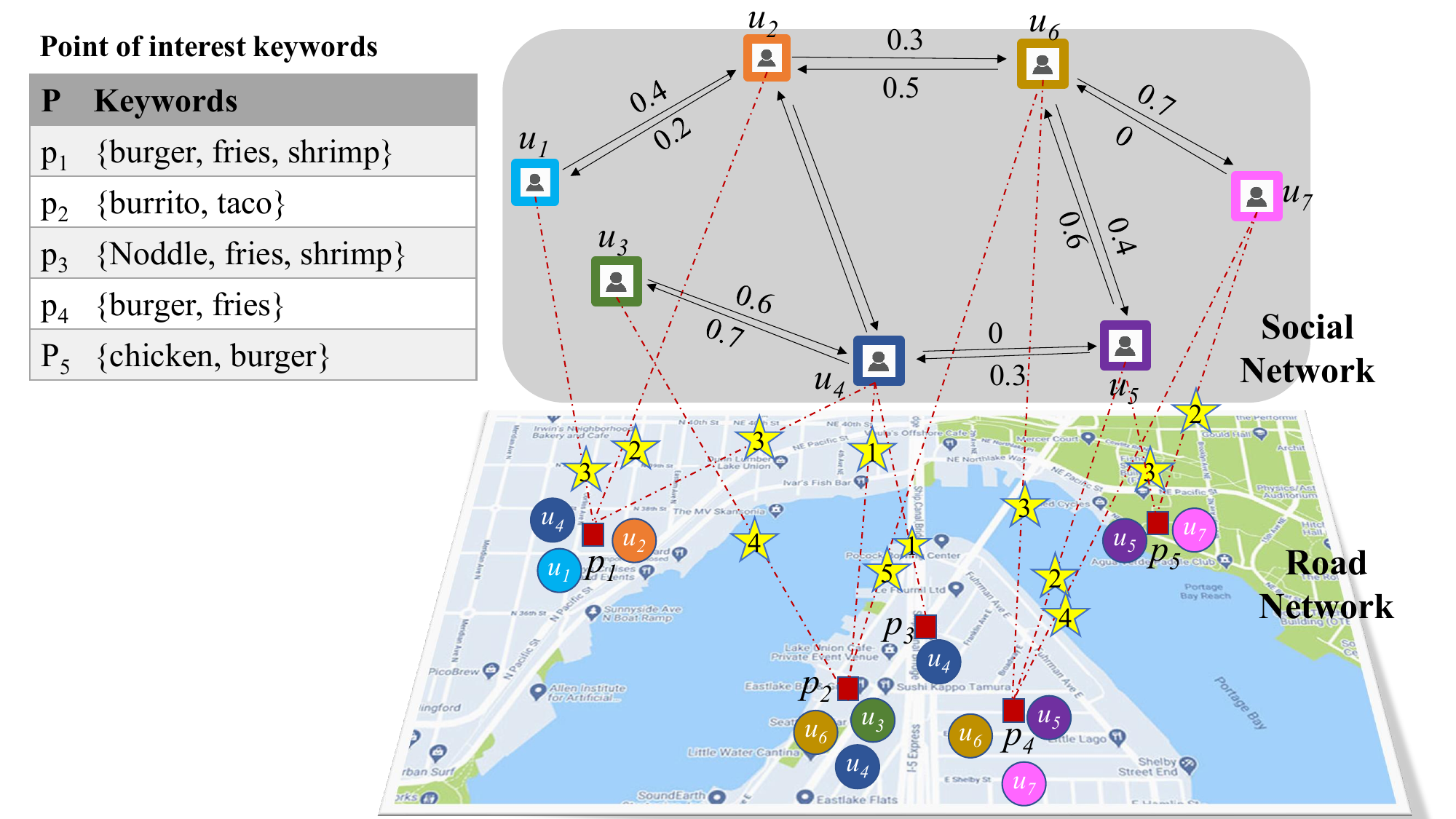}
    \caption{\small An example of bipartite spatial-social networks.}
      \Description{An example of bipartite spatial-social networks.}
    \label{fig:example1}
\end{figure}

\begin{example}
\label{ex:resturant-ads}
(\textbf{Online Advertising})
 Figure \ref{fig:example1} illustrates an example of a bipartite-spatial-social network.  In this example, the social network $G_s$ consists of vertices that represent users $u_1{\sim}u_7$, and friendship among users shown by edges (e.g., $e({u_1, u_2}))$. Each edge $e$ has a weight representing the user's influence on a specific topic (e.g., traveling) between 0 and 1. On the other hand, in the road network $G_r$, the vertices represent intersection points, and the edges represent road segments that connect the vertices. In the social network $G_s$, each user could have multiple checked-in locations on $G_r$ representing visited points of interest, where the points of interest $p_1{\sim}p_5$ are located on the road network. The users from $G_s$ visit the points of interest on $G_r$ with a positive frequency number, illustrated inside stars on dotted edges. On the other hand, each $p_i$ is associated with a list of keywords that describe that point of interest. The list of all keywords is presented in Figure \ref{fig:example1}.

 Consequently, a business wants to target its advertising toward a particular group of people who enjoy eating burgers and fries at restaurants. The chosen individuals should be geographically and socially close to each other. The group members should also be reasonably close to the restaurant. It is common for a group to choose a restaurant that some members have visited often in the past. The company may select a loyal customer as a reference user and can influence other's opinions. To ensure that the group is a good fit, each member should have prior experience visiting places that provide at least one shared service.
\end{example}

\setlength{\textfloatsep}{0pt}
\begin{figure}
  \centering
     \includegraphics[width=3.3in,height=1.9in]{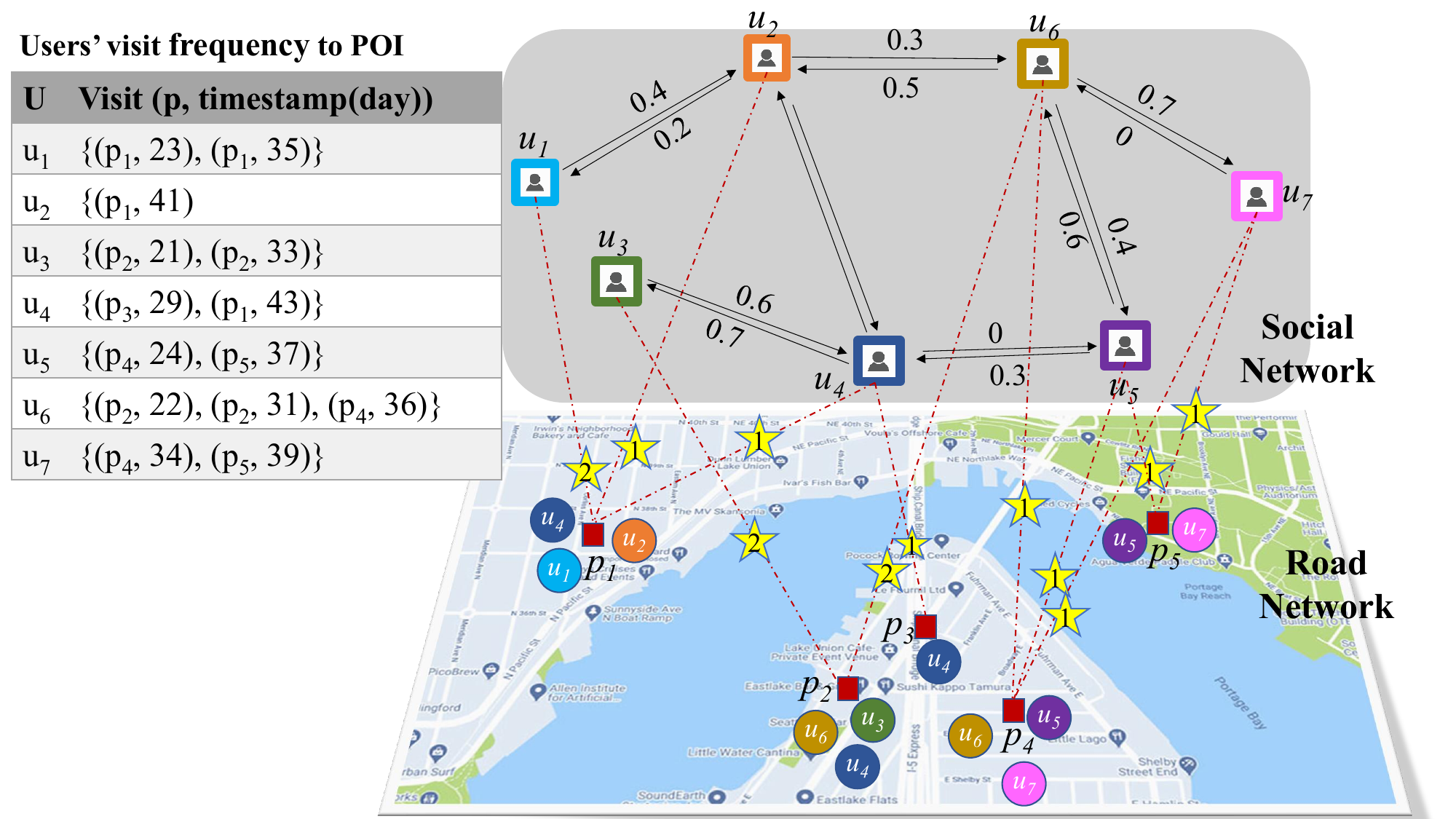}
    \caption{\small An example of a temporal bipartite spatial-social network with a visit updating frequency for the last 30 days, considering today as day \#50. }
      \Description{An example of temporal bipartite spatial-social networks.}
    \label{fig:example2}
\end{figure}

Online advertising is usually not a one-time event, but a continuous process to affect the shopping attitudes and behaviors of users in the community. It is therefore important to continuously monitor potential customer groups (communities) in the bipartite spatial-social networks over time (namely temporal $BSSN$, or $TBSSN$, with dynamic updates, e.g., visiting frequencies of POIs by users). 

We have the following motivation example on the continuous customer group monitoring for online advertising.

\begin{example}
\label{ex:resturant-ads-temporal}
(\textbf{Continuous Customer Group Monitoring for Online Advertising}) Figure \ref{fig:example2} illustrates a temporal bipartite-spatial-social network ($TBSSN$). The social network $G_s$ consists of vertices representing users $u_1{\sim}u_7$, where friendship ties are denoted by edges weighted by social influence. Within the road network $G_r$, POIs $p_1{\sim}p_5$ are situated at specific intersections and associated with descriptive keywords. Each user in $G_s$ maintains a set of checked-in locations on $G_r$; every visit from a user to a POI is associated with a specific timestamp (e.g., in days), as presented in the temporal visit vectors of Figure \ref{fig:example2}. 

Consider a business launching a limited-time burger promotion. To maximize conversion, the business seeks a target group that is geographically proximal, socially cohesive, and significantly influenced by a loyal reference user $q$. To ensure the advertisement reaches an active audience, the company imposes a recency constraint $\uptau$, that is, we are interested in those users who have frequently visited burger restaurants for the past 30 days. In this case, we can issue a $KCS\mbox{-}TBSSN$ query, which dynamically considers these historical interactions, identifying a community that is structurally and spatially connected, and actively engaged with the company's products within a valid temporal window (i.e., a recent sliding window).
\end{example}

In Examples~\ref{ex:resturant-ads} and \ref{ex:resturant-ads-temporal}, we illustrate how to identify a community of users within a social network where members maintain close relationships, and a selected influential user can affect the opinions of others. Such an influence is essential for effectively persuading the group to spread/propagate online advertising. Moreover, some members should have frequently visited the restaurant, thereby enhancing its reputation and credibility among their friends. Finally, the selected restaurant should be geographically convenient for all members, ensuring minimal travel distance and encouraging collective participation.

Most existing research on community search in bipartite graphs primarily focuses on identifying structural cores within bipartite networks. These studies emphasize structural properties but often overlook additional social and spatial dimensions. Conversely, prior work on spatial-social networks typically analyzes the social graph and the road network separately, without fully integrating their inherent bipartite relationships. In contrast, our work jointly considers multiple essential aspects. Specifically, we incorporate social cohesiveness by ensuring strong relationships among users, account for social influence by modeling how a selected user can affect others’ opinions, and enforce spatial proximity by minimizing travel distance for all community members. Furthermore, we integrate the bipartite structure by prioritizing the reputation of Points of Interest (POIs), measured through users’ visit frequencies. This unified approach enables the discovery of communities that are structurally cohesive, influence-aware, spatially compact, and semantically meaningful.
 
\noindent\textbf{Challenges.} Dealing with spatial-social networks can be challenging due to the vast amounts of information associated with these networks. In addition, the problem at hand involves many constraints that require considerable computation. Addressing these constraints separately can increase computational time. However, finding a solution that combines these constraints presents its own challenges.

\noindent\textbf{Contributions.} This paper makes the following contributions:
\begin{itemize}[noitemsep,topsep=0pt,leftmargin=10pt]
\item \textbf{New Community Model:} we define  $(\omega,\pi)\mbox{-}keyword\mbox{-}core$ a bipartite model that captures the complex interactions between social influence and road-network proximity (Section \ref{sec:awcore}).
\item \textbf{Novel Indexing:} We design a unified indexing tree and cost model that integrates social and road-network data to enable direct user-centric filtering within index nodes (Section \ref{sec:indexing-mechnism}).
\item \textbf{Efficient Algorithm:} We propose a comprehensive $KCS\mbox{-}BSSN$ algorithm featuring multi-stage pruning techniques to eliminate false-alarm users and irrelevant POIs (Sections \ref{sec:query_answering} and \ref{sec:pruning}).
 \item\textbf{Temporal Extension:} We extend $KCS\mbox{-}BSSN$ to temporal networks, integrating time-aware constraints to ensure communities reflect the most recent user behaviors (Section \ref{sec:temporal_BSSN}).
\item \textbf{Empirical Validation:} Extensive experiments on real and synthetic datasets demonstrate the superior efficiency and scalability of our solution over baseline approach (Section \ref{sec:exper}).
\end{itemize}

\section{Problem Definition}
\label{sec:problem_def}

In this section, we will formally define the data model for \textit{bipartite spatial-social networks} ($BSSN$) and our \textit{keyword-based community search} problem over $BSSN$.

\subsection{Social Networks}

In this subsection, we first give the data model for social networks below.

\begin{definition} {\bf (Social Network, $G_s$).} A social network $G_s$ is a graph in the triple form ($V_s$, $E_s$, $\theta_s$), where $V_s$ is a set of $m$ user vertices $u_1$, $u_2$, $...$, and $u_m$, $E_s$ is a set of directed edges $e(u_j,u_k)$, each connecting two users $u_j$ and $u_k$ and associated with a weight $w(u_j,u_k)$, and $\theta_s$ is a mapping function $V_s \times V_s \rightarrow E_s$. 
\label{def:sn}
\end{definition}

In Definition \ref{def:sn}, a social network can be considered as an influence graph, where the weight $w(u_j,u_k)$ of each edge $e(u_j, u_k)\in E_s$ is an influence of user $u_j$ on user $u_k$, with respect to some topic or interest of the user (e.g., movie, sports, etc.).

The influence weight between users in $G_s$ can be computed using the \textit{text-based topic discovery algorithm} \cite{barbieri2013topic}. For simplicity, in this paper, we assume that each user influences other users only with respect to one single topic. We can easily extend our proposed solution to the scenario of multiple topics \cite{chen2015online,ahmedTopicBasedCom}, by keeping influence weight vectors for different topics, which we will leave it as our future work.
 
\noindent {\bf Influence Score Function.} Assume that user vertices $u$ and $v$ can be directly or indirectly connected by a path, $u \leadsto v$, of length $(l-1)$ in social networks $G_s$, that is, $u= x_1 \rightarrow x_2 \rightarrow \dots \rightarrow x_l=v$, where $x_i$ (for $1\leq i\leq l$) is a vertex on the path. We define an \textit{influence score function} (ISF) between any two user vertices $u$ and $v$ as follows.


\begin{definition} {\bf(Influence Score Function, $ISF$ \cite{chen2015online}) } Given a social network $G_s$, the \textit{influence score}, $w(u \leadsto v)$, of user $u$ on user $v$ through a path $u \leadsto v$ with length $(l-1)$ is given by:
\begin{eqnarray}
w(u \leadsto v)= \prod_{i=1}^{l-1} w(x_i,x_{i+1}).
\label{eq:tw_path_uv}
\end{eqnarray}


The \textit{influence score function} (ISF), $ISF(u,v)$, is defined as the maximum influence score among all possible paths from $u$ to $v$:
\begin{eqnarray}
ISF(u,v)= \max_{\forall u\leadsto v}\left\{ w(u \leadsto v)\right\}.
\label{eq:ISF}
\end{eqnarray}

\label{def:InfluenceScore}
\end{definition}


Note that, the influence score function (given by Eq.~(\ref{eq:ISF})) in Definition \ref{def:InfluenceScore} is not symmetric (i.e., $ISF(u,v) \neq ISF(v,u)$). In other words, the influence of user $u$ on user $v$ could be different from that of user $v$ on user $u$.

\noindent {\bf Social Cohesiveness.} There are many existing techniques \cite{fang2020survey} to capture the cohesiveness of a community in social networks $G_s$, such as $k$-core \cite{DBLP:journals/corr/cs-DS-0310049,SEIDMAN1983269, kcore-fang-2019}, $k$-truss \cite{cohen2008trusses,10.1145/2588555.2610495, k-truss-zhang-2019, k-truss-xie-2025}, $k$-clique \cite{acquisti2006imagined, k-clique-yuan-2017}, $k$-(edge connected component) \cite{gibbons1985algorithmic, 10.1007/s00778-016-0451-4}, etc. 


In this paper, we consider the  $(k, d)\mbox{-}truss$ \cite{kdtruss-def} as follows.

\begin{definition} {\bf ($(k, d)\mbox{-}truss$  \cite{kdtruss-def}).} Given a social network $G_s$, a query user $q\in V_s$, and two positive integers $k$ ($>2$) and $d$, a $(k, d)\mbox{-}truss$ is a connected subgraph $g \subseteq G_s$, such that: (1) each edge $e \in E_s(g)$ is contained in at least $(k-2)$ triangles, and; (2) 
for any user $u \in V_s(g)$, $dist_s(q, u) \leq d$ holds, where $dist_s(q, u)$ is the shortest path distance between $q$ and $u$ on the social network $G_s$.
\label{def:kdtruss}
\end{definition}

Intuitively, the $(k, d)\mbox{-}truss$ in Definition \ref{def:kdtruss} returns a community with high connectivity of graph structures (i.e., with $\geq (k-2)$ triangles) and with friend relationships close to the query user $q$ (i.e., within $d$ hops away from $q$). 


\subsection{Spatial Road Networks}

Next, we define a spatial road network, $G_r$, as follows.

\begin{definition} {\bf (Spatial Road Network, $G_r$).} A spatial road network, $G_r$, is a planar graph, represented by a triple ($V_r$, $E_r$, $\theta_r$), where $V_r$ is a set of $n$ vertices (i.e., intersection points) $r_1, r_2, . . . ,$ and $r_n$, $E_r$ is a set of edges $e(r_j,r_k)$ (each connecting two intersection points $r_j$ and $r_k$), and $\theta_r$ is a mapping function $V_r \times V_r \rightarrow E_r$.
\label{def:rn}
\end{definition}

In Definition \ref{def:rn}, the road network $G_r$ is a graph, with intersection points as vertices and road line segments as edges. Each vertex $r_i \in V_r$ has its 2D location, $r_i.\ell$, with longitude and latitude, ($r_i.x$, $r_i.y$), in Euclidean space.

In the road network $G_r$, there are \textit{points of interest} (POIs) on road segments (edges), defined as follows.

\begin{definition} {\bf (Points of Interest, $V_p$).} Given a spatial road network $G_r$, we have a set, $V_p$, of \textit{points of interest} (POIs) on edges in $E_r$, where each POI $p\in V_p$ is associated with its 2D location $p.\ell$ $(= (p.x, p.y))$ and a list, $p.K$, of its descriptive keywords.
\label{def:poi}
\end{definition}

Examples of POIs in Definition \ref{def:poi} include restaurants, hotels, cinemas, airports, etc.


\subsection{Bipartite Spatial-Social Networks}

Users in social networks $G_s$ can have one or multiple checked-in locations (i.e., POIs), $u.L$, on spatial road networks $G_r$. Here, checked-in locations can be obtained via GPS or WiFi location services from social networks (e.g., Twitter or Yelp). 

Essentially, users over social networks $G_s$ and POI locations on spatial road networks $G_r$ can form a checked-in bipartite graph, defined as follows. 

\begin{definition} {\bf (Bipartite Spatial-Social Network, $G_b$).} Given users in $V_s$ on social networks $G_s$ and POIs in $V_p$ on spatial road networks $G_r$, a bipartite spatial-social network ($BSSN$), $G_b$, is a bipartite graph, in the form of a quadruple $(V_s, V_p, E_b, \mathcal{F}_b)$, where the edge set $E_b$ contains edges from user vertices $u \in V_s$ to POIs $p \in V_p$, and function $\mathcal{F}_b$ is a mapping: $E_b \rightarrow R^+$ that assigns each edge $e(u, p)$ with a positive, real-valued weight $f_{u,p}$.
\label{def:Gb}
\end{definition}

In Definition \ref{def:Gb}, we model the checked-in relationships between users and POIs as a bipartite graph (i.e., $BSSN$), where each mapping edge $e(u, p)$ in $E_b$ from a user $u\in V_s$ to a POI $p\in V_p$ is associated with a weight $f_{u,p}$. In practice, the edge weight $f_{u,p}$ can be the frequency that user $u$ visits the POI $p$.


\subsection{Keyword-Based Community Search Over Bipartite Spatial-Social Networks}
\label{sec:awcore}
\noindent {\bf Keyword-Weight-Constrained Community, $(\omega, \pi)\mbox{-}keyword\mbox{-}core$.} In this subsection, we will define the problem of \textit{keyword-based community search over bipartite spatial-social networks} ($KCS\mbox{-}BSSN$). Before that, we first introduce the concept of $(\omega, \pi)\mbox{-}keyword\mbox{-}core$ over bipartite spatial-social networks $G_b$.

\begin{definition} {\bf ($(\omega, \pi)\mbox{-}keyword\mbox{-}core$).} Given a bipartite spatial-social network $G_b$, a query keyword set $Q$, and parameters $\omega$ and $\pi$, a $(\omega, \pi)\mbox{-}keyword\mbox{-}core$, $B=(V_s', V_p', E'_b, \mathcal{F}'_b)$, is a connected, maximal bipartite subgraph of $G_b$, such that:
\begin{itemize}
  \item each POI $p \in V_p'$ contains at least one query keyword (i.e., $p.K\cap Q \neq\emptyset$);
  \item each user $u \in V_s'$ has the summed visiting frequency $f_{sum}(u,$ $V_p') =\sum_{p\in V_p'}f_{u,p}$ $\geq \omega$, and;
  \item each POI $p \in V_p'$ has the average visiting frequency $f_{avg}(V_s',$ $p) =\frac{\sum_{u\in V_s'}f_{u,p}}{|\{u\in V_s' | f_{u, p} \ne 0\}|}$ $\geq \pi$, where $|\{u\in V_s' | f_{u, p} \ne 0\}| \ne 0$.
\end{itemize}
\label{def:awcore}
\end{definition}

In Definition \ref{def:awcore}, we define a bipartite community (i.e., $(\omega, \pi)\text{-}$ $keyword\text{-}core$), $B$, satisfying the constraints of keywords and aggregated edge weights. In particular, each POI $p \in V_p'$ must contain at least one query keyword in $Q$ (i.e., $p.K\cap Q \neq\emptyset$). Moreover, each user $u\in V_s'$ should have his/her summed frequency of visits no less than $\omega$ (i.e. $f_{sum}(u,$ $V_p')$ $\geq \omega$), which indicates the preference of user $u$ to POIs in $V_s'$ (with the specified query keywords). Similarly, each POI $p \in V_p'$ must have the average visiting frequency $f_{avg}(V_s', p)$ higher than or equal to $\pi$, which implies the popularity of the POI $p$.



\begin{definition} {\bf (The Average Spatial Distance Function).} Given a social-network user, $u \in V_s$ and $p \in V_p$, we calculate the $avg\_dist_r(u,p)$ as the average shortest path distance between all the user $u$'s checked-in locations $u.L$ and the $p.\ell$ 
\begin{eqnarray}
avg\_dist_r(u,p)= 
    \frac{\sum_{i=1}^{|u.L|} dist_r(u.L_i, p.\ell)}{|u.L|},
    \label{eq:avgdist}
    \end{eqnarray}
where $|u.L|$ is the total number of points of interest $p$ visited by the user $u$, and $dist_r (., .)$ is the shortest path distance between two locations in the road network $G_r$.
\label{def:avgdist}
\end{definition}


    

Definition \ref{def:avgdist} gives the average distance $avg\_dist_r(u,p)$ between user $u$ and a POI $p$, which intuitively captures the spatial closeness between user $u$ and POI $p$ (visited by other users) in $G_r$. The small average distance indicates the possibility of recommending POI $p$ to user $u$ within the same bipartite community in real applications such as online marketing and advertising.


\noindent {\bf Keyword-Based Community Search Over Bipartite Spatial-Social Networks.} With constraints of keyword, (aggregated) edge weights, and average user-POI distances, we are now ready to define the keyword-aware community search over bipartite spatial-social networks ($KCS\mbox{-}BSSN$) as follows:

\begin{definition} {\bf (Keyword-Based Community Search Over Bipartite Spatial-Social Networks, $KCS\mbox{-}BSSN$).} Given a bipartite spatial-social network $G_b$, a keyword query set $Q$, parameters $k$,  $d$, $\omega$, and $\pi$, a spatial distance threshold $\sigma$, an influence score threshold $\theta$, and a query user $q$, a keyword-based community search over bipartite spatial-social networks ($KCS\mbox{-}BSSN$) returns a maximal bipartite subgraph (community), $B =(V_s', V_p', E'_b, \mathcal{F}'_b)$, of $G_b$ such that: 
\begin{itemize}
    \item $q \in V_s'$;
    \item $V_s'$ is a $(k, d)\mbox{-}truss$;
    \item for any user $u \in V_s'$, we have the influence score $ISF (q\leadsto u) \geq \theta$;
    \item for any user $u \in V_s'$ and POI $ p\in V_p'$, it holds that $avg\_dist_r (u,$ $p)$ $\leq \sigma$;
    \item the bipartite subgraph $B$ is a $(\omega, \pi)\text{-}keyword\text{-}core$, and;
    \item any subgraph $B' \supset B$ is not a $KCS\mbox{-}BSSN$ community.
\end{itemize}
\label{def:kcsbssn}
\end{definition}


Intuitively, in Definition \ref{def:kcsbssn}, the $KCS\mbox{-}BSSN$ problem returns a maximal group, $B =(V_s', V_p', E'_b, \mathcal{F}'_b)$, of users (including query user $q$) in social networks $G_s$ and their interested (or frequently visited) POIs $p \in V_p$ on road networks, where any user $u$ in $V_s'$ have impact influence by $q$ (i.e., $\geq \theta$) in $G_s$, any user $u$ in $V_s'$ has close average road-network distance to POIs $p$, and POIs $p$ in $V_p'$ contain some query keyword(s) in $Q$.


\begin{table}
 \centering
 \caption{\small Notations and Descriptions}\small  \scriptsize\vspace{-2ex}
 \begin{tabular}{l|l}\hline
 {\bf Symbol} & \qquad\qquad{\bf Description}\\
 \hline \hline
    $G_s$ & social networks \\
    $G_r$ & spatial road networks \\
    $G_b$ & bipartite spatial-social networks \\
    $V_p$ & a set of points of interest (POIs)\\
    $w(u,v)$ & the weight on edge $e(u,v)$ of $G_s$\\  
    $ISF$ & influence score function\\ 
    $V_s$ & a set of users in $G_b$\\
    $f_{u,p}$ & the frequency of a user $u$ visiting the POI $p$ in $G_b$\\ 
    $Q$ & a query keyword set\\
    $q$ & a query user\\
    $\sigma$ & a spatial distance threshold\\
    $\theta$ & an influence score threshold\\
     $\omega$ & a user summed visiting frequency threshold\\
    $\pi$ & a POI average visiting frequency threshold\\
    $\uptau$ &  a timestamp threshold\\
 \hline
\end{tabular}\vspace{2ex}
\label{notationTbl}
\end{table}


\subsection{Challenges}

The problem $KCS\mbox{-}BSSN$ (as given in Definition \ref{def:kcsbssn}) is rather challenging to tackle, in terms of efficiency. One straightforward method is as follows: we first enumerate all possible users in social networks (including the query user $q$) based on the $(k,d)\mbox{-}truss$ properties (as given in Definition \ref{def:kdtruss}) and the \textit{influence score} (as given in Definition \ref{def:InfluenceScore}), then prune all POIs $p\in V_p$ that do not contain any keywords in the query keyword set $Q$, check the constraints of $(\omega, \pi)\mbox{-}keyword\mbox{-}core$ (as given in Definition \ref{def:awcore}), and finally examine the distance between users $u\in V_s'$ and POIs $p\in V_p'$. 
This straightforward method is, however, not efficient, due to a large number of possible communities (with an exponential number of possible user-and-POI combinations). Therefore, in the sequel, we will design effective pruning techniques to filter out as many false alarms of users/POIs as possible to reduce the problem search space, and develop an indexing mechanism to enable our proposed efficient $KCS\mbox{-}BSSN$ query processing algorithm.

\section{Pruning Techniques}
\label{sec:pruning}
In this section, we provide pruning techniques to rule out as many false alarms of users in social networks and POIs on road networks as possible.  

\subsection{Keyword-Based Pruning}


In this subsection, we present a \textit{keyword-based pruning} method, which filters out those users from the social network $G_s$ based on their previous visits to POIs $V_p$ (without any keywords in the query keyword set $Q$).

\begin{lemma} {\bf (Keyword-Based Pruning).}  Given a bipartite graph $G_b$, a query keyword set $Q$, and a set, $V_p$, of POIs that user $u$ visited in $G_b$, user $u$ can be safely pruned, if $p.K \cap Q = \emptyset$ holds for all POIs $p \in V_p$, where $p.K$ is a keyword set associated with POI $p$.  
\label{lemma:keyword-based-pruning-u}
\end{lemma}
\begin{proof} The proof is provided in the Appendix \ref{appendix:sec:proofs}.
\end{proof}

\subsection{\texorpdfstring{$\omega$}{omega}-Based Pruning}



As given in Definition \ref{def:awcore}, any user $u$ who is in a community $C\in G_s$ must satisfy the condition that $f_{sum}(u,V_p')\geq \omega$. Thus, our \textit{$\omega$-based pruning} method aims to filter out any user $u$ with low $f_{sum}(u,V_p')$ (i.e., $< \omega$). However, directly computing $f_{sum}(u,V_p')$ requires an online summation of $f_{u,p}$ for all POIs $p\in V_p'$, which is rather costly. To accelerate the process, we will alternatively obtain an upper bound, $ub\_f_{sum}(u,V_p')$, of $f_{sum}(u,V_p')$ offline, and online prune a user $u$, if it holds that $ub\_f_{sum}(u,V_p')< \omega$.

We have the following $\omega$-based pruning lemma.

\begin{lemma} {\bf ($\omega$-based Pruning).} Given a user $u$, a POI set $V_p$, and a threshold $\omega$, any user $u$ can be safely pruned, if $ub\_f_{sum}(u,V_p')$ $ < \omega$ holds, where $ub\_f_{sum}(u,V_p')$ is an upper bound of $f_{sum}(u,V_p')$. 
\label{lemma:omega-based-pruning}
\end{lemma}
\begin{proof} The proof is provided in the Appendix \ref{appendix:sec:proofs}.
\end{proof}

\noindent {\bf Discussions on How to Compute $ub\_f_{sum}(u,V_p')$:} In Lemma \ref{lemma:omega-based-pruning}, we need to compute the upper bound $ub\_f_{sum}(u,V_p')$ of $f_{sum}(u,V_p')$. Note that, $V_p'$ in the community may not include all POIs that user $u$ has visited before. Therefore, we can sum up $f_{u,p}$ for all POIs in $V_p$ ($\supseteq V_p'$) that user $u$ has visited, and obtain this upper bound $ub\_f_{sum}(u,V_p')=f_{sum}(u,V_p) = \sum_{p\in V_p}f_{u,p} \geq \sum_{p\in V_p'}f_{u,p} = f_{sum}(u,V_p')$. 

\subsection{\texorpdfstring{$\pi$}{pi}-Based Pruning}

 Based on Definition \ref{def:awcore}, any POI $p$ in $(\omega, \pi)\mbox{-}keyword\mbox{-}core$ must have the average visiting frequency $f_{avg}(V_s',p)\geq \pi$, where $V_s'$($\subseteq V_s$). Then, we can filter out any POI $p$ if $f_{avg}(V_s',p) < \pi$. However, our \textit{$\pi$-based pruning} method aims to filter out any user $u$ who does not visited any $p$ with high $f_{avg}(V_s',p)$ (i.e., $\geq \pi$). Thus, to avoid calculating the average of $f_{u,p}$ online for all $p\in V_p'$ visited by user $u$, we calculated offline the upper bound $ub\_f_{avg}(u)$ (i.e.,  $ub\_f_{avg}(u)=max_{\forall p\in V_p}f_{u,p}$). The calculated $ub\_f_{avg}(u)$ can help prune $u$ for an online query if it holds that $ub\_f_{avg}(u)< \pi$.

In Lemma \ref{lemma:pi-based-pruning}, we can safely prune false alarm users based on $\pi$-based pruning.

\begin{lemma} {\bf ($\pi$-Based Pruning).} Given a user $u$, a POI set $V_p$, and thresholds $\pi$, the $u$ can be safely pruned, if $ub\_f_{avg}(u)< \pi$.
\label{lemma:pi-based-pruning}
\end{lemma}
\begin{proof} The proof is provided in the Appendix \ref{appendix:sec:proofs}.
\end{proof}

\noindent {\bf Discussions on How to Compute $ub\_f_{avg}(u)$:} In order to prune user $u$ in Lemma \ref{lemma:pi-based-pruning}, we need to compute $ub\_f_{avg}(u)$. Therefore, we obtain the upper $ub\_f_{avg}(u)=max_{\forall p\in V_p}f_{u,p}$, where the upper bound $ub\_f_{avg}(u)$ holds for all POIs $p$ in $V_p'$ ($\subseteq V_p$) that have been visited by the user $u$.

\subsection{Influence-Based Pruning}


In order to ensure that the influence score between any two users $u,v$ in $S\subseteq G_s$ satisfies the constraint $ISF ( u\leadsto v) \geq \theta$ (as given in Definition \ref{def:kcsbssn}), we need to compute the influence score between all pairs in $S$. However, it is possible for any pair to have multiple paths between them, making the calculation of $ISF$ for online queries very complex. Therefore, for each $u \in V_s$, we computed the upper bound In-influence (Out-influence) $ub\_w_{in}(u)$ ($ub\_w_{out}(u)$) offline, respectively. Later, we obtain online the upper bound influence score $ub\_ISF(.,.)$ between any two $u,v\in V_s$ based on previously calculated $ub\_w_{in}(.)$ and $ub\_w_{out}(.)$.

Then, we pruned any user $v$ if $ub\_ISF(u,v)<\theta$ as stated in the following Influence-based pruning lemma.

\begin{lemma} {\bf  (Influence-Based Pruning).}  Given a social network $G_s$, and a threshold $\theta$, for any two users $u,v \in V(G_s)$, we can safely prune $v$, if $ub\_ISF(u,v)<\theta$ holds. 
\label{lemma:Influence-based-pruning}
\end{lemma}
\begin{proof}
Derived from the Definition \ref{def:kcsbssn}.
\end{proof}

\noindent {\bf Discussions on How to Compute $ub\_ISF(u,v)$:}
Chen et al. \cite{chen2015online} proposed techniques to efficiently compute the upper bound influence among users, one of which is the \textit{Neighborhood-Based Estimation}. For each user $u \in V_s$, let in-neighbors, $\mathcal{N}_{in}(u)=\{v|\exists e(v,u)\in E_s\}$, and out-neighbors, $\mathcal{N}_{out}(u)=\{v|\exists e(u,v)\in E_s\}$. Then, we compute the In-influence (Out-influence) upper bound as $ub\_w_{in}(u)= \max_{v\in \mathcal{N}_{in}(u)}\{w(v,u)\}$ ($ub\_w_{out}(u)= \max_{v\in \mathcal{N}_{out}(u)}\{ w(u,v)\}$), respectively. 

Then, the upper bound influence score between any two vertices $u,v$ can be computed as follows:
\begin{eqnarray}
ub\_ISF(u,v)= \begin{cases}
    \max\{ub\_w_{out}(u), ub\_w_{in}(v)\}, & \text{if $e(u,v)\in E_s$};\\
    ub\_w_{out}(u)\cdot ub\_w_{in}(v), & \text{if $e(u,v)\notin E_s$}.\\
  \end{cases}
\label{eq:ub-ISF}
\end{eqnarray}

Clearly, in Eq.~(\ref{eq:ub-ISF}), we consider only one direct neighbor of vertex $u$ to compute the upper bound influence score. However, \cite{chen2015online} and \cite{ahmedTopicBasedCom} confirmed that the proposed techniques are practical and efficient.

\subsection{Structural Cohesiveness Pruning}
The $(k,d)\mbox{-}truss$ (as given in Definition \ref{def:kdtruss}) has a cohesiveness constraint that indicates that $\forall e \in \{E(H)|H\subseteq G_s\}$ must be contained in at least $(k-2)$ triangles denoted $sup(e)$. However, despite only having to compute $sup(e)$ for $e \in E(H)$, it is still inefficient for online queries. Thus, we computed $sup(e)$ offline for all $e\in E_s$. Then, based on $sub(e)$, we calculate the upper bound $ub\_sup(.)$ for each user $u \in G_s$.

The pruning users from $G_s$ based on structural cohesiveness pruning is provided in the following lemma.

\begin{lemma} {\bf (Structural Cohesiveness Pruning).}  Given a social network $G_s$,  and an integer $k$, 
for any $u\in G_s$, we can safely prune $u$ if $ub\_sup(u)<k-2$. 
\label{lemma:structural-cohesiveness-pruning}
\end{lemma}
\begin{proof} The proof is provided in the Appendix \ref{appendix:sec:proofs}.
\end{proof}

\noindent {\bf Discussions on How to Compute $ub\_sup(u)$:} For each edge $e\in E_s$, we computed $sup(e)$. Then, for each user $u$, we computed the maximum $sub(e)$ from its in-neighbors, $\mathcal{N}_{in}(u)$, and out-neighbors, $\mathcal{N}_{out}(u)$, as follows:

\begin{eqnarray}
    &&ub\_sup(u)\\
    &=&\max\bigg\{\max_{\forall v \in \mathcal{N}_{out}(u)} sup(e(u,v)), \max_{\forall v \in \mathcal{N}_{in}(u)} sup(e(v,u))\bigg\}.\notag
\label{eq:ubsup}
\end{eqnarray}

\subsection{Social-Distance-Based Pruning}
\label{sec:social-distance-based-pruning}

In addition to the cohesiveness constraint, the Definition \ref{def:kdtruss} requires the distance constraint among users in $C\subseteq G_s$, which is the number of hops between any user $u \in C$ and the query user $q$. For a query user $q$ and any user $u \in C$, if $dist_s(u,q)>d$, then we can safely prune $u$. However, computing the distance in $G_s$, online, between $q$ and all users in $G_s$ could be time-consuming, especially in large social networks. Thus, we precompute (offline) the distance between $u$ and a set of pivot points in $G_s$, denoted $\mathbb{P}_s$. Then, for an online query, we easily obtain the lower bound distance $lb\_dist_s(.)$ between two users utilizing the triangle inequality. 

For a query user $q$ and any user $u \in C$, if $lb\_dist_s(u,q)>d$, then we can safely prune $u$. Then, we formally provide the following lemma.

\begin{lemma} {\bf (Social-Distance-Based Pruning).}   Given a social network $G_s$, a query user $q$, and a social distance threshold $d$, for any $u \in V_s$, we can safely prune $u$ if it holds that $lb\_dist_s(u,q)>d$. 

\label{lemma:social-distance-based-pruning}
\end{lemma}
\begin{proof}
Derived from the Definition \ref{def:kdtruss}.
\end{proof}

\noindent {\bf Discussions on How to Compute $lb\_dist_s(.)$}: We precompute (offline) the distance between all users $u \in G_s$ and all $spv_i \in \mathbb{P}_s$, where $\mathbb{P}_s$= $\{spv_1, \dots, spv_\mathfrak{a}\}$. Then, for an online query, we use the triangle inequality to calculate the lower bound distance $lb\_dist_s(.)$ between two users. 

For any $u\in V_s$ and $q$, the $dist_s(u,q)\geq$ $|dist_s(u,spv_i) - $ $dist_s(q,\\ spv_i)|$. Then, we can compute the lower bound distance between $u$ and $q$, as follows:
\begin{eqnarray}
lb\_dist_s(u,q)= \max_{\forall spv_i \in \mathbb{P}_s} \{ |dist_s(u,spv_i) -  dist_s(q,spv_i)|\}.
\label{eq:lb-dist}
\end{eqnarray}

\noindent {\bf The Social Network Pivots $\mathbb{P}_{s}$ Selection.}
To enhance our pruning, we must identify a set of pivot users within $G_s$ based on the lower bound distance between any user in $G_s$ and a query user $q$. We have created a cost model $\mathbb{P}_{s}\mbox{\_}Cost$, as follows:
\begin{eqnarray}
 &&\mathbb{P}_{s}\mbox{\_}Cost \\
 &=&\sum_{\forall u\in G_s}\sum_{\forall v\in G_s}\max_{\forall spv_i \in \mathbb{P}_s} \{ |dist_s(u,spv_i) -  dist_s(v,spv_i)|\}.\notag
\label{eq:ps-cost}
\end{eqnarray}

When dealing with online queries, we typically compare and prune any user $u \in C$ against the query user $q$. However, we calculate our cost model against all users in the social network (i.e. any user query) and choose the set with the minimum cost.

\subsection{Spatial-Distance-Based Pruning}
\label{sec:spatial-distance-based-pruning}

Calculating the average distance (as given in Eq.~(\ref{eq:avgdist})) for online queries can be very costly. For a user $u$ with $u.L$, we have to compute $avg\_dist_r(.)$ between $u$ and $\forall p\in V_p'$, which requires $O(|u.L|\cdot|V_p'|)$ complexity. In order to answer online queries faster, we precompute (offline) the distance between $\forall p \in V_p$ to a set of pivot points in $G_r$, denoted $\mathbb{P}_r$. Then, we utilize the triangle inequality for an online query to calculate the lower bound average distance $lb\_avg\_dist_r(.)$ more efficiently. 

For a user $u$ and $\forall p \in \{q$'s checked-in locations $|p.K \cap Q \neq \emptyset \}$, if $lb\_avg\_dist_r(u, p)>\sigma$, then we can safely prune $u$. For this purpose, we formally provide the following lemma.

\begin{lemma} {\bf (Spatial-Distance-Based Pruning).} 
 Given a social network $G_s$, a spatial road network $G_r$, a query user $q$, a keyword query set $Q$, and a spatial distance threshold $\sigma$, 
we can safely prune any user $u \in V_s$, if $\forall p \in \{q$'s checked-in locations $|p.K \cap Q \neq \emptyset \}$ the $lb\_avg\_dist_r(u,p)>\sigma$ holds. 


\label{lemma:spatial-distance-based-pruning}
\end{lemma}
\begin{proof}
Derived from the Definition \ref{def:kcsbssn}.
\end{proof}

\noindent {\bf Discussions on How to Compute $lb\_avg\_dist_r(.)$:} We precompute (offline) the distance between all POIs $p \in V_p$ and all $rpv_i \in \mathbb{P}_r$ in the road network $G_r$, where $\mathbb{P}_r$= $\{rpv_1, \dots, rpv_\mathfrak{b}\}$. Then, for an online query, we use the triangle inequality to calculate the lower bound average distance $lb\_avg\_dist_r(.)$ between a user $u$ and any POI $p$ that the query user $q$ visited and satisfies the constraint $p.K \cap Q \neq \emptyset$. 

For any $u \in V_s$, and any $p \in \{q$'s checked-in locations $|p.K \cap Q \neq \emptyset \}$, the $dist_r(u.L_i,p.\ell)\geq$ $|dist_r(u.L_i,rpv_i) - $ $dist_r(p.\ell,rpv_i)|$. Then, we can compute the average distance lower bound, as follows:
\begin{eqnarray}
&&lb\_avg\_dist_r(u,p)\\
&=&\frac{\sum_{j=1}^{|u.L|}  \max_{\forall rpv_i \in \mathbb{P}_r}\{ dist_r(u.L_j,rpv_i) -dist_r(p,rpv_i)|\}}{|u.L|}.\notag
\label{eq:lb-avgdist}
\end{eqnarray}

\noindent {\bf The Road Network Pivots $\mathbb{P}_{r}$ Selection.}
To choose a set of pivot points on the road network, we must first calculate a cost model that considers the user's checked-in location and the location of the POIs on the road network. To maximize the pruning effectiveness of the  $lb\_avg\_dist_r(u,p)$, our chosen set of $\mathbb{P}_{r}$ must have the maximum average distance between each $u\in G_s$ and $p\in V_p$ on one side and the $rpv_i \in \mathbb{P}_r$ on the other. We then calculate the cost model to select $\mathbb{P}_{r}$, as follows:

\begin{eqnarray}
 &&\hspace{-2ex}\mathbb{P}_{r}\mbox{\_}Cost \\
 &\hspace{-4ex}=&\hspace{-1ex}\sum_{\forall u\in G_s}\hspace{-1ex}\sum_{\forall p\in V_p} \hspace{-2ex}\frac{\sum_{j=1}^{|u.L|}  \max_{\forall rpv_i \in \mathbb{P}_r}\{ dist_r(u.L_j,rpv_i) -dist_r(p,rpv_i)|\}}{|u.L|}\notag
\label{eq:pr-cost}
\end{eqnarray}

The lowest value of $\mathbb{P}_{r}\mbox{\_}Cost$ will enhance the pruning power of techniques in Spatial-Distance-Based Pruning. 

\section{Indexing Mechanism}
\label{sec:indexing-mechnism}

In the following subsections, we describe the indexing mechanism for $KCS\mbox{-}BSSN$ queries and the corresponding pruning strategies applied to index nodes. 

\subsection{Index Structure}
\label{subsec:indexing_cons}
To enable efficient $KCS\mbox{-}BSSN$ query processing, we construct a tree index $\mathcal{I}$ over the bipartite spatial-social network. The social network is first partitioned into subgraphs, each stored in a leaf node. These leaf nodes are then recursively grouped into intermediate nodes until a single root node is formed. The index tree consists of two types of nodes:

\noindent{\bf Leaf Node.}
Each leaf node $N$ contains a set of users from the social network. For each user $u \in N$, we maintain:

\begin{itemize}[noitemsep,topsep=0pt,leftmargin=10pt]
\item a keyword set $u.K={key_1,\dots,key_{|K|}}$ with aggregate statistics $f_{sum}$ and $f_{max}$. Let $P_u$ denote the set of POIs visited by $u$. The keyword aggregates are computed as:
\begin{eqnarray}
key_j.f_{sum}=\sum_{\forall p\in P_u, key_j\in p.K} f_{u,p}
\label{eq:key-f-sum}
\end{eqnarray}
\begin{eqnarray}
key_j.f_{max}=\max_{\forall p\in P_u, key_j\in p.K} f_{u,p}
\label{eq:key-f-max}
\end{eqnarray}

\item  an upper bound edge support $ub\_sup(u)$, where $ub\_sup(u)$= $\max \{sup(e)\}$. 
 \item an In\mbox{-}influence upper bound $ub\_w_{in}(u)$ and an Out\mbox{-}influence upper bound $ub\_w_{out}(u)$.
\item a vector of social distance to each social pivot $spv_i \in \mathbb{P}_s$, \\ $\{dist_s(u, spv_1),$ $ dist_s(u, spv_2),$ $ \dots, dist_s(u, spv_\mathfrak{b})\}$, where $\mathbb{P}_s$ is a set of pivot points in social network.
\end{itemize}

\noindent{\bf Non-Leaf Node.}
Each non-leaf node $N$ contains a set of child nodes $N_i$. For each such node, we maintain aggregated upper-bound information:

\begin{itemize}[noitemsep,topsep=0pt,leftmargin=10pt]
\item a keyword set $N.K={key_1,\dots,key_{|K|}}$ with upper bounds:
\begin{eqnarray}
key_j.ub\_f_{sum}= max_{\forall u\in N}  \{ key_{j}.f_{sum}\}
\label{eq:key-ub-f-sum}
\end{eqnarray}
\begin{eqnarray}
key_{j}.ub\_f_{max}= max_{\forall u\in N} \{key_{j}.f_{max}\}
\label{eq:key-ub-f-max}
\end{eqnarray}

\item an upper bound edge support:
 \begin{eqnarray}
ub\_sup(N)= max_{\forall u\in N}\{ub\_sup(u)\}
\label{eq:ub-sub-N}
\end{eqnarray}

\item an upper bound in-influence:
  \begin{eqnarray}
ub\_w_{in}(N)=\{ max_{\forall u\in N} ub\_w_{in}(u)\}
\label{eq:ub-w-in-N}
\end{eqnarray}

\item a vector of minimum and maximum social distances between node $N$ and each pivot $spv_i \in \mathbb{P}s$:
\begin{eqnarray}
mindist_s(N,spv_i)= \min_{\forall u \in N}\{dist_s(u,spv_i)\}
\label{eq:mindist-s-N-spv}
\end{eqnarray}
\begin{eqnarray}
maxdist_s(N,spv_i)= \max_{\forall u \in N}\{dist_s(u,spv_i)\}
\label{eq:maxdist-s-N-spv}
\end{eqnarray}
\end{itemize}

\subsection{Index-Level Pruning}
\label{subsec:indexing_level_pruning}
Our pruning techniques provided in Section \ref{sec:pruning} are for pruning users and points of interest. However, applying those techniques directly to large-scale databases is very costly. Since all users in $\mathcal{I}$ are bounded by MBRs, we can use this property by pruning the entire MBR. In the following subsection, we will provide pruning techniques at the index level to prune a set of false alarm users.

\subsubsection{Keyword-based Pruning for Index Nodes} As discussed in Section \ref{subsec:indexing_cons}, every node $N$ is associated with a set $N.K$, then we can prune any $N$ if $N.K\cap Q=\emptyset$. In the following lemma, we formally provide Keyword-based Pruning for index nodes.

\begin{lemma} {\bf (Keyword-based Pruning for Index Nodes).}  Given an index node $N$ and a keyword query set $Q$, the $N$ can be safely pruned, if $N.K\cap Q=\emptyset$.    
\label{lemma:keyword-based-pruning-index-level}
\end{lemma}

In Lemma \ref{lemma:keyword-based-pruning-index-level}, if $N.K\cap Q=\emptyset$ indicates that for any $u \in N$, the POI visited by $u$ do not have any key in $Q$. Then, node $N$ can be safely pruned.

\subsubsection{$\omega$-based Pruning for Index Nodes} 
Since for any $key_j\in N.K$, the $key_j.ub\_f_{sum}$ is the upper bound $f_{sum}$ for $key_j$ (as given in Eq.~(\ref{eq:key-ub-f-sum})). Then, we can directly prune any node $N$ if it holds $\max_{\forall key_j\in N.K|key_j\in Q} key_j.ub\_f_{sum}< \omega$.

The pruning nodes of $\mathcal{I}$ based on $\omega$-based pruning are provided in the following lemma.

\begin{lemma} {\bf ($\omega$-based Pruning for Index $\mathcal{I}$ Nodes).} Given an index node $N$, a keyword query set $Q$, and thresholds $\omega$, the $N$ can be safely pruned, if $\max_{\forall key_j\in N.K|key_j\in Q}key_j.ub\_f_{sum}< \omega$.
\label{lemma:omega-based-pruning-index-level}
\end{lemma}

In Lemma \ref{lemma:omega-based-pruning-index-level}, if $\max_{\forall key_j\in N.K|key_j\in Q}key_j.ub\_f_{sum}< \omega$ indicates that for any $u \in N$ the $\max_{\forall key_j\in u.K|key_j\in Q}$ $key_j.ub\_f_{sum}< \omega$. Then, node $N$ can be safely pruned.

\subsubsection{$\pi$-based Pruning for Index Nodes} As given in Eq.~(\ref{eq:key-ub-f-max}), for each $key_j\in N.K$, the $key_j.ub\_f_{max}$ is the upper bound of $f_{max}$ for all $u\in N$. Then, for this pruning, we can filter out any $N \in \mathcal{I}$ if $\max_{\forall key_j\in N.K|key_j\in Q } key_j.ub\_f_{max}< \pi$. Then, we formally provide the following lemma. 

\begin{lemma} {\bf ($\pi$-based Pruning for Index $\mathcal{I}$ Nodes).} Given an index node $N$, a keyword query set $Q$, and thresholds $\pi$, node $N$ can be safely pruned, if $\max_{\forall key_j\in N.K|key_j\in Q } key_j.ub\_f_{max}< \pi$.
\label{lemma:pi-based-pruning-index-level}
\end{lemma}

In Lemma \ref{lemma:pi-based-pruning-index-level}, if $\max_{\forall key_j\in N.K|key_j\in Q } key_j.ub\_f_{max}< \pi$, then for all $u \in N$ the  $\max_{\forall key_j\in u.K|key_j\in Q }$ $ key_j.ub\_f_{max}< \pi$. In this case, we can safely prune the node $N$.

\subsubsection{Influence-based Pruning for Index Nodes}
 Since we calculated for each $N\in \mathcal{I}$ the $ub\_w_{in}(N)$ (as given in Eqs.~(\ref{eq:ub-w-in-N})), we can easily compute the $ISF( q\leadsto N)$ between a query user $q$ and node $N$. Then, for any node $N$ if $ISF ( q\leadsto N)<\theta$, we can safely prune $N$.
 We estimate the upper bound influence score between a query user $q$ and node $N$ as follows:
 \begin{eqnarray}
 ub\_ISF(q,N)=ub\_w_{out}(q)\cdot ub\_w_{in}(N).
\label{eq:ub1-ISF_indel-level}
\end{eqnarray}

\begin{lemma} {\bf  (Influence-based Pruning for Index Nodes).}  Given an index node $N$ and a threshold $\theta$, for any $q$, we can prune safely $N$ if $ub\_ISF(q,N)<\theta$. 
\label{lemma:Influence-based-pruning-index-level}
\end{lemma}

In Lemma \ref{lemma:Influence-based-pruning-index-level}, if $ub\_ISF(q,N)<\theta$ confirms that for any $u\in N$ the $ub\_ISF(q,u)<\theta$, then node $N$ can be safely pruned.

\subsubsection{Structural Cohesiveness Pruning for Index Nodes}
For any node $N$, we can safely remove $N$ if for all its child users $u$ hold $ub\_sup(u)< k-2$. Since $ub\_sup(N)$ is the upper bound edge support for all $u\in N$, as stated in Eq.~(\ref{eq:ub-sub-N}). Then, we formally provide the following lemma.

\begin{lemma} {\bf (Structural Cohesiveness Pruning for Index Nodes).} Given an index node $N$ and an integer $k$, for any $N\in \mathcal{I}$, we can safely prune $N$ if $ub\_sup(N)< k-2$. 
\label{lemma:structural-cohesiveness-pruning-index-level}
\end{lemma}

In Lemma \ref{lemma:structural-cohesiveness-pruning-index-level}, the upper bound $ub\_sup(N)<k-2$ indicates that for any $u\in N$ the $ub\_sup(u)\leq ub\_sup(N)$. Then node $N$ can be safely filtered out.

\subsubsection{Social-distance-based Pruning for Index Nodes} 
\label{sec:social-distance-based-pruning-index-nodes}
In order to take advantage of the distance constraint between users in $C\subseteq G_s$ (as given in Definition \ref{def:kdtruss}), we can filter out a node $N\in \mathcal{I}$ if $N$ does not satisfy the distance constraint. Specifically, for a given query user $q$ if $dist_s(q,N)>d$, we can safely prune node $N$. Since for each $N\in \mathcal{I}$, we store the minimum and maximum social distance between $N$ and every $spv_i\in \mathbb{P}_s$. Then, we can easily compute the $lb\_dist_s(q,N)$ utilizing the triangle inequality. 

For a query user $q$ and any node $N \in \mathcal{I}$, if $lb\_dist_s(q,N)>d$, then we can safely prune $N$. Then, we formally provide the following lemma.

\begin{lemma} {\bf (Social-distance-based Pruning for Index Nodes).}   Given an index $\mathcal{I}$, a query user $q$, and a social distance threshold $d$. For any $N\in \mathcal{I}$, we can safely prune $N$ if $lb\_dist_s(q,N)>d$. 

\label{lemma:social-distance-based-pruning-index-level}
\end{lemma}

\noindent {\bf Discussions on How to Compute $lb\_dist_s(.)$}: As discussed in Section \ref{subsec:indexing_cons}, we calculate the minimum and maximum social distance between $N$ and all $spv_i\in \mathbb{P}_s$. Then, for a given $q$ and any $N\in \mathcal{I}$, the $dist_s(q,N)\geq$ $|dist_s(q,spv_i) - dist_s(N,spv_i)|$. We can compute the lower bound distance between $q$ and $N$ as follows:
\begin{eqnarray}
&&lb\_dist_s(q,N) \\
&=&\min 
\begin{cases}
&\max_{\forall spv_i \in \mathbb{P}_s} \{ |dist_s(q,spv_i) -mindist_s(N,spv_i)|\}\\
&\max_{\forall spv_i \in \mathbb{P}_s} \{ |dist_s(q,spv_i) - maxdist_s(N,spv_i)|\}\\
&\textit{0, \quad if $q\in N$}
\end{cases},\notag
\label{eq:lb-dist-indexing}
\end{eqnarray}
where $mindist_s(.)$ and $maxdist_s(.)$ are given in Eqs.~ (\ref{eq:mindist-s-N-spv}) and (\ref{eq:maxdist-s-N-spv}), respectively.

\subsection{The Index Construction}
\label{sec:index_construction} 
To optimize query performance, our index construction ensures that each leaf node encapsulates a subgraph that strongly reflects prevalent $KCS\mbox{-}BSSN$ community characteristics. Such localization increases the likelihood that a query result resides within a small number of leaf nodes, thereby reducing traversal overhead in the index tree $\mathcal{I}$. Algorithm~\ref{alg:pivot-index-refinement} refines the set of pivot index users $\mathbb{P}_{index}$ using the cost model described in Section~\ref{sec:index_pivot_selection}. We first compute the cost of an initial pivot set $\mathbb{P}_{index}$. During each iteration, a pivot user $piv_i \in \mathbb{P}_{index}$ is swapped with a candidate user $u \in G_s$ ($u \neq piv_i$), and the new cost is evaluated. After the refinement process, the pivot set with the minimum cost is selected. Algorithm~\ref{alg:partition-social-network} then partitions the social network $G_s$ using the optimized pivot set $\mathbb{P}_{index}$. Each user $v$ is assigned to the pivot $piv_i \in \mathbb{P}_{index}$ that maximizes a quality function defined in Section~\ref{sec:index_pivot_selection}. The resulting subgraphs form the leaf nodes of the index tree. These leaf nodes are recursively grouped into intermediate nodes, where a cost model determines the most suitable parent-child relationships. For each intermediate node, a subset $\mathbb{P'}_{index} \subset \mathbb{P}_{index}$ is selected following a similar pivot selection strategy.

More specifically, \textbf{Algorithm~\ref{alg:partition-social-network}: Partition\_Social\_Network} takes $G_s$ and $\mathbb{P}_{index}$ as input. For each user $u$, the algorithm evaluates all pivots and assigns $u$ to the pivot that yields the highest quality score (lines 1–8). This produces a set of disjoint subgraphs, which serve as the foundation for index construction. We first initialize the required local variables (line 2) and then evaluate each $piv_i \in \mathbb{P}_{index}$ to identify the pivot that provides the highest quality score for user $u$ (lines 3–7). In line 8, user $u$ is assigned to the corresponding subgraph. The algorithm ultimately outputs the resulting set of subgraphs.

In \textbf{Algorithm~\ref{alg:pivot-index-refinement}: Pivot\_Index\_Refinement}, the inputs include the social network $G_s$, the spatial network $G_r$, and a parameter $threshold\mbox{-}iter$ specifying the maximum number of refinement iterations. Initially, $\mathfrak{s}$ pivot users are randomly selected (line 1), and corresponding subgraphs ${S_1,\dots,S_\mathfrak{s}}$ are generated using Algorithm~\ref{alg:partition-social-network}. For up to $threshold\mbox{-}iter$ iterations, one pivot $piv_i$ is randomly replaced with a candidate $temp\mbox{-}piv$ (lines 4–5). The network is repartitioned using the updated pivot set (line 7), and the new cost is computed using Eq.~(\ref{eq:pindexcost}). If the new configuration achieves a lower cost, it replaces the current pivot set (lines 8–9). After convergence or completion of iterations, the final pivot set $\mathbb{P}_{index}$ and its corresponding subgraphs are returned.

\begin{algorithm}[t!]\small
\KwIn{a social network $G_s$, a spatial network $G_r$, $threshold\mbox{-}iter$}
\KwOut{ a set of pivot index $\mathbb{P}_{index}$ }
$\mathbb{P}_{index}$= Randomly select $\mathfrak{s}$ of $\mathbb{P}_{index}$\\
$Partition\_Social\_Network(G_s,\mathbb{P}_{index}$) \\
 \While {$iter < threshold\mbox{-}iter$}{
    Randomly select a new pivot user $temp\mbox{-}piv\in G_s$\\
    Randomly select a pivot user $piv_i\in \mathbb{P}_{index}$\\
    $temp\mbox{-}\mathbb{P}_{index}=\mathbb{P}_{index}-\{piv_i\}+\{temp\mbox{-}piv\}$\\
    $Partition\_Social\_Network(G_s,temp\mbox{-}\mathbb{P}_{index}$) \\
    \If{$temp\mbox{-}\mathbb{P}_{index}\_cost < \mathbb{P}_{index}\_cost$}{
         $\mathbb{P}_{index}= temp\mbox{-}\mathbb{P}_{index}$
    }
 }
 $Partition\_Social\_Network(G_s,\mathbb{P}_{index}$) \\
\Return \{$\mathbb{P}_{index}$\}\newline
\caption{Pivot\_Index\_Refinement}
\label{alg:pivot-index-refinement}\vspace{-1ex}
\end{algorithm}

\begin{algorithm}[t!]\small
\KwIn{a social network $G_s$, a set of pivot index $\mathbb{P}_{index}$}
\KwOut{a set of subgraphs $S_1,\dots, S_\mathfrak{s}$ }
 \For{\textbf{each} user $u \in G_s$}{
 $i$=1; $best\_quality=0$; $j=i$\\
 \While {$i \leq \mathfrak{s}$}{
    \If{$quality(u,piv_i) > best\_quality$}{
        $best\_quality=quality(u,piv_i)$\\
        $j=i$\\
    }
    $i=i+1$;\\
  }
  $S_j=S_j+\{u\}$\\
  }
\Return $\{S_1,\dots, S_\mathfrak{s}\}$\newline
\caption{Partition\_Social\_Network}
\label{alg:partition-social-network}\vspace{-1ex}
\end{algorithm}

\section{\texorpdfstring{$KCS\mbox{-}BSSN$}{KCS\mbox{-}BSSN} Query Answering}
\label{sec:query_answering}
In this section, we present\textbf{ Algorithm \ref{alg:kcs-bssn-query-answer}, the $KCS\mbox{-}BSSN$ Query Answering Algorithm}, designed for the efficient retrieval of community results. The algorithm adopts a "filter-and-refine" framework, partitioned into a Pruning Phase and a Refinement Phase. It takes as input the social network $G_s$, spatial network $G_r$, weighted bipartite network $G_b$, keyword query set $Q$, structural thresholds $k$ and $d$, keyword/influence thresholds $\omega$, $\pi$, and $\theta$, a spatial distance threshold $\sigma$, and the query user $q$.

\noindent \textbf{Pruning Phase:} The algorithm begins by initializing $POI_q$, the set of $q$'s checked-in locations that contain at least one keyword from $Q$ (line 1). To facilitate an efficient search, two priority queues are initialized: a max-heap $\mathcal{H}$ for index tree traversal and a min-heap $\mathcal{H}_{cand}$ for candidate management. $\mathcal{H}$ stores entries $(N, \textit{heap-key})$, where $N$ is an index node and $\textit{heap-key}$ represents the upper bound edge support, $ub\_sup(N)$ (line 2). $\mathcal{H}_{cand}$ stores potential candidate users $u \in G_s$ prioritized by the lower bound average road distance, $lb\_avg\_dist_r(u,p)$, (line 3). The traversal begins by pushing the root of the index tree $\mathcal{I}$ into $\mathcal{H}$ (line 4). While $\mathcal{H}$ is not empty, the top node $N$ is extracted (line 6). If its $\textit{heap-key}$ falls below $k-2$, the search terminates early as no remaining nodes can satisfy the structural cohesiveness requirements (line 7). If $N$ is a leaf node, we iterate through each user $u \in N$ and apply a suite of user-level pruning filters, including keyword, $\omega$, $\pi$, influence, structural, social-distance, and spatial-distance pruning. Users who survive these filters are inserted into $\mathcal{H}_{cand}$ (lines 9-12). If $N$ is a non-leaf node, we apply index-level pruning to its children. Any child node $N_i$ that cannot be pruned is inserted into $\mathcal{H}$ with its corresponding upper bound support (lines 15-17).

\noindent \textbf{Refinement Phase:} The refinement phase (line 18) processes the candidate set stored in $\mathcal{H}_{cand}$. We first perform a threshold check: any candidate with a $lb\_avg\_dist_r$ exceeding $\sigma$ is discarded, along with all remaining entries in the min-heap. For the remaining candidates, we verify the connected subgraphs against the full suite of $KCS\mbox{-}BSSN$ requirements as defined in Definition \ref{def:kcsbssn}. Finally, the algorithm returns the exact communities that satisfy all constraints.

\begin{algorithm}[t!]
\small
\KwIn{a social network $G_s$, a spatial network $G_r$, a weighted bipartite network  $G_b$, a keyword query set $Q$, query thresholds $k$,  $d$, $\omega$, and $\pi$, a spatial distance threshold $\sigma$, an influence score threshold $\theta$, and a query user $q$.}
\KwOut{ a community C, satisfying  $KCS\mbox{-}BSSN$ (as given in Definition \ref{def:kcsbssn})}

set $POI_q=\{q$'s checked-in locations $|p.K \cap Q \neq \emptyset \}$\\
initialize a max-heap $\mathcal{H}$ accepting entries in the form ($N$, $heap\mbox{-}key$)\\
initialize a min-heap $\mathcal{H}_{cand}$ accepting entries in the form ($u$, $heap\mbox{-}key_{cand}$)\\
insert entry ($root(\mathcal{I}),0$) into heap $\mathcal{H}$\\
\While{$\mathcal{H}$ is not empty}{
   ($N$, \textit{heap\mbox{-}key}) = de-heap $\mathcal{H}$\\
   \If{\textit{heap\mbox{-}key }$< k-2$,}{ break and terminate.}   
    \If{$N$ is a leaf node}{
         \For{\textbf{each} user $u \in N$}{
            \If{$u$ cannot be pruned by Lemma 
             \ref{lemma:keyword-based-pruning-u},
             \ref{lemma:omega-based-pruning}, 
             \ref{lemma:pi-based-pruning}
             \ref{lemma:Influence-based-pruning}, 
             \ref{lemma:structural-cohesiveness-pruning}, \ref{lemma:social-distance-based-pruning}, and             \ref{lemma:spatial-distance-based-pruning},
w.r.t $q$ }{  
                insert \big($u,min_{\forall p \in POI_q}(lb\_avg\_dist_r(u,p))$ \big) into the heap $\mathcal{H}_{cand}$\\
            }
        }
        }
    \Else {
        // \textit{$N$ is a non-leaf node}\\
        \For{\textbf{each} entry $N_i \in N$}{
             \If{$N_i$ cannot be pruned by Lemma 
             \ref{lemma:keyword-based-pruning-index-level}, 
             \ref{lemma:omega-based-pruning-index-level}, 
             \ref{lemma:pi-based-pruning-index-level}, \ref{lemma:Influence-based-pruning-index-level}, \ref{lemma:structural-cohesiveness-pruning-index-level}, and 
            \ref{lemma:social-distance-based-pruning-index-level} }  {
                insert ($N_i,ub\_sup(N_i)$ ) into the heap $\mathcal{H}$
            }
        }
        }
   }
   
$C$= Refinement ($\mathcal{H}_{cand}$) \textit{//refinement phase}\\
\Return $C$\newline
\caption{ $KCS\mbox{-}BSSN$\_Query\_Answer}
\label{alg:kcs-bssn-query-answer}\vspace{-1ex}
\end{algorithm}

\subsection{The Index Pivots Selection for Leaf Nodes}
\label{sec:index_pivot_selection}

To partition the social network into manageable subgraphs, we select $\mathfrak{s}$ specific users to serve as index pivots ($\mathbb{P}_{index}$). We utilize a cost model to identify the most suitable pivots by evaluating the social network across three dimensions: \textit{Bipartite Structure}, \textit{Social Structure}, and \textit{Spatial Structure}. Each dimension is represented by a specific scoring function, which is integrated into the total cost calculation.\\
\noindent\textbf{Total Cost.}The cost associated with selecting index pivots $\mathbb{P}_{index}$ for all subgraphs $S \in G_s$ is defined as follows:

\begin{eqnarray}\label{eq:pindexcost}
&&\mathbb{P}_{index}\_cost \\
&=&W_{bs} \cdot (1- \sum_{\forall S\in G_s}\sum_{\forall u\in S}\sum_{\forall v\in S}  bs\_score(u,v))\notag\\
&&+W_{rs} \cdot \sum_{\forall S\in G_s}\sum_{\forall u\in S}\sum_{\forall v\in S}  rs\_score(u,v) \notag\\
&&+W_{ss} \cdot(1- \sum_{\forall S\in G_s}\sum_{\forall u\in S}\sum_{\forall v\in S}  ss\_score(u,v))\notag
\end{eqnarray}
where $bs\_score$, $rs\_score$, and $ss\_score$ represent the bipartite, road-network (spatial), and social structure scores, respectively (detailed in the Appendix \ref{appendix:sec:index_pivot_selection}). The coefficients $W_{bs}, W_{ss},$ and $W_{rs}$ are user-defined weights.

\noindent \textbf{The Quality of Selecting Index Pivots $\mathbb{P}_{index}$.}
To assign a user $u$ to a specific pivot $piv_i \in \mathbb{P}_{index}$, we evaluate their assignment quality based on the three structural scores:
\begin{eqnarray}\label{eq:total-quality}
&&\hspace{-7ex} quality(u,piv_i)\\
&=& W_{bs} \cdot bs\_score(u,piv_i) +W_{ss} \cdot ss\_score(u,piv_i)\notag\\
&&+ (1- (W_{rs} \cdot rs\_score(u,piv_i))\notag
\end{eqnarray}

This quality metric ensures that each user $u$ is mapped to the ideal pivot $piv_i$ that maximizes the structural and spatial coherence of the resulting leaf node.

\subsection{The Index Pivots Selection for Non-Leaf Nodes}
\label{sec:tree_construction}
Constructing the intermediate levels of the index tree requires a distinct set of pivots, $\mathbb{P'}_{index}$. For non-leaf nodes, the cost model prioritizes the \textit{Bipartite} and \textit{Social} structures to maintain hierarchical cohesiveness.\\
\noindent\textbf{Tree Cost.} The index tree cost is calculated by evaluating the structural affinity of intermediate nodes as follows:

\begin{eqnarray}
&&Tree\_cost\\
&\hspace{-2ex}=&W_{bs} \cdot(1-  \sum_{\forall N\in Nodes}\sum_{\forall piv'_i\in \mathbb{P'}_{index}} bs\_score\_node(N,piv'_i))  \notag\\
&&\hspace{-2ex} + W_{ss} \cdot(1- \sum_{\forall N\in Nodes}\sum_{\forall piv'_i\in \mathbb{P'}_{index}} ss\_score\_node(N,piv'_i))\notag
\label{eq:pindexcost-node}
\end{eqnarray}
The node-level scores $bs\_score\_node$ and $ss\_score\_node$ are defined in the  Appendix \ref{appendix:sec:tree_construction}.

\noindent\textbf{The Quality of Assigning Nodes to Parent Nodes}. To construct the hierarchy, each child node $N_i$ is assigned to a parent node $N$ by evaluating the quality of the node relative to the non-leaf pivots:
\begin{eqnarray}
&&\hspace{-4ex} quality\_node(N,piv'_i)\\
&\hspace{-2ex}=& \hspace{-2ex} W_{bs} \cdot bs\_score\_node(N,piv'_i) +  W_{ss} \cdot ss\_score\_node(N,piv'_i)\notag
\label{eq:total-quality-node}
\end{eqnarray}
The index tree $\mathcal{I}$ is built bottom-up; for each $piv'_i \in \mathbb{P'}_{index}$, a new node $N$ is generated, and child nodes $N_i$ are assigned to the parent node that yields the highest quality. This process iterates until the root node is established, resulting in a balanced, structurally-aware index.

\section{Keyword-Based Community Search Over Temporal Bipartite Spatial-Social Network \texorpdfstring{$KCS\text{-}TBSSN$}{KCS-TBSSN}} 
\label{sec:temporal_BSSN}

Real-world (bipartite) graphs usually evolve dynamically over time, such as with edge weight updates. In this work, we also extend the original static $BSSN$ graph to the data model of \textit{temporal bipartite spatial-social network} ($TBSSN$), by considering edge weights $f_{u,p}$ dynamically change over a temporal dimension. Specifically, in the original model (Definition \ref{def:Gb}), edge weight $f_{u,p}$ represents a static/fixed frequency of historical visits. In our temporal $TBSSN$ model, we re-define the relationships between users and POIs to account for the time-varying nature of user-POI interactions.

\vspace{2ex}
\subsection{Dynamic Edge Weight Updates Under the Sliding Window Model}

In the $TBSSN$ graph, each user $u \in V_s$ is associated with a 2D temporal visit vector, $\mathcal{T}_{u,p}$, where each element $(p,t')\in \mathcal{T}_{u,p}$ represents a visit from user $u$ to POI $p$ at a specific timestamp $t'$. We consider a \textit{sliding window} of size $\uptau$ for obtaining the edge weight $f_{u,p}$. In particular, given the current timestamp $t$, the frequency $f_{u,p}$ is calculated as the count of all visits occurring for the most recent $\uptau$ timestamps, that is, $$f_{u,p}(t) = |\{(p, t') \in \mathcal{T}_{u,p} \mid t-\uptau+1 \leq t'\leq t  \}|.$$

At a new timestamp $(t+1)$, new visits $(p, t+1)\in \mathcal{T}_{u,p}$ are added to $f_{u,p}(t)$, and expired visits $(p, t-\uptau+1)\in \mathcal{T}_{u,p}$ will be removed from $f_{u,p}(t)$, which result in an updated edge weight $f_{u,p}(t+1)$ for the new sliding window between timestamps $(t-\uptau+2)$ and $(t+1)$.




\subsection{Updates of Temporal \texorpdfstring{$TBSSN$}{TBSSN} Graph}
\label{sec:temporal_pruning}

To optimize query processing, we introduce a \textit{temporal pruning} strategy. This mechanism identifies and discards "stale" interactions before the core computation begins. As established in Definition \ref{def:awcore}, if $f_{u,p}$ is updated and equal to zero, the edge will be removed from the graph. This approach offers two primary advantages:

\begin{enumerate}
\item \textbf{Dynamic Validity:} By evicting invalid/expired visits at the start of the query, the resulting edge weights $f_{u,p}$ always represent the latest state of the network relative to the users' temporal constraints.
\item \textbf{Space Reduction:} By removing the frequency of visits before timestamp $(t - \uptau + 1)$ ($t$ is the current timestamp), the edge weights may drop to zero, and thus the edges no longer exist. This can reduce the space cost of the graph storage. For example, we would not need the visiting records ten years ago, as it may not reflect the current popularity of POIs by users.

\end{enumerate}

\subsection{Dynamic Updates of Pre-Computed Data and Index for \texorpdfstring{$KCS\mbox{-}TBSSN$}{KCS-TBSSN}}
\label{sec:dynamic_updates}

This subsection details the mechanisms for dynamically updating temporal $TBSSN$ graph with low computational overhead. To manage updates on user-POI interactions, we address two main scenarios: including the frequency of new visits in a user’s activity vector and removing the frequency of the expired visits that outside the sliding window (i.e., before timestamp $(t-\uptau+1)$). To maintain high throughput, we employ a batch updating strategy. That is, rather than processing each check-in individually, we aggregate multiple updates in a batch to synchronize/update the $TBSSN$ graph data. These updates can propagate changes to two critical components:  the validity of community answers and the structural integrity of the indexing tree $\mathcal{I}$ (provided in Section \ref{sec:indexing-mechnism}).

\subsubsection{Incremental Maintenance of Communities}
\label{sec:communities_maintenance}
For a user $u \in V_s$ and a set of active communities, we verify the validity of each community $C$ as follows:\\
\noindent \textbf{Insertion.} For a new visit from user $u$ to POI $p$:

\begin{itemize}[noitemsep,topsep=2pt,leftmargin=15pt]
    \item \textbf{Case 1 ($u \in C$):} If $u$ is already a member of community $C$, the update is recorded, and we simply increment $u$'s visit frequency. Structural cohesiveness remains unchanged.
    
    \item \textbf{Case 2 ($u \notin C$):} If $u \notin C$, we evaluate the potential inclusion of $u$ into the community $C \cup \{u\}$. We apply the multi-stage pruning (i.e., the refinement phase of Algorithm \ref{alg:kcs-bssn-query-answer}). If $C \cup \{u\}$ satisfies all constraints, $u$ is integrated; otherwise, the update is discarded for that specific $C$.

\end{itemize}\noindent \textbf{Deletion (Expiration).} When an old visit by user $u$ at timestamp $t'$ expires (i.e., $t' < (t - \uptau + 1)$, for current timestamp $t$):
\begin{itemize}[noitemsep,topsep=2pt,leftmargin=15pt]
\item \textbf{Case 1 ($u \in C$):} If $u \in C$, the reduction in frequency may violate the $(\omega, \pi)\mbox{-}keyword\mbox{-}core$ constraints. We trigger the refinement process for $C$. If $C$ fails any constraint (e.g., $f_{sum}(.) < \omega$), it is removed from the community answer set of our $KCS\mbox{-}BSSN$ problem.

\item \textbf{Case 2 ($u \notin C$):} The expiration has no impact on the community's validity and is ignored.
\end{itemize}

\subsubsection{Indexing Level Maintenance}
\label{sec:indexing_maintenance}
The index tree $\mathcal{I}$ is built on the relationships between users and POIs. The batch updates modify frequencies and associations, the cost model metrics—specifically the bipartite score ($bs\_score$), road-network score ($rs\_score$), and social score ($ss\_score$)—may change over time, where $bs\_score$, $rs\_score$, and $ss\_score$ are given in Eqs.~(\ref{eq:bsscores}), (\ref{eq:rsscore}) and (\ref{eq:ssscore}), respectively. 

To support efficient for query processing over temporal $TBSSN$ graph (with dynamic updates), we perform the maintenance of index tree $\mathcal{I}$ as follows:

\noindent \textbf{Boundary Constraint Validation.} Following each batch update, we perform a validation of the index structure for both leaf nodes and non-leaf nodes, as detailed in Section \ref{subsec:indexing_cons}. This process ensures the structural integrity of the hierarchy by verifying that all upper and lower bounds constraints are strictly maintained for every user and index nodes.

\noindent \textbf{Incremental Pivot Re-evaluation.} We monitor the cumulative change in $quality(u, piv_i)$ (given by Eq.~ (\ref{eq:total-quality})). For each user $u$ in a leaf node, if a visit pattern changes such that another pivot $piv_j \in \mathbb{P}_{index}$ offers a significantly higher quality score, $u$ is migrated to the subgraph of $piv_j$. A stability margin is used to avoid unnecessary updates.

\noindent \textbf{Structural Synchronization.} Leaf node updates propagate upward. For any node $N$, if the difference in $quality\_node(N, piv'_i)$ (Eq.~(\ref{eq:total-quality-node})) exceeds a specific margin relative to another $piv'_j$, we remap $N$ to the most qualified upper-level node. This ensures pruning properties remain robust against temporal drift.

\begin{algorithm}[h!]
\small
\KwIn{Batch updates $\Delta U$, $OP$ (insertions/deletions), threshold $\uptau$, offline-data, a set of active communities $\mathbb{C}$, an Index tree $\mathcal{I}$, and a stability margin $\delta$}
\KwOut{Updated $(\text{offline-data}, \mathbb{C}', \mathcal{I}')$}

\For{\textbf{each} update $(u, p, t') \in \Delta U$}{
\If{$OP = Insertion$}{
\For{\textbf{each} checkin\mbox{-}location p}
{$f_{u,p} \leftarrow f_{u,p} + |\{(p,t') \in \mathcal{T}_{u,p}|$}}
\ElseIf{$OP = Deletion$}{
\For{\textbf{each} checkin\mbox{-}location p}{
{$f_{u,p} \leftarrow  |\{(p,t') \in \mathcal{T}_{u,p} \mid t' \geq (t - \uptau + 1)\}|$}
}}
{Recompute user $u$ offline-data }\\
}
 $\mathbb{C}'= Communities\_Maintenance( \mathbb{C}$, $\Delta U$, $OP$) \\
 $\mathcal{I}'=  Indexing\_Tree\_Maintenance( \mathcal{I}$, $\Delta U$, offline-data, $\delta$ ) \\

\Return  (offline-data, $\mathbb{C}', \mathcal{I}'$)\\
\caption{Dynamic\_Maintenance}
\label{alg:dynamic-maintenance}
\end{algorithm}

\begin{algorithm}[h!]
\small
\KwIn{A set of active communities $\mathbb{C}$, batch updates $\Delta U$, $OP$ (insertion/deletion)}
\KwOut{Updated set of valid communities $\mathbb{C}'$}

\For{\textbf{each} $(u, p, t') \in \Delta U$}{
\For{\textbf{each} $C \in \mathbb{C}$}{
\If{$OP = Insertion$}{
\If{$u \notin C$}{
\If{$Refinement(C \cup \{u\})$ satisfies constraints}
{$C \leftarrow C \cup \{u\}$}}}
\ElseIf{$OP = Deletion$}{
\If{$u \in C$}{$C \leftarrow Refinement(C)$ \tcp*{Trigger pruning}
\If{$C = \emptyset$}{Remove $C$ from $\mathbb{C}$;}}}}}
\Return $\mathbb{C}'$\\
\caption{Communities\_Maintenance}
\label{alg:communities-maintenance}
\end{algorithm}

\begin{algorithm}[h!]
\small
\KwIn{Index tree $\mathcal{I}$, batch updates $\Delta U$, updated offline-data, stability margin $\delta$}
\KwOut{Synchronized Index tree $\mathcal{I}'$}

\For{\textbf{each} leaf-node $L \in \mathcal{I}$}{
\For{\textbf{each} user $u \in (L \cap \Delta U)$}{
{$piv_{curr} \leftarrow u.pivot$}\\
{$piv_{best} \leftarrow \text{argmax}_{piv_j \in \mathbb{P}_{index}} quality(u, piv_j)$}\\
\If{$quality(u, piv_{best}) > quality(u, piv_{curr}) + \delta$}{Migrate $u$ to leaf-node of $piv_{best}$}}
{Update boundary constraints (Upper/Lower bounds)}
}
\For{\textbf{each} non-leaf node $N \in \mathcal{I}$ (bottom-up)}{
{Update boundary constraints (Upper/Lower bounds)}\\
{$piv'_{curr} \leftarrow N.parent\_pivot$}\\
{$piv'_{best} \leftarrow \text{argmax}_{piv'_j} quality\_node(N, piv'_j)$}\\
\If{$quality\_node(N, piv'_{best}) > quality\_node(N, piv'_{curr}) + \delta$}{Remap $N$ to $piv'_{best}$ and propagate updates upward;}}
\Return $\mathcal{I}'$\\
\caption{Indexing\_Tree\_Maintenance}
\label{alg:indexing-maintenance}
\end{algorithm}


\vspace{3ex}
\subsection{Incremental Community Maintenance Upon Dynamic Updates}

In this subsection, we detail the algorithmic framework for managing temporal network changes. The process is governed by three primary procedures: \textbf{Algorithm~\ref{alg:dynamic-maintenance}} for high-level data synchronization, \textbf{Algorithm~\ref{alg:communities-maintenance}} for the community validity, and \textbf{Algorithm~\ref{alg:indexing-maintenance}} for structural index integrity. In \textbf{Algorithm~\ref{alg:dynamic-maintenance}: Dynamic\_Maintenance}, the procedure coordinates the high-level synchronization of the temporal network state. The process begins by iterating through each update $(u, p, t')$ in the batch $\Delta U$ (line 1). Depending on the operation type, it either increments the visit frequency $f_{u,p}$ for new insertions (lines 2–4) or recalculates it based on the temporal threshold $\uptau$ for deletion/expiration (lines 5–7). After updating these frequencies, the user’s offline-data are re-computed to reflect current visit patterns (line 8). The algorithm then triggers sub-routines for community maintenance (line 9) and indexing tree maintenance (line 10) before returning the synchronized dataset (line 11).

Furthermore, in \textbf{Algorithm~\ref{alg:communities-maintenance}: Communities\_Maintenance}, the procedure ensures that active communities $\mathbb{C}$ remain valid according to all constraints. For each update in the batch (line 1), it iterates through existing communities (line 2) and checks the operation type. For an Insertion (line 3), if a user $u$ is not already a member, it evaluates the community's validity including $u$ through the $Refinement$ process (lines 5-6). For a Deletion (line 7), if the user $u$ is a member, the $Refinement$ process is triggered to reevaluate the community based on the reduced frequency (line 9). If the community becomes empty during this process, it is removed from the active set (line 11). Finally, the set of survived communities $\mathbb{C}'$ is returned (line 12).

Finally, in \textbf{Algorithm~\ref{alg:indexing-maintenance}: Indexing\_Tree\_Maintenance}, the hierarchical index $I$ is updated to handle temporal drift and maintain search efficiency. The algorithm first performs Leaf-Level Migration: for every leaf node $L$ and affected user $u$ (lines 1–2), it compares the $quality$ of the current pivot versus all alternative pivots in $\mathbb{P}_{index}$ (lines 3-4). In lines 5-6, if an alternative pivot offers an improvement greater than the current leaf node with the stability margin $\delta$, the user is migrated. Boundary constraints for the leaf node are then updated (line 7). Finally, the algorithm performs a Bottom-Up synchronization (line 8). For each non-leaf node $N$, it updates boundary constraints and evaluates the $quality\_node$ score (line 9). If a more suitable parent pivot is found (exceeding $\delta$), the node is remapped, and updates are propagated upward to ensure the pruning properties of the index remain robust against temporal changes (lines 10–13). Then, in line 14, we return a synchronized Index tree $\mathcal{I}'$.

\section{Experimental Evaluation}
\label{sec:exper}

\subsection{Experimental Settings}

We evaluate the efficiency of our proposed algorithm using both real-world and synthetic datasets.

\noindent\textbf{Real-World Datasets.}
We use three widely adopted social networks:  Epinions \cite{Epinions} ($Epin$), Twitter \cite{NIPS2012_7a614fd0} ($Twit$), and DBLP \cite{dblp_dataset} ($DBLP$). Each edge $e(u,v)$ is assigned a weight in $(0,1]$ following a $Gaussian$ distribution to model the influence of $u$ on $v$. To construct spatial-social networks, users are mapped to 2D coordinates on the California road network \cite{cal_road_netword} using a $Uniform$ distribution. Each user is associated with $[1,10]$ check-in locations, with visit frequencies in $[1,10]$, both uniformly distributed. Dataset statistics are summarized in Table~\ref{tbl:real-datasets}.

\begin{table}
 \centering
 \caption{\small Statistics of Real-World Graph Datasets}\small \scriptsize\vspace{-2ex}
\label{tbl:real-datasets}   
 \begin{tabular}{l|l|l}\hline
 {\bf Name} &  {\bf Nodes} &  {\bf Edges}\\
 \hline \hline
  Epinions social network \cite{Epinions}&75879 &508837\\
  Twitter \cite{NIPS2012_7a614fd0}& 81306 & 1768149\\
  DBLP collaboration network \cite{dblp_dataset} &317080 & 1049866 \\
 \hline
\end{tabular}
\end{table}

\begin{table}
 \centering\vspace{-2ex}
 \caption{\small Parameter Settings}\small \scriptsize\vspace{-2ex}
\label{tbl:parameter}   
 \begin{tabular}{l|l}\hline
 {\bf Parameter} & \qquad\qquad\qquad {\bf Values}\\
 \hline \hline
    the number of users $|V_s|$ in $G_s$ & 10K, 20K, \textbf{30K}, 40K, 50K,\\ 
    & 100K, 200K\\
    the number of triangles $k$ in $G_s$ & 2, \textbf{3}, 4, 5, 6\\
    the social distance threshold $d$ & 1, 2, \textbf{3}, 4, 5\\
    the size of query keyword set |Q|& 3, 5, \textbf{7}, 9, 11\\    
    the users visiting frequency threshold $\omega$ & 0.2, \textbf{0.4}, 0.6, 0.7, 0.9 \\
    the POIs visited frequency threshold $\pi$ & 0.2, \textbf{0.4}, 0.6, 0.7, 0.9  \\
    the influence score threshold $\theta$ & 0.2, \textbf{0.4}, 0.6, 0.7, 0.9  \\
    the spatial distance threshold $\sigma$ & 1, 2, 3, 4, \textbf{5}, 6\\
 \hline
\end{tabular}\vspace{2ex}
\end{table}

\noindent\textbf{Synthetic Datasets.}
We generate an artificial social network $G_s$ with varying numbers of users, where each user connects to $[8,40]$ other users. For each edge $e(u,v)$, we assign a weight in $(0,1]$ following a $Gaussian$ distribution to represent social influence. A road network $G_r$ is constructed by generating $20K$ intersection points in a 2D space and connecting them using the Gabriel Graph Algorithm \cite{Gabriel-Graph-Algorithm}. POIs are placed uniformly along road edges and assigned $[1,8]$ keywords selected from a dictionary of up to 50 terms. We utilize three distributions—$Uniform$, $Gaussian$, and $Skew$ (Zipf skewness = 0.8)—to generate three synthetic datasets, denoted as $Unif$, $Gaus$, and $Skew$, respectively. The social and road networks are integrated into a bipartite spatial-social network, where each user is assigned $[1,10]$ check-ins with visit frequencies in $[1,10]$, both following a $Uniform$ distribution.

\noindent\textbf{Temporal Dynamic Updates Evaluation.} We evaluate the efficiency of our maintenance mechanisms by measuring the computational overhead across batch sizes $|\Delta U| \in \{10, 15, 25, 50, 100, 200\}$ on both the real-world $Twit$ and the synthetic  $Unif$ $TBSSN$ graphs. The users in $|\Delta U|$ were selected randomly following a $Uniform$ distribution. The evaluation is categorized into four key areas: first, $Data_B$ assesses the time required to synchronize check-in frequencies and recompute offline spatial-social metrics. To validate active community answers, we compare $Comm_1$, which represents the average CPU time to update community answers by testing affected users individually, against $Comm_B$, the average time for processing the entire update batch simultaneously to minimize $Refinement$ triggers and redundant checks. Finally, $Tree_B$ evaluates the cost of updating hierarchical boundary constraints and performing user and node migrations within the index $\mathcal{I}$. Each experiment is conducted for both Insertion and Deletion operations to simulate the sliding window model for affected users in the batch.

\noindent\textbf{Evaluation Methodology.}
All datasets are indexed using our proposed indexing tree, and experiments are conducted on the constructed bipartite spatial-social networks. To the best of our knowledge, this is the first experimental study of the $KCS\mbox{-}BSSN$ query. We evaluate efficiency by comparing the runtime of our approach with a baseline method. In the baseline, five subgraphs within social distance $d$ from the query user $q$ are sampled, and the total CPU cost is estimated by multiplying their average runtime by the total number of possible subgraphs within distance $d$ in $G_s$. In each experiment, we vary one parameter while fixing the others to their default values (highlighted in bold in Table~\ref{tbl:parameter}). The thresholds $\omega$ and $\pi$ are normalized to $(0,1]$, where 0.1 and 1 denote the minimum and maximum possible values, respectively. All experiments were conducted on a machine with an Intel Core i7 2.8GHz CPU and 16GB RAM.

\subsection{\texorpdfstring{$KCS\mbox{-}BSSN$}{KCS-BSSN} Performance Evaluation}

\begin{figure}[t]
    \centering
     \begin{subfigure}[t]{0.48\linewidth}
        \centering
        \includegraphics[width=\linewidth]{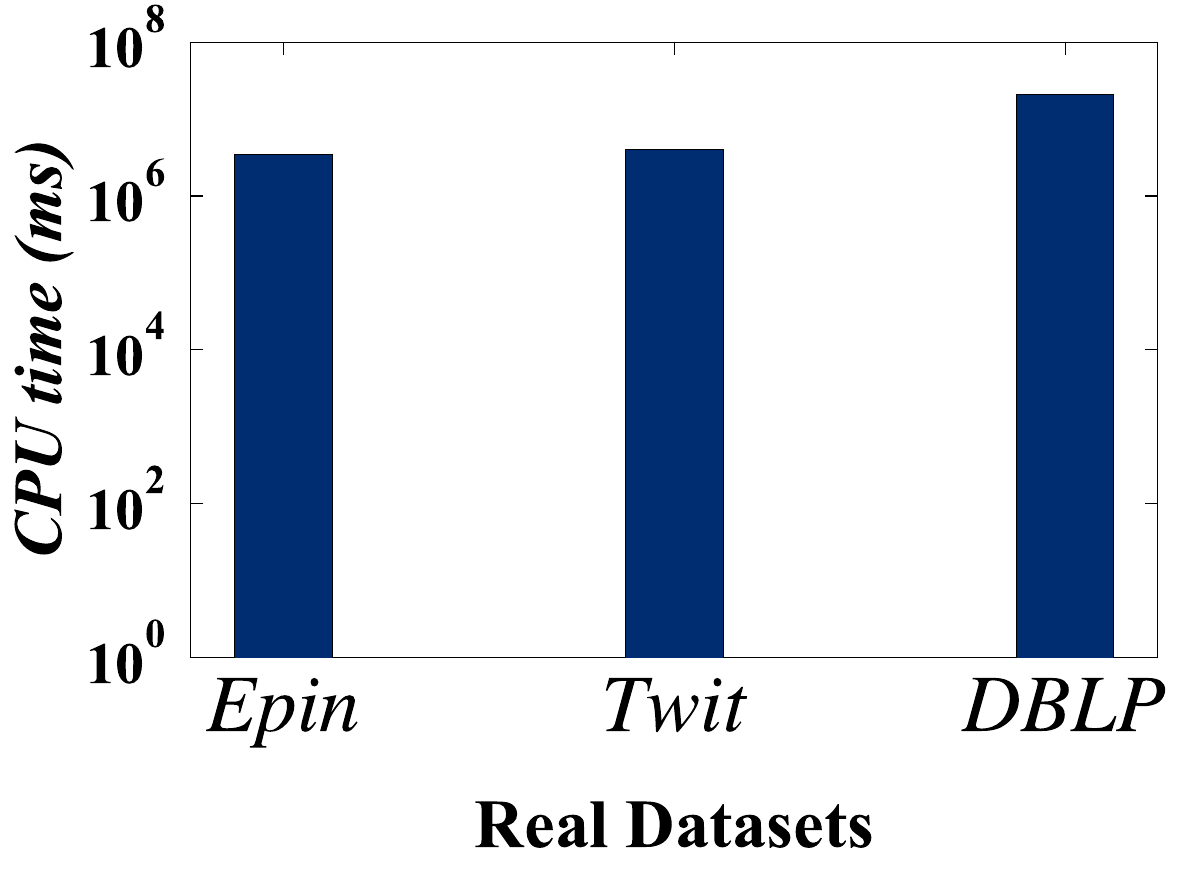}
        \caption{real-world graphs}
        \Description{real-world graphs}
        \label{fig:cpu_real_offline}
    \end{subfigure}
    \hfill
    \begin{subfigure}[t]{0.48\linewidth}
        \centering
        \includegraphics[width=\linewidth]{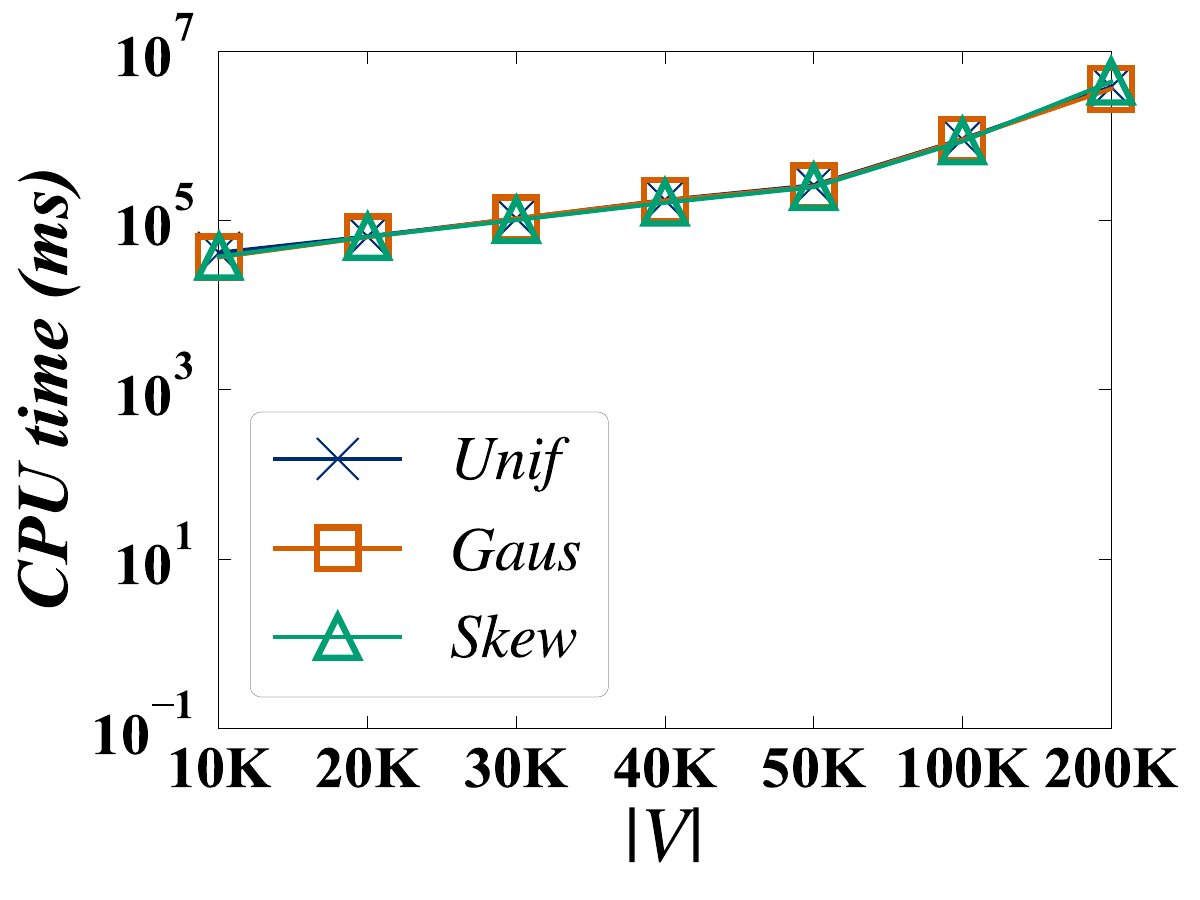}
        \caption{synthetic graphs}
        \Description{synthetic graphs}
        \label{fig:cpu_syn_offline}
    \end{subfigure}
    
    \caption{The $KCS\mbox{-}BSSN$ offline time vs. real/synthetic graph datasets.}
    \Description{The $KCS\mbox{-}BSSN$ offline time vs. real/synthetic graph datasets.}
    \label{fig:cpu_results_datasets_offline}
\end{figure}

\begin{figure}[t]
    \centering
    \begin{subfigure}[t]{0.48\linewidth}
        \centering
        \includegraphics[width=\linewidth]{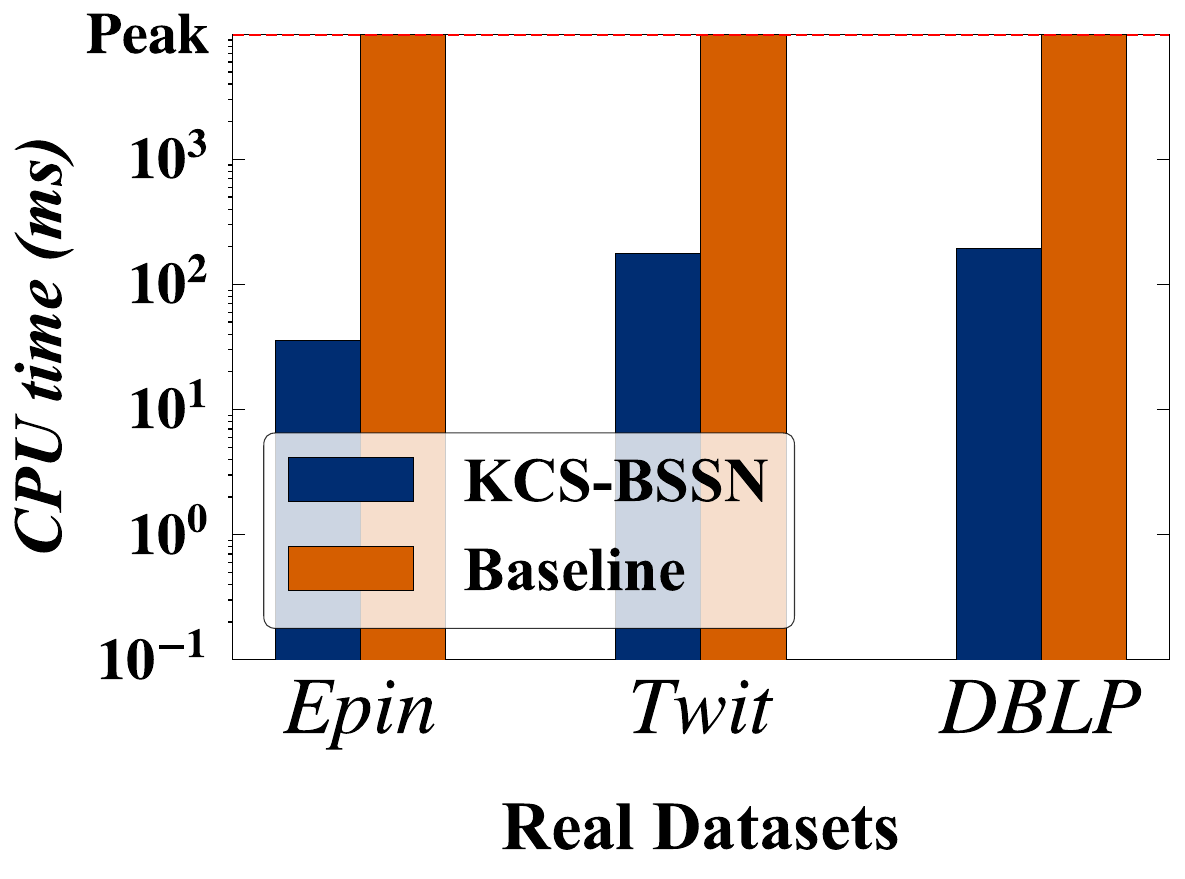}
        \caption{real-world graphs}
        \Description{real-world graphs}
        \label{fig:cpu_real}
    \end{subfigure}    
    \hfill    
    \begin{subfigure}[t]{0.48\linewidth}
        \centering
        \includegraphics[width=\linewidth]{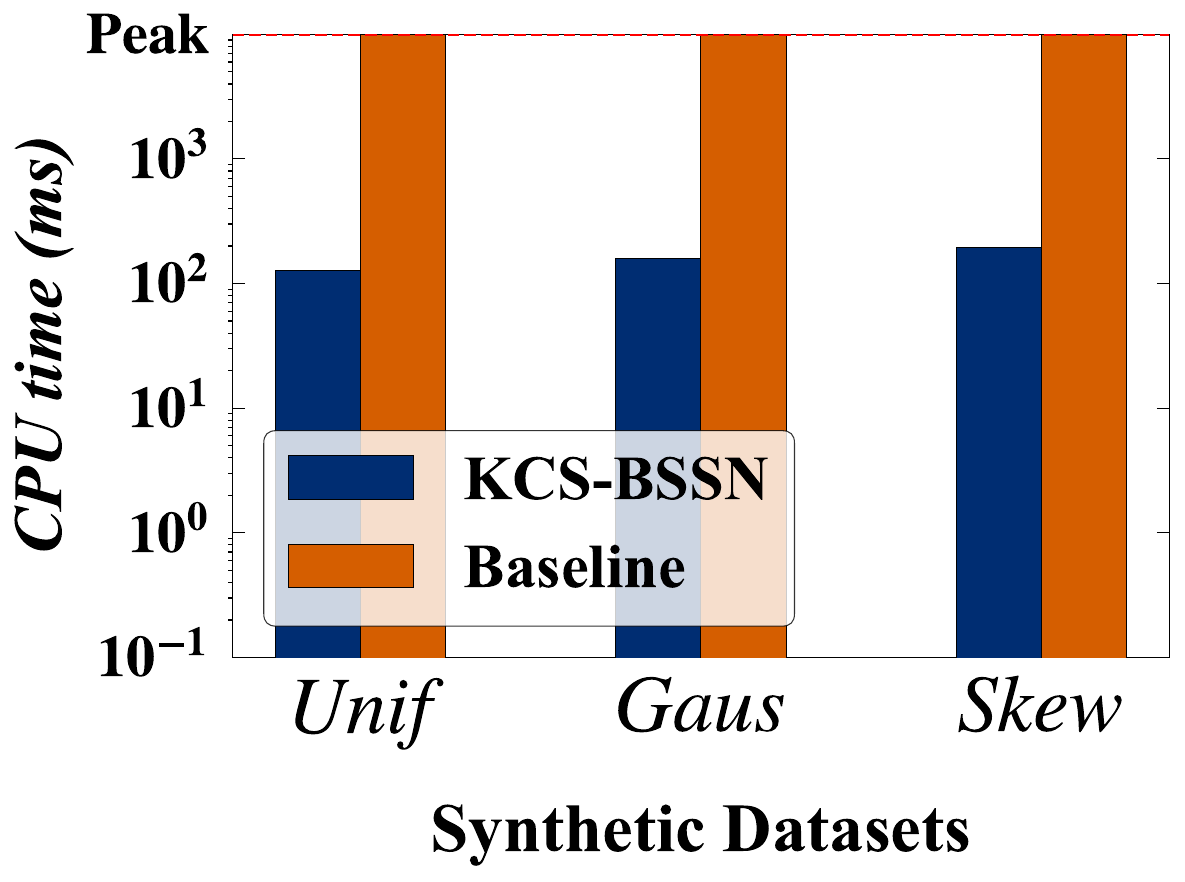}
        \caption{synthetic graphs}
         \Description{synthetic graphs}
        \label{fig:cpu_syn}
    \end{subfigure}
    
    \caption{The $KCS\mbox{-}BSSN$ performance vs. real/synthetic graph datasets.}
     \Description{The $KCS\mbox{-}BSSN$ performance vs. real/synthetic graph datasets.}
    \label{fig:cpu_results_datasets}
\end{figure}

\begin{figure*}[t]
    \centering
    
    \begin{subfigure}{0.23\linewidth}
        \includegraphics[width=\linewidth]{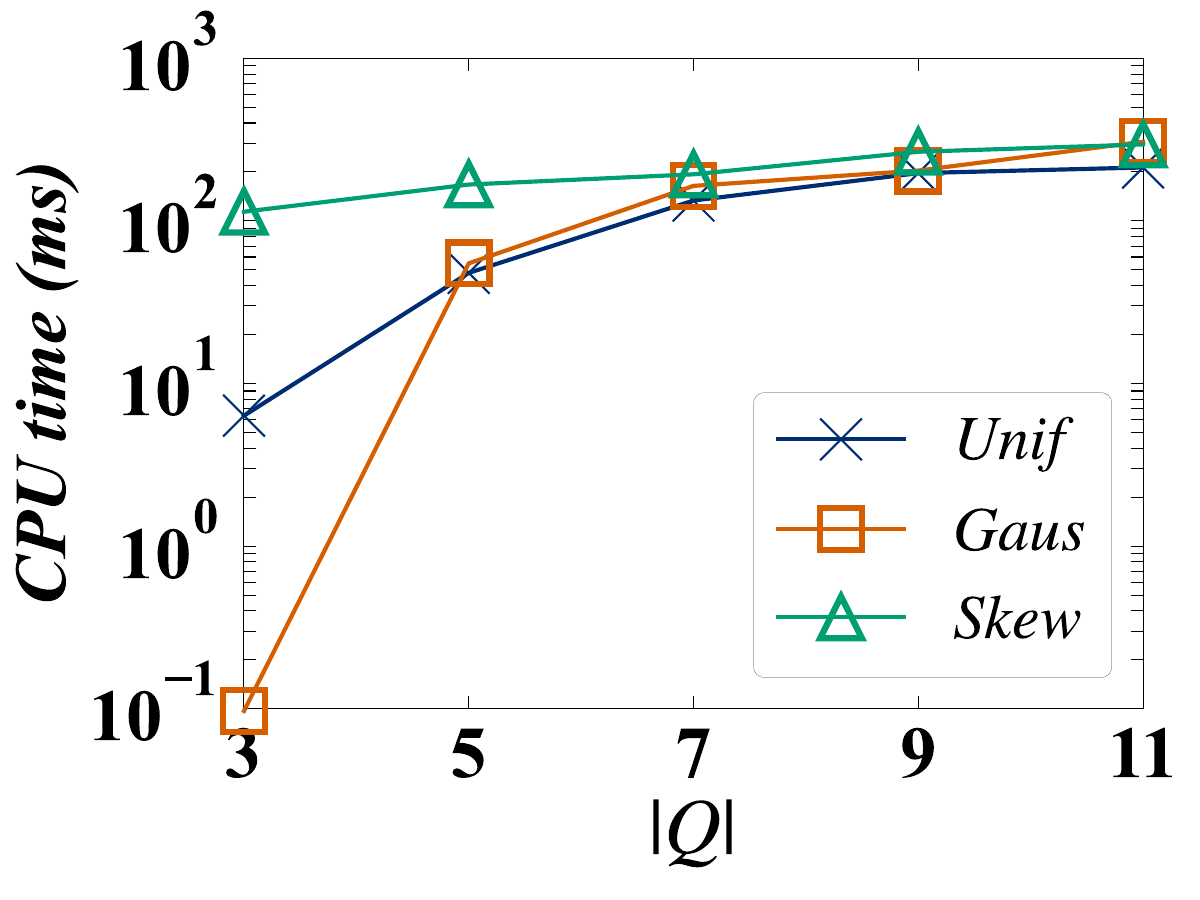}
         \caption{time cost vs. $|Q|$ }
         \Description{time cost vs. $|Q|$ }
    \end{subfigure}
    \hfill
    \begin{subfigure}{0.23\linewidth}
        \includegraphics[width=\linewidth]{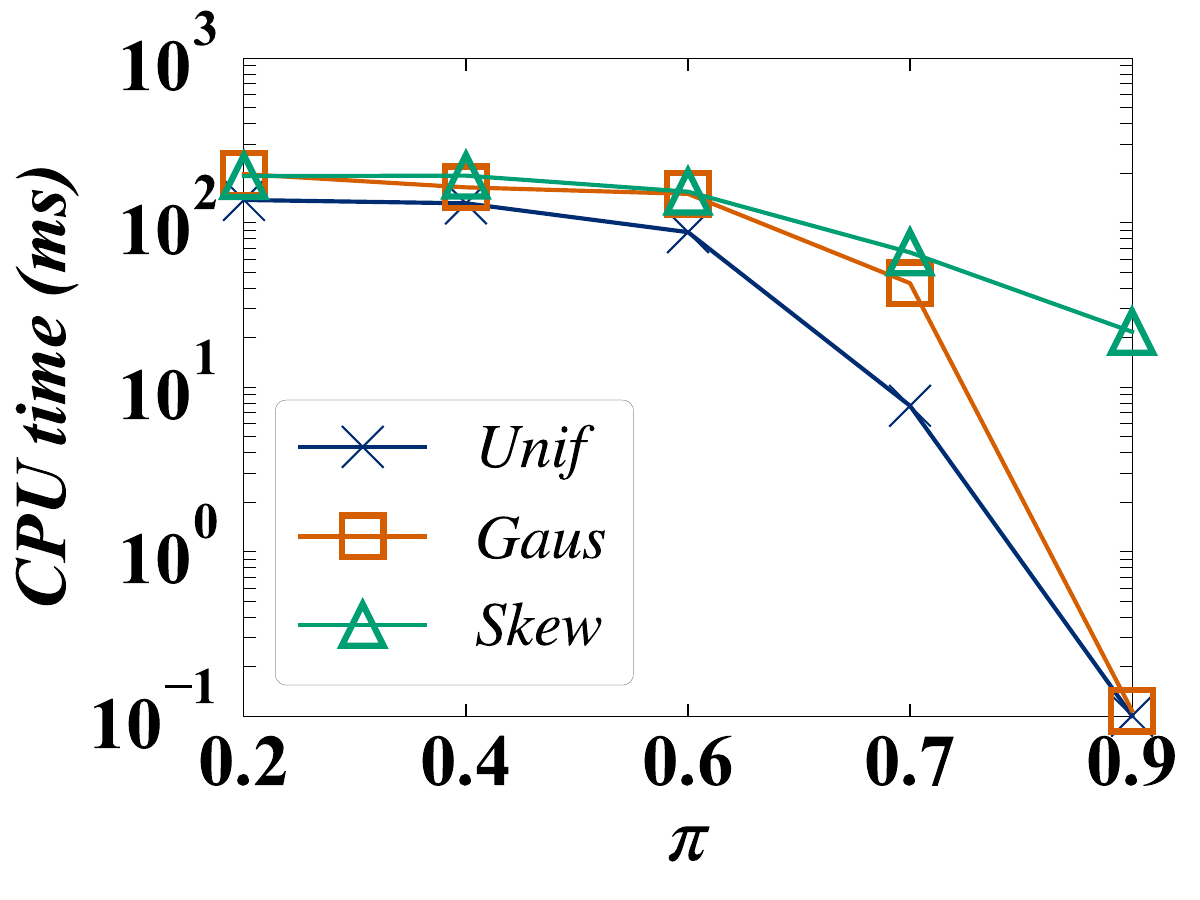}
        \caption{time cost vs. $\pi$}
         \Description{time cost vs. $\pi$}
    \end{subfigure}
     \hfill
        \begin{subfigure}{0.23\linewidth}
        \includegraphics[width=\linewidth]{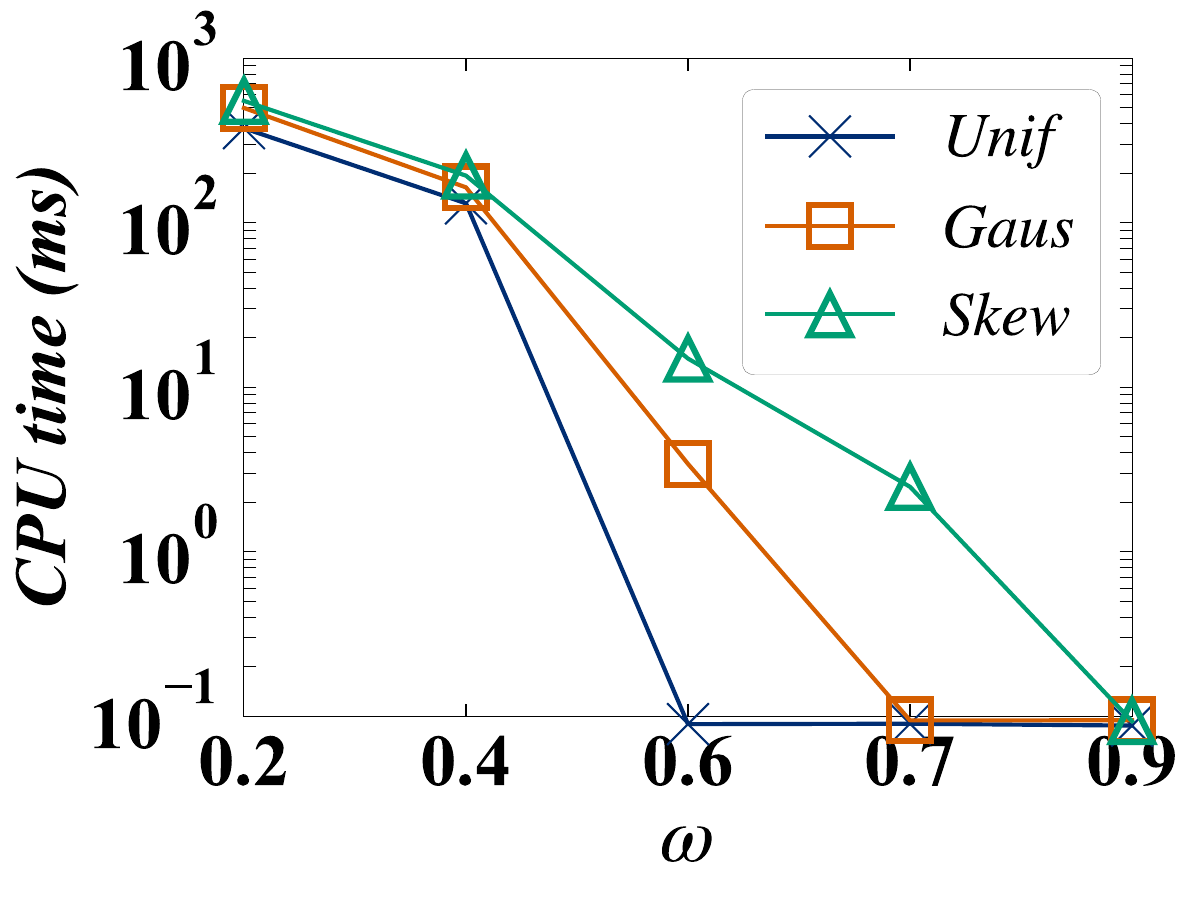}
         \caption{time cost vs. $\omega$}
         \Description{time cost vs. $\omega$}
    \end{subfigure}
    \hfill
        \begin{subfigure}{0.23\linewidth}
        \includegraphics[width=\linewidth]{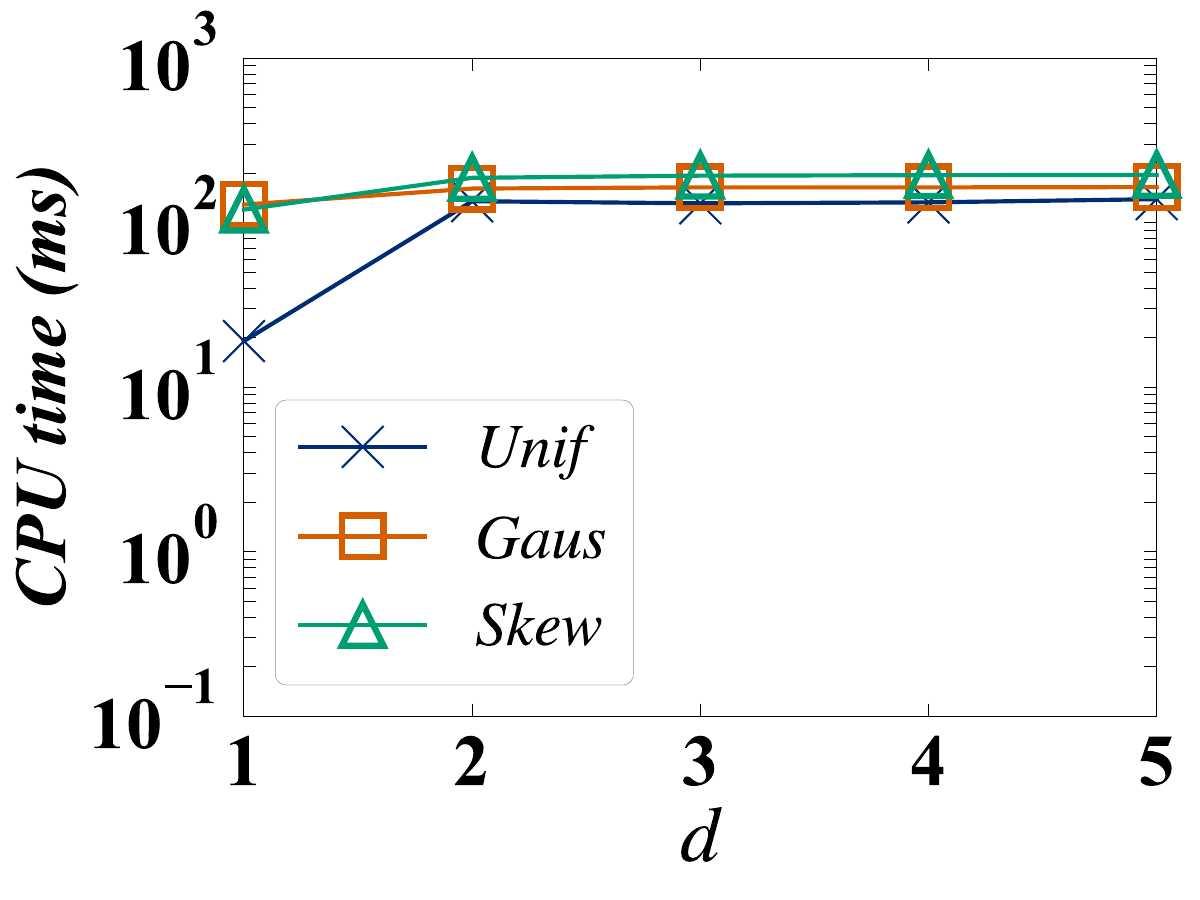}
        \caption{time cost vs. $d$}
         \Description{time cost vs. $d$}
    \end{subfigure}
    \vspace{0.4cm}

    \begin{subfigure}{0.23\linewidth}
        \includegraphics[width=\linewidth]{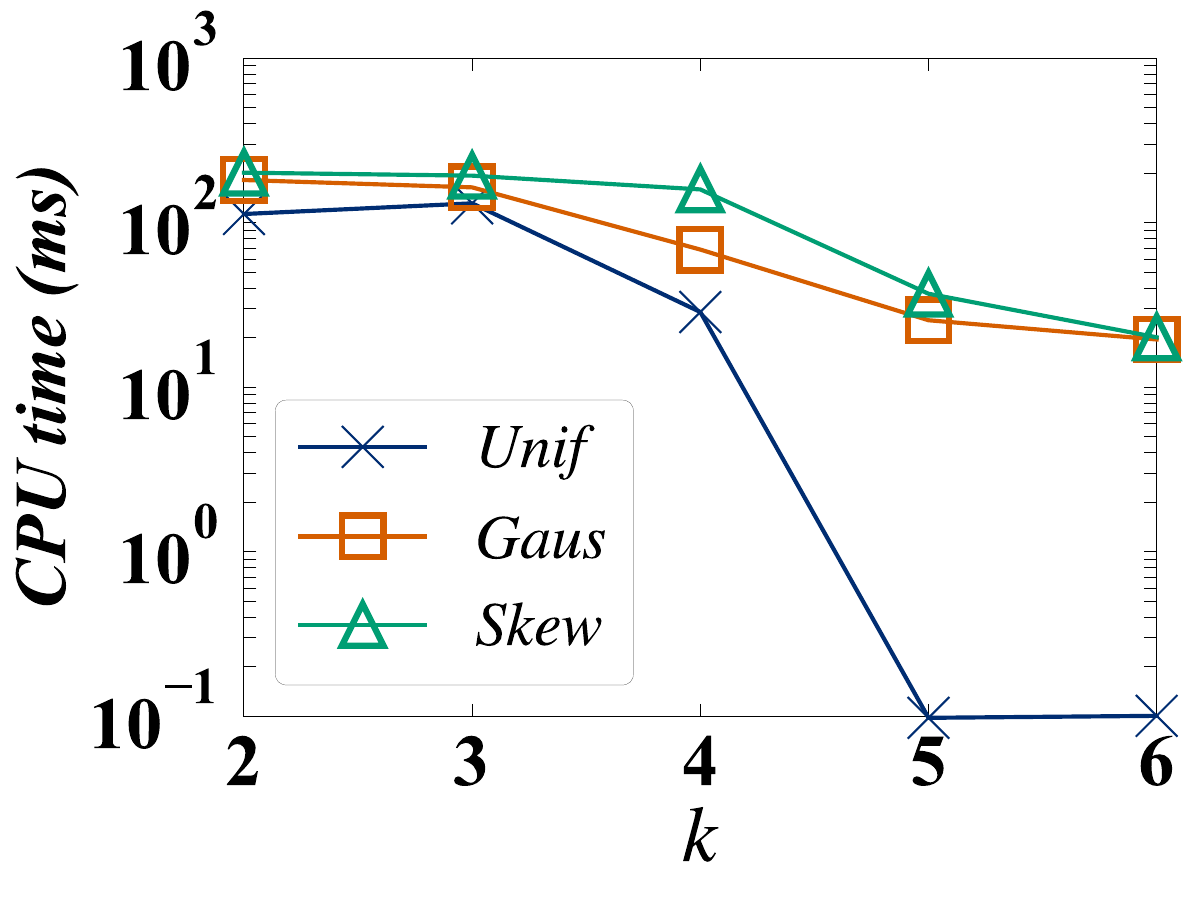}
        \caption{time cost vs. $k$}
        \Description{time cost vs. $k$}
    \end{subfigure}
    \hfill
     \begin{subfigure}{0.23\linewidth}
        \includegraphics[width=\linewidth]{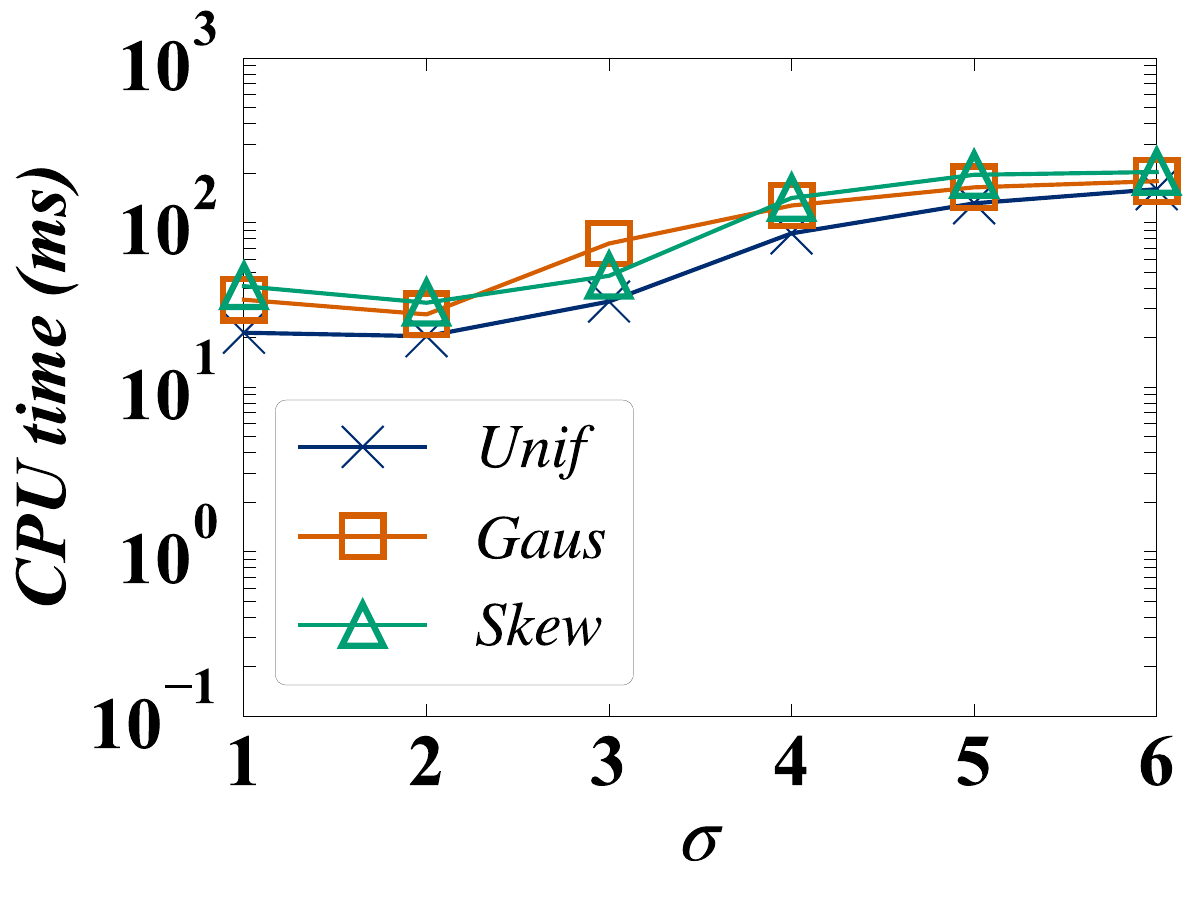}
         \caption{time cost vs. $\sigma$}
          \Description{time cost vs. $\sigma$}
    \end{subfigure}
    \hfill
    \begin{subfigure}{0.23\linewidth}
        \includegraphics[width=\linewidth]{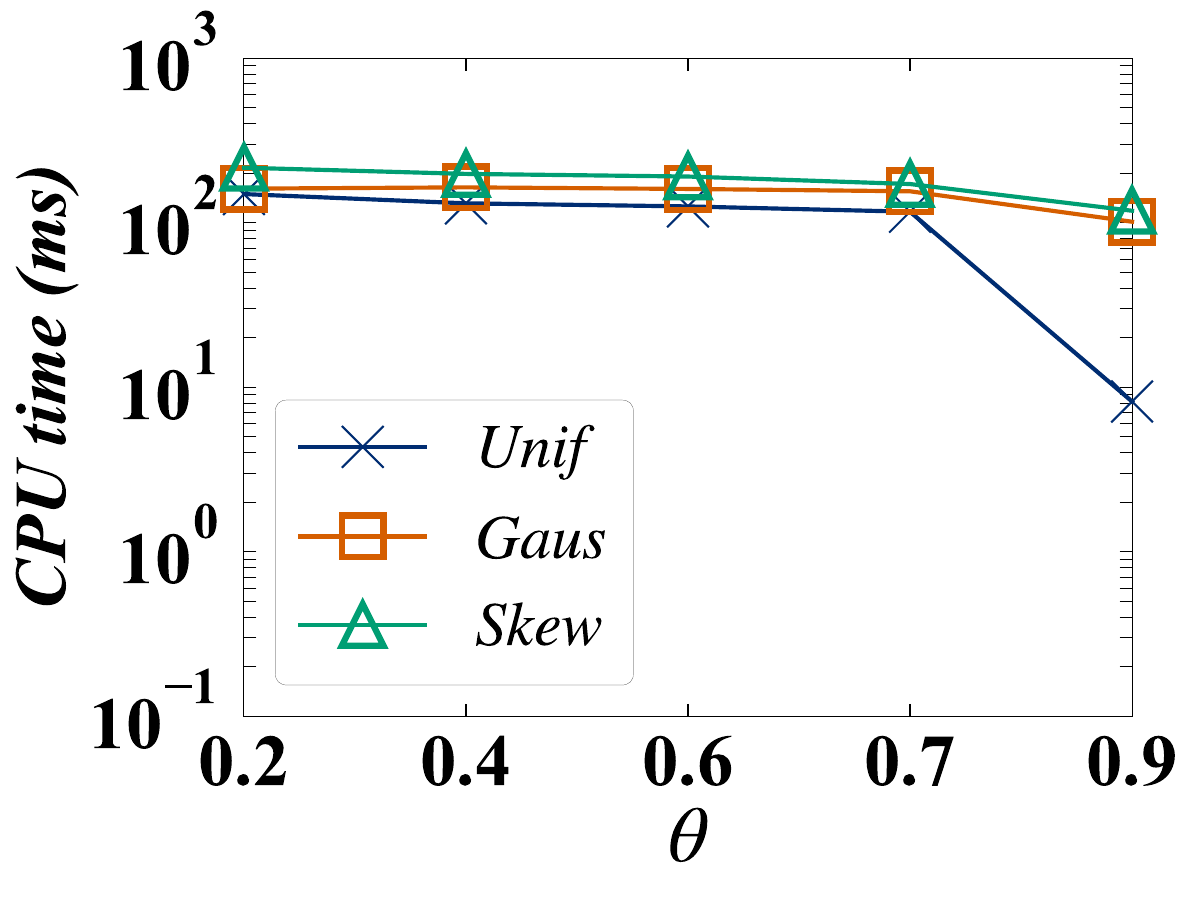}
         \caption{time cost vs. $\theta$}
          \Description{time cost vs. $\theta$}
    \end{subfigure}
    \hfill
    \begin{subfigure}{0.23\linewidth}
        \includegraphics[width=\linewidth]{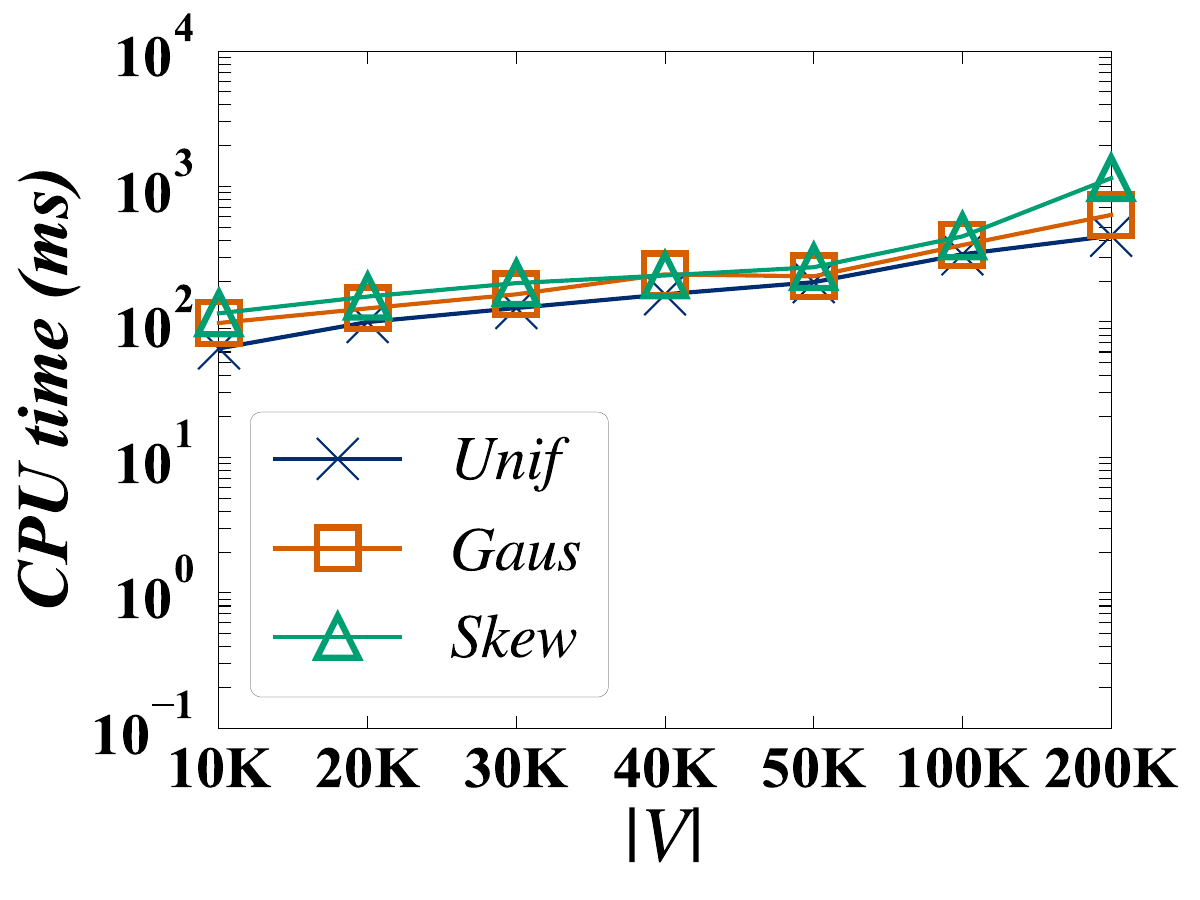}
        \caption{time cost vs. $N$ }
        \Description{time cost vs. $N$ }
    \end{subfigure}
    \caption{Overall $KCS\mbox{-}BSSN$ performance comparison for different parameter settings.}
    \Description{Overall $KCS\mbox{-}BSSN$ performance comparison for different parameter settings.}
    \label{fig:all_cpu_results}
\end{figure*}

\begin{figure}[t]
    \centering
     \begin{subfigure}[t]{0.48\linewidth}
        \centering
        \includegraphics[width=\linewidth]{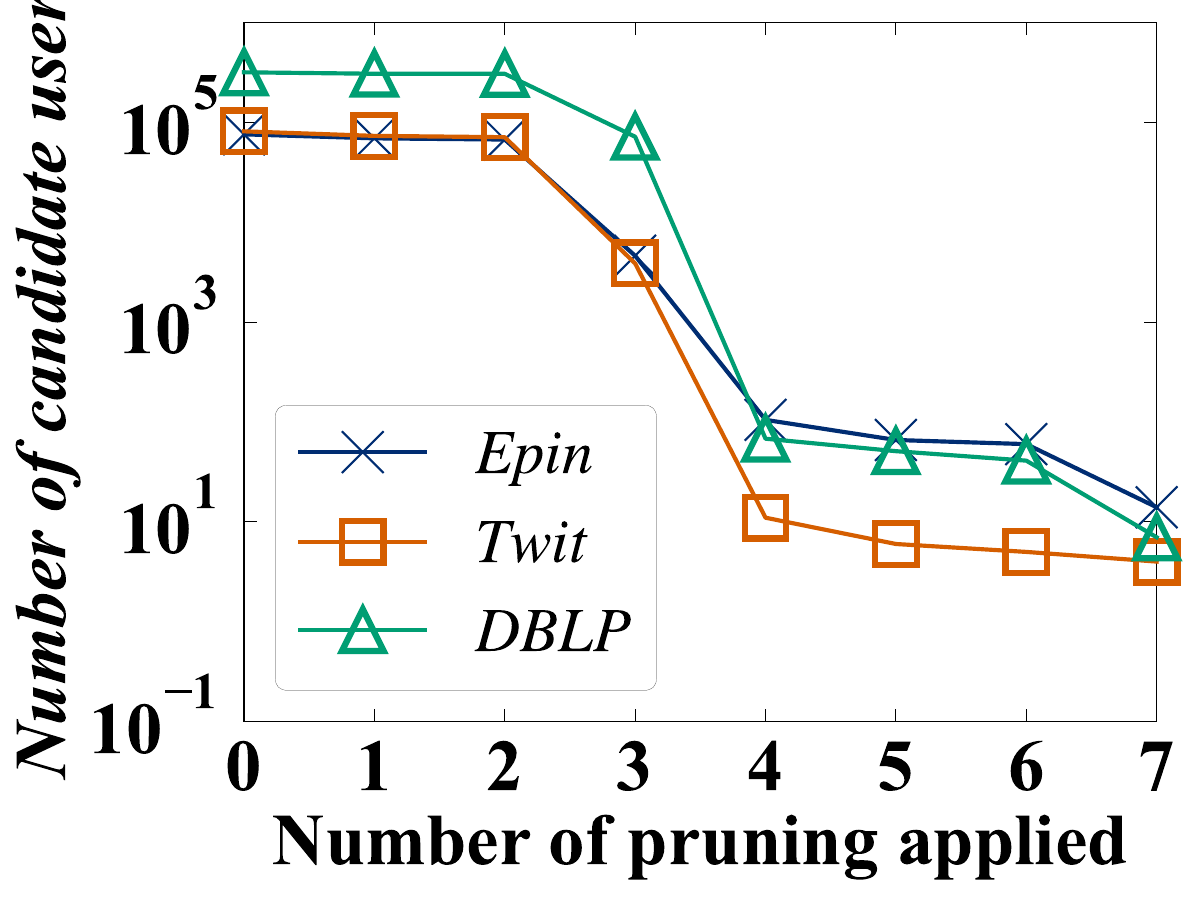}
        \caption{real-world graphs}
          \Description{real-world graphs}
        \label{fig:pruning_power_real}
    \end{subfigure}
    \hfill
    \begin{subfigure}[t]{0.48\linewidth}
        \centering
        \includegraphics[width=\linewidth]{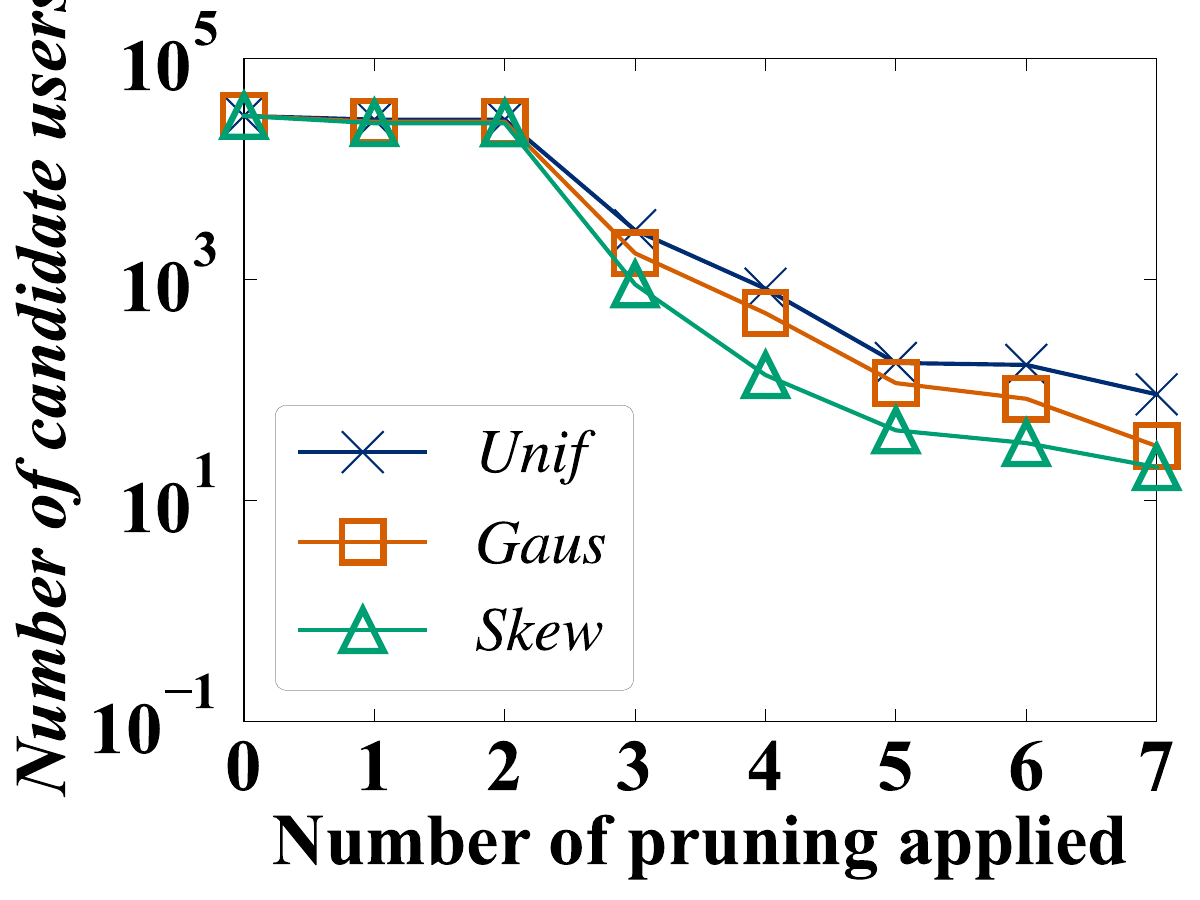}
        \caption{synthetic graphs}
         \Description{synthetic graphs}
        \label{fig:pruning_power_syn}
    \end{subfigure}
    
    \caption{The ablation study of our $KCS\mbox{-}BSSN$ pruning strategies on real/synthetic graphs, in terms of the pruning power.}
    \Description{The ablation study of our $KCS\mbox{-}BSSN$ pruning strategies on real/synthetic graphs, in terms of the pruning power.}
    \label{fig:pruning_power}
\end{figure}

\textbf{The $KCS\mbox{-}BSSN$ Offline Processing Time vs. Real/Synthetic Datasets.}
Figure~\ref{fig:cpu_results_datasets_offline} reports the offline processing time of $KCS\mbox{-}BSSN$, including data preparation and indexing tree construction, on both real and synthetic datasets. For real datasets (Figure~\ref{fig:cpu_results_datasets_offline}(a)), networks are processed at their original sizes (Table~\ref{tbl:real-datasets}). For synthetic datasets (Figure~\ref{fig:cpu_results_datasets_offline}(b)), we vary the social network size $|V|$ to evaluate scalability, while keeping other parameters at their default values. Results show a clear correlation between $|V|$ and offline time, confirming that dataset size significantly impacts preprocessing cost.\\
\textbf{The $KCS\mbox{-}BSSN$ Performance vs. Real/Synthetic Datasets.}
Figure~\ref{fig:cpu_results_datasets} compares our approach with the baseline under default parameter settings. For real datasets (Figure~\ref{fig:cpu_results_datasets}(a)), original network sizes are used; for synthetic datasets (Figure~\ref{fig:cpu_results_datasets}(b)), $|V|$ is fixed at its default value. In all cases, $KCS\mbox{-}BSSN$ consistently outperforms the baseline. The baseline frequently reaches peak execution times due to the exponential number of subgraphs it explores, whereas our approach remains stable and scalable across different data environments.\\
\textbf{Performance vs. Query Keyword Set Size $|Q|$.}
Figure~\ref{fig:all_cpu_results}(a) evaluates performance for $|Q|\in{3,5,7,9,11}$. Runtime increases moderately as $|Q|$ grows, since more keywords expand the number of matching POIs and candidate users. Nevertheless, the overall variation in computation time remains small, demonstrating robustness to keyword expansion.\\
\textbf{Performance vs. POIs Average Visiting Frequency Threshold $\pi$.}
Figure~\ref{fig:all_cpu_results}(b) analyzes $\pi \in {0.2,0.4,0.6,0.7,0.9}$. As $\pi$ increases, execution time decreases due to stronger pruning of POIs with low average visit frequency, which in turn reduces candidate users. Runtime drops from roughly $10^2$ ms at $\pi=0.2$ to near $10^0$–$10^1$ ms at $\pi=0.9$. This pruning operates after keyword-based filtering, where POIs that do not match the query keywords have already been removed. However, as $\pi$ becomes larger, the constraint becomes much stricter, leading to a significant reduction in both POIs and their associated users, and thus a noticeable improvement in execution time.\\
\textbf{Performance vs. Users Total Visiting Frequency Threshold $\omega$.}
Figure~\ref{fig:all_cpu_results}(c) varies $\omega$ over ${0.2,0.4,0.6,0.7,0.9}$. Increasing $\omega$ significantly reduces candidate users, leading to notable runtime improvement. In fact, runtime decreases sharply between $\omega=0.4$ and $\omega=0.6$, and approaches minimal values for $\omega \geq 0.7$. In some distributions, the improvement spans nearly three orders of magnitude, indicating that $\omega$-based pruning is one of the most dominant filtering mechanisms. Keyword distribution also affects pruning behavior, particularly in the $Unif$ and $Gaus$ datasets, where selective POI pruning further reduces computational cost.\\
\textbf{Performance vs. Social Network Distance Threshold $d$.}
Figure~\ref{fig:all_cpu_results}(d) varies $d$ from 1 to 5 hops. As expected, runtime increases with $d$ because expanding the social radius enlarges the candidate subgraph and increases the number of users that must be examined. Execution time grows from approximately $10^1$ ms to around $10^2$ ms across the tested range. This behavior indicates that the evaluated datasets are relatively dense and well connected, as increasing $d$ quickly incorporates additional users into the search space. Nevertheless, the growth remains smooth and bounded within a single order of magnitude, without exhibiting exponential escalation. This demonstrates that the pruning mechanisms effectively control search expansion and ensure good scalability with respect to the social distance threshold.\\
\textbf{Performance vs. Triangle Support Threshold $k$.}
Figure~\ref{fig:all_cpu_results}(e) evaluates $k \in {2,3,4,5,6}$. Execution time decreases as $k$ increases, since stronger structural constraints enable more aggressive pruning of low-support users. For higher values of $k$ (e.g., 5 or 6), runtime drops close to $10^0$–$10^1$ ms in certain datasets. This indicates that a large portion of edges have support values below 5 or 6, and thus are eliminated early when stricter cohesiveness constraints are enforced. Consequently, higher $k$ values significantly reduce candidate density and overall computation time.\\
\textbf{Performance vs. Road Network Distance Threshold $\sigma$.}
Figure~\ref{fig:all_cpu_results}(f) varies $\sigma$ from 1 to 6. Larger $\sigma$ values expand the spatial search region, increasing the number of candidate POIs and thus the runtime. Execution time increases steadily from approximately $10^1$ ms to nearly $10^2$ ms as $\sigma$ grows, reflecting the expansion of the spatial search region. However, the growth remains smooth and bounded within a single order of magnitude, indicating stable spatial scalability and effective pruning control without noticeable performance spikes.\\
\textbf{Performance vs. Influence Score Threshold $\theta$.}
Figure~\ref{fig:all_cpu_results}(g) varies $\theta \in {0.2,0.4,0.6,0.7,0.9}$. Since influence pruning is applied in a late filtering stage, increasing $\theta$ gradually reduces candidate users and slightly decreases execution time. The reduction is moderate compared to $\omega$ and $k$, confirming that influence acts mainly as a refinement constraint rather than a primary pruning driver.\\
\textbf{Performance vs. Number of Users $|V|$.}
Figure~\ref{fig:all_cpu_results}(h) evaluates scalability for $|V|$ ranging from $10K$ to $200K$. Runtime increases with network size across all datasets ($Unif$, $Gaus$, $Skew$), as larger networks enlarge both social and spatial search spaces. The CPU time is about $10^2$ ms for our default setting $|V|=30K$, even under less restrictive parameter settings, the runtime remains within practical bounds. At $|V|=200K$, execution time is around $10^3$ ms for all three datasets, indicating that the method scales efficiently and maintains stable performance even for large-scale networks.

\noindent \textbf{Pruning Power Analysis on Real/Synthetic Datasets.}
Figures~\ref{fig:pruning_power} present an incremental evaluation of pruning effectiveness on real-world datasets ($Epin$, $Twit$, $DBLP$) and synthetic datasets ($Unif$, $Gaus$, $Skew$) under default settings. We measure the reduction in candidate users across eight stages, starting from the initial user set (Stage 0) and cumulatively activating: (1) Keyword-Based, (2) $\pi$-Based, (3) $\omega$-Based, (4) Social-Distance-Based, (5) Structural Cohesiveness, (6) Spatial-Distance-Based, and (7) Influence-Based pruning. Across all datasets, the candidate set decreases steadily at each stage. A significant reduction is observed after $\omega$-based pruning, highlighting its strong filtering power. In contrast, Keyword-Based and $\pi$-Based pruning show more moderate effects due to the high density of POIs associated with multiple keywords. Overall, the layered pruning strategy effectively reduces the search space and improves query efficiency in both real and synthetic environments.

\begin{figure}[t]
    \centering
     \begin{subfigure}[t]{0.48\linewidth}
        \centering
        \includegraphics[width=\linewidth]{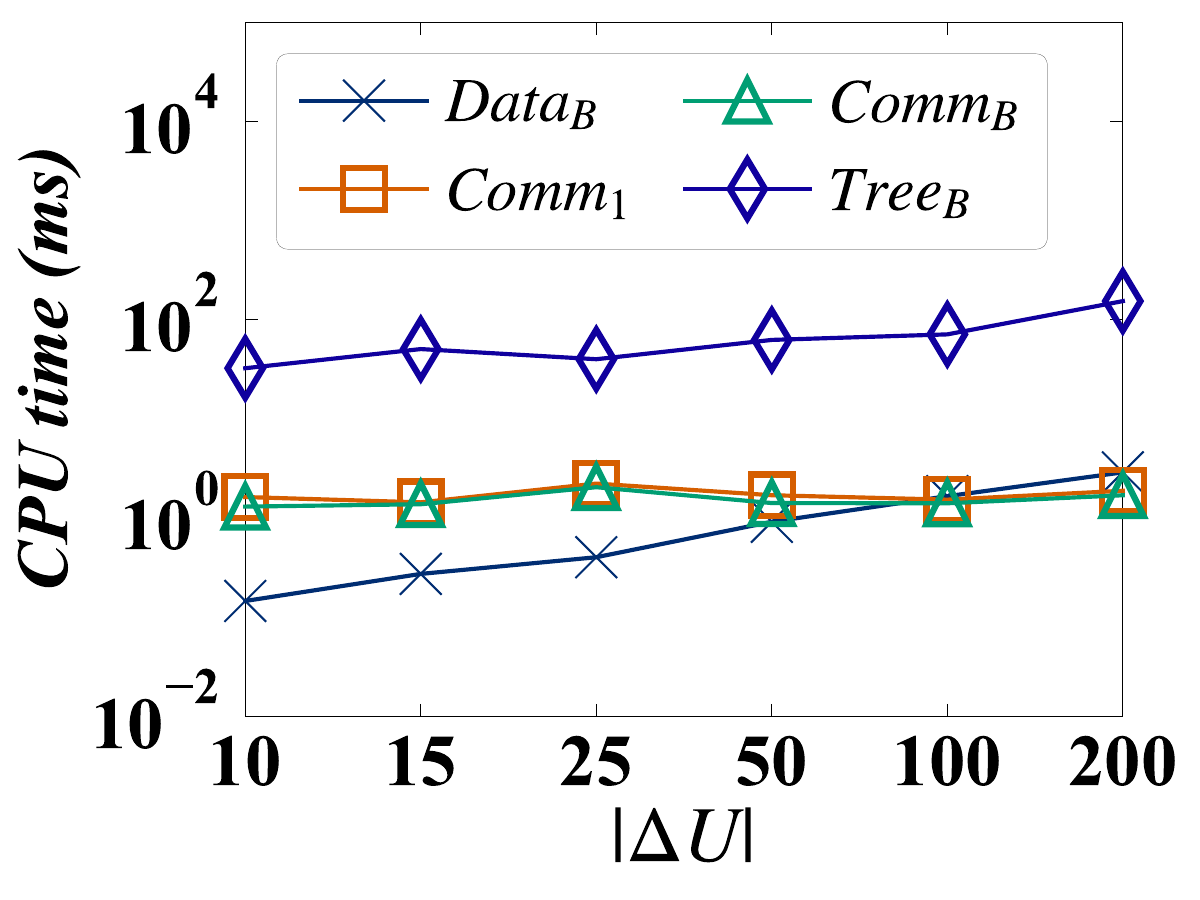}
        \caption{real-world graph ($Twit$)}
          \Description{real-world graphs}
        \label{fig:dynamic_inset_real}
    \end{subfigure}
    \hfill
    \begin{subfigure}[t]{0.48\linewidth}
        \centering
        \includegraphics[width=\linewidth]{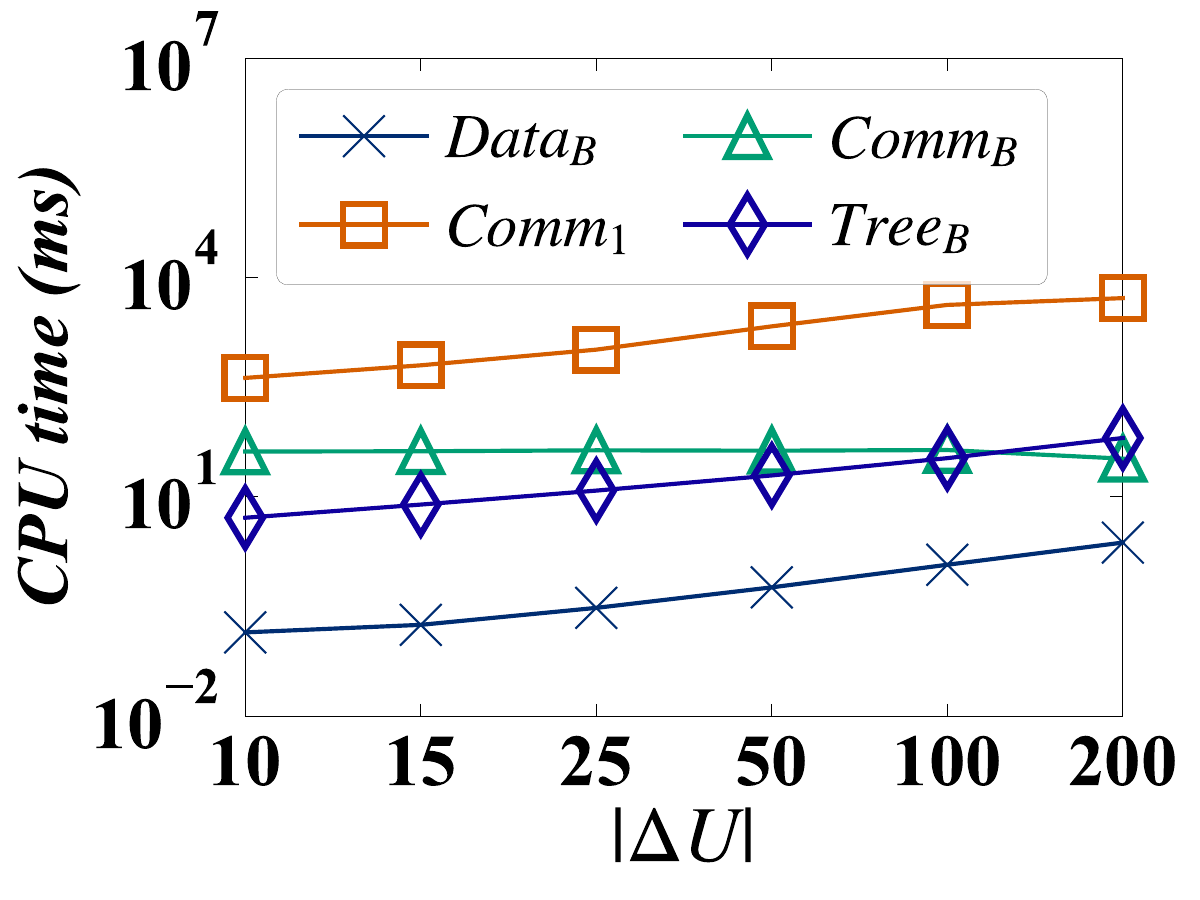}
        \caption{synthetic graph ($Unif$)}
         \Description{synthetic graphs}
        \label{fig:dynamic_inser_syn}
    \end{subfigure}
    \caption{The $KCS\mbox{-}TBSSN$ performance of temporal dynamic updates (insertion) vs. real ($Twit$) and synthetic ($Unif$) graphs.}
    \Description{Temporal dynamic updates (insertion) performance vs. real ($Twit$)/synthetic ($Unif$) graphs.}
    \label{fig:dynamic_insert}
\end{figure}

\vspace{2ex}
\subsection{\texorpdfstring{$KCS\mbox{-}TBSSN$}{KCS-TBSSN} Performance Evaluation}

\noindent\textbf{Temporal Dynamic Updates / Insertion Performance Analysis on Real ($Twit$) and Synthetic ($Unif$) $TBSSN$ Graphs.} Figure \ref{fig:dynamic_insert} illustrates the $KCS\text{-}TBSSN$ performance comparison of keyword-based community search over dynamic real-world ($Twit$) and synthetic ($Unif$) $TBSSN$ graphs, where the batch size, $|\Delta U|$, of insertion updates varies from 10 to 200, with all other community parameters held at default settings. The results reveal distinct performance trends across the maintenance categories. In the real graph ($Twit$), as shown in Figure \ref{fig:dynamic_insert}(a), processing times for all strategies remain relatively stable despite increases in batch size. Within this environment, $Tree_B$ consistently requires the highest CPU time, maintaining a range between $10$ and $100$ $ms$, while $Comm_1$ and $Comm_B$ demonstrate nearly identical efficient performance hovering near $1 ms$, and $Data_B$ functions as the most efficient strategy. In contrast, the synthetic graph ($Unif$) in Figure \ref{fig:dynamic_insert}(b) exhibits a more pronounced upward trend in CPU time as batch sizes grow. Specifically, $Comm_1$ is significantly more expensive, rising from approximately $100$ $ms$ to nearly $10^4$ $ms$ at 200 users, largely because communities in this dataset are generally larger than those in the $Twit$ dataset. In both datasets, $Comm_1$ requires more time because many selected users not initially in a community must undergo the refinement process. However, our proposed $Comm_B$ remains highly stable and efficient, consistently performing near $10$ $ms$, which validates that batch updates scales significantly better than individual updates under synthetic workloads for insertion operations.

\begin{figure}[t]
    \centering
     \begin{subfigure}[t]{0.48\linewidth}
        \centering
        \includegraphics[width=\linewidth]{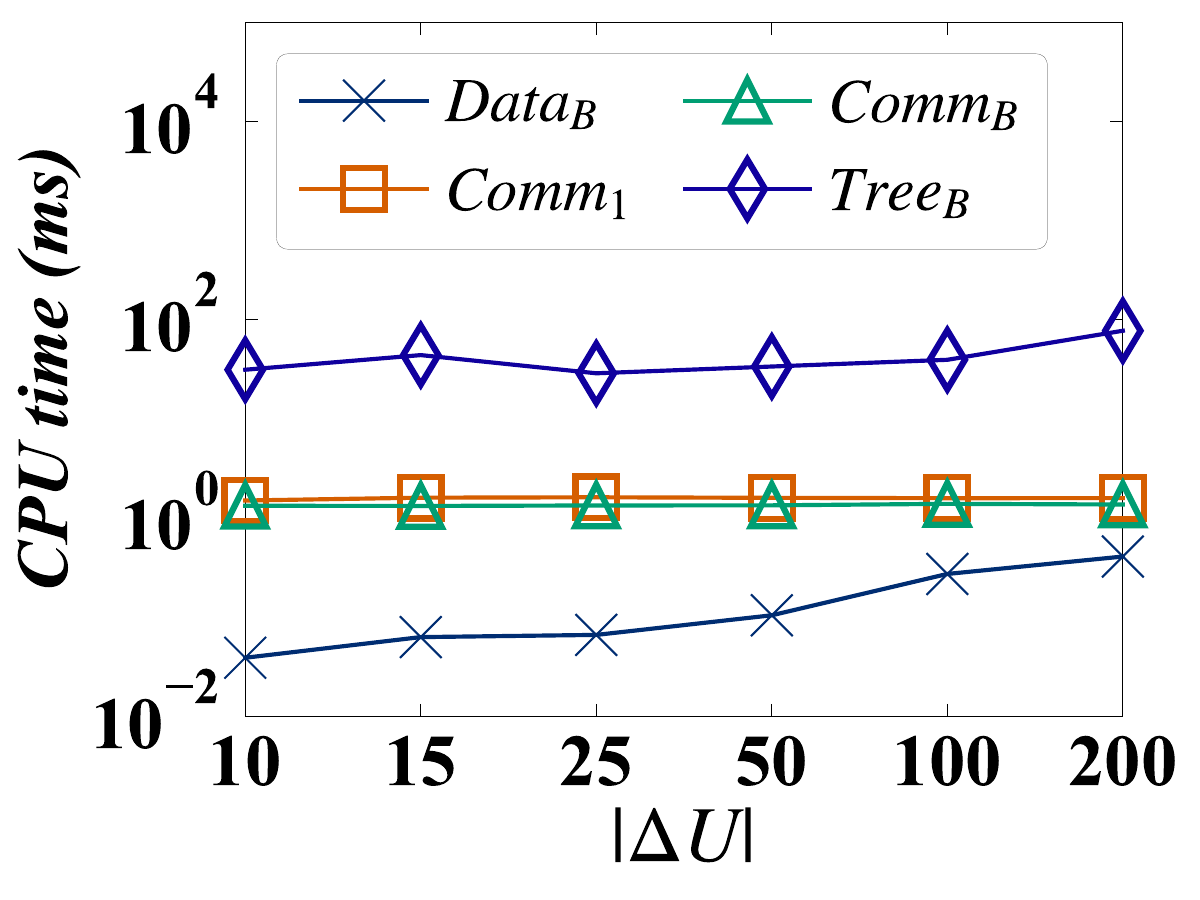}
        \caption{real-world graph ($Twit$)}
          \Description{real-world graphs}
        \label{fig:dynamic_delete_real}
    \end{subfigure}
    \hfill
    \begin{subfigure}[t]{0.48\linewidth}
        \centering
        \includegraphics[width=\linewidth]{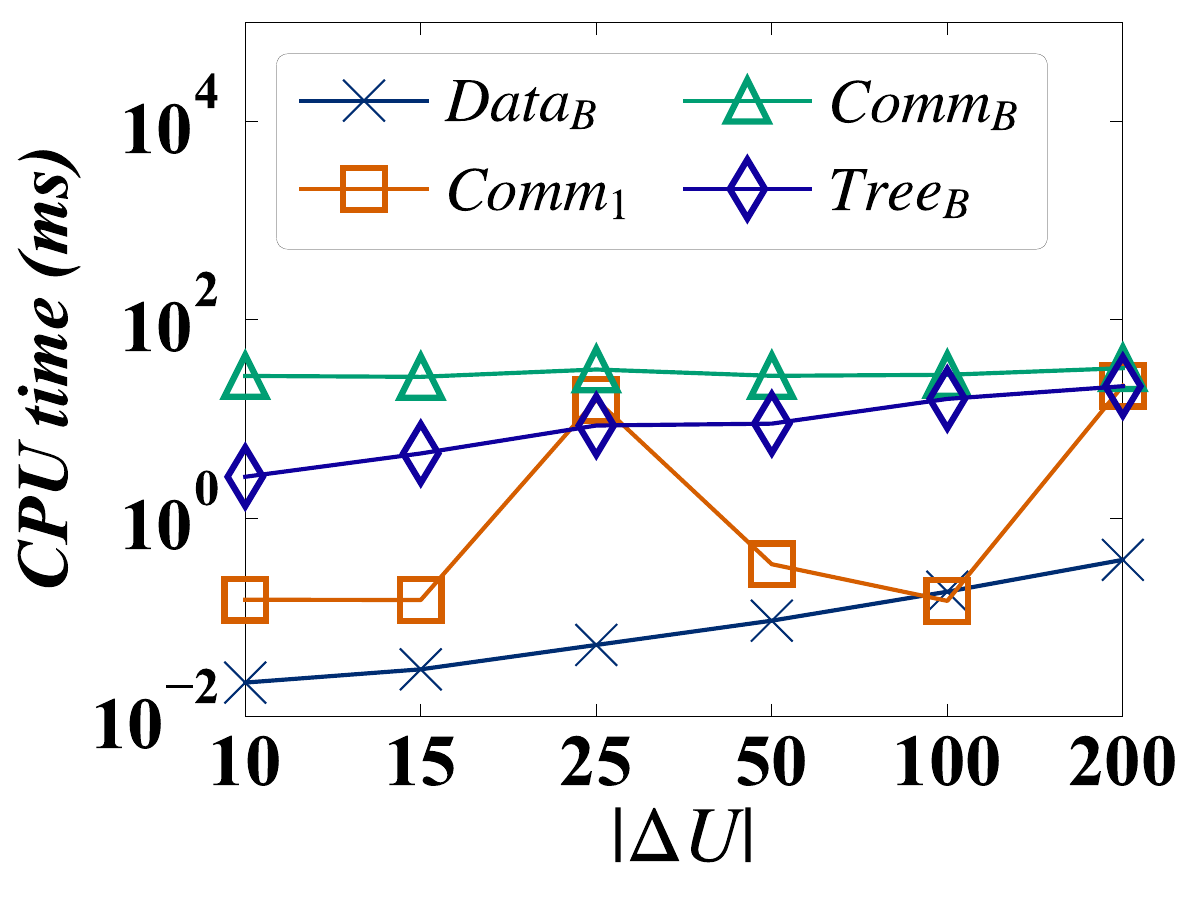}
        \caption{synthetic graph ($Unif$)}
         \Description{synthetic graphs}
        \label{fig:dynamic_delete_syn}
    \end{subfigure}
    
    \caption{The $KCS\mbox{-}TBSSN$ performance of temporal dynamic updates (deletion) vs. real ($Twit$) and synthetic ($Unif$) graphs.}
    \Description{Temporal dynamic updates (deletion) performance vs. real ($Twit$)/synthetic graph ($Unif$)}
    \label{fig:dynamic_delete}
\end{figure}

\noindent\textbf{Temporal Dynamic Updates / Deletion Performance Analysis on Real ($Twit$) and Synthetic ($Unif$) $TBSSN$ Graphs.} Figure \ref{fig:dynamic_delete} illustrates the $KCS\text{-}TBSSN$ performance upon deletion operations, simulating check-in expiration within the sliding window model. In this experiment, the batch size $|\Delta U|$ also varies from 10 to 200, while all other community parameters are set to their default values. For the real-world graph ($Twit$), as shown in Figure \ref{fig:dynamic_delete}(a), $Tree_B$ remains the most computationally intensive operation, with the CPU time stabilizing between $10$ $ms$ and $100$ $ms$. Conversely, the community maintenance strategies $Comm_1$ and $Comm_B$ (with individual updates each time and batch updates, respectively) exhibit high ad similar efficiency (i.e. around $1 ms$), indicating that most users in $|\Delta U|$ do not belong to community answers (w.r.t. registered community search queries), allowing the update process to bypass them. $Data_B$ consistently emerges as the most efficient strategy, re-computing metrics with minimal time cost. For synthetic graph ($Unif$) in Figure \ref{fig:dynamic_delete}(b), $Comm_1$ performs well across most trials but suffers significant latency spikes at $|\Delta U|=25$ and $200$, while $Comm_B$ maintains stable performance. These spikes suggest that specific batch sizes included more users belonging to community answers, triggering more intensive validation. While $Tree_B$ exhibits an upward trend in the CPU time as batch sizes increase, $Data_B$ remains remarkably stable and efficient. These experimental results confirm that the batch maintenance approach effectively reduces the computational overhead of expiring temporal data, except in cases where individual user validation ($Comm_1$) may be faster for users who are not part of existing community answers.

The experimental results for other real/synthetic graphs are similar and thus we do not report them here.

\section{Related Work}
\label{sec:related_work}
\noindent \textbf{Community Search (CS).}  Community search strategies are generally categorized into four main methodologies \cite{10.1145/3768576}: peeling \cite{10.1016/j.ins.2023.119511, LI2024111961, HIC-Zhou-2023}, expansion \cite{10.1016/j.ins.2023.119511, LI2024111961}, pruning \cite{LI2024111961, HIC-Zhou-2023}, and indexing-based strategies \cite{zhang2025topr, Temp-bipartite-Li-2024}. The primary objective of CS is to retrieve a subgraph that satisfies specific constraints while containing a user-specified query vertex. Existing literature focuses on diverse constraints, including structural cohesiveness, keyword (attribute) homogeneity, embedded bipartite cores, community size, member influence, and spatial proximity.\\
\noindent \textbf{CS in Heterogeneous Graphs.} CS in heterogeneous graphs where different vertex and edge types coexist—aims to identify communities spanning multiple entity types. Chen et al. \cite{FCS-HGNN-chen-2024} proposed algorithms for both regular and large-scale heterogeneous graphs to identify single-type and multi-type communities. Zhou et al. \cite{HIC-Zhou-2023} introduced the concept of heterogeneous influential communities based on meta-path core models \cite{EECS-fang2020, PathSim-Sun-2011}. For a comprehensive overview of heterogeneous community models, refer to the recent survey by \cite{10.1145/3768576}.\\
\noindent \textbf{CS in Bipartite Graphs.} Bipartite graphs are a specialized form of heterogeneous graphs consisting of two disjoint vertex sets. Research in this area often focuses on the core structure. For instance, the community with $(\omega,\beta)$-core structure, where upper-layer vertices have a minimum degree of $\omega$ and lower-layer vertices have a minimum degree of $\beta$ \cite{10.1145/3132847.3133130}. To manage community scale, size constraints are often imposed, such that the upper-layer size $\leq \mathfrak{x}$ and the lower-layer size $\leq \mathfrak{z}$ \cite{10.1016/j.ins.2023.119511}. Zhou et al. \cite{10.1016/j.ins.2023.119511} developed peeling and expansion algorithms for these structures, while Li et al. \cite{LI2024111961} investigated the Maximal Size Constraint CS (MSCC), proving it to be NP-Hard and proposing the Expand-and-Filter (EFA++) algorithm for acceleration. To incorporate social impact, Zhang et al. \cite{zhang2025topr} addressed the $(\omega,\beta)$-influential community problem, while Li et al. \cite{Temp-bipartite-Li-2024} extended CS to temporal bipartite graphs using time-windowed $(\omega,\beta)$-cores. Furthermore, Xu et al. \cite{CSABiGXu2023} explored attributed bipartite graphs to maximize attribute similarity, and Wang et al. \cite{Scalable-CS-wang-2024} demonstrated that finding $k$-core communities with combined textual and numerical attributes is NP-Hard in in heterogeneous graph.\\
\noindent \textbf{CS in Spatial-Social Networks.} Recent research has increasingly focused on geo-social or spatial-social networks \cite{geo-social-wang-2022-tkde, Co-Engaged-haldar-2024-TKDE, Haldar-top-k-2023-TKDE, NGCS-wu-2022-access, multi-att-cs-guo-2021, group-plianing-ahmed-2022}. Haldar et al. \cite{Co-Engaged-haldar-2024-TKDE, Haldar-top-k-2023-TKDE} investigated location-based social networks, utilizing query locations to identify connected subgraphs that satisfy $k$-core requirements and distance thresholds. Similarly, Wang et al. \cite{geo-social-wang-2022-tkde} provided solutions for $k$-core structures with road-network radius constraints. Wu et al. \cite{NGCS-wu-2022-access} proposed top-$n$ community retrieval based on proximity to a query user $q$, on the social network with $k$-core cohesive and radius constraint. Rai et al. \cite{topk-cs-rai-2023-tkde} focused on retrieving communities similar to a given query community on the road network. Guo et al. \cite{multi-att-cs-guo-2021} investigated multi-attributed CS in road-social networks, balancing structural cohesiveness with road-network distance. Similarly, Ahmed et al. \cite{ahmedTopicBasedCom} incorporated both member influence and spatial proximity into multi-attributed frameworks.\\
\noindent \textbf{Summary.} The query proposed in this paper introduces the\\ $(\omega, \pi)\mbox{-}keyword\mbox{-}core$, which integrates a unique combination of constraints: keyword similarity, user influence, structural cohesion via $(k,d)\mbox{-}truss$, and road network distance $avg\_dist_r(.)$. Due to this multi-faceted constraint set, existing techniques cannot be applied directly to the $KCS\mbox{-}BSSN$ problem.

\section{Conclusions}
\label{sec:conclusion}
In this paper, we propose a novel community search query, called Keyword-based Community Search in Bipartite Spatial-Social Networks ($KCS\mbox{-}BSSN$), motivated by many real-world applications. The $KCS\mbox{-}BSSN$ query identifies a cohesive community that satisfies the $(\omega, \pi)\mbox{-}\textit{keyword-core}$ constraint, and ensures social influence, while considering both social connectivity and spatial proximity. To efficiently process the query, we develop effective pruning techniques to eliminate non-promising users early. We also design a novel indexing tree with a tailored cost model that integrates social and road network information. Based on this framework, we propose a two-phase algorithm that generates candidate communities and refines them to obtain exact results.  Then, we introduced a temporal extension ($TBSSN$) to the $KCS\mbox{-}BSSN$ model, integrating time-sensitive pruning to filter historical data. This approach ensures that discovered communities are socially and spatially cohesive while reflecting contemporary user behaviors. Extensive experiments on real and synthetic datasets demonstrate the efficiency and effectiveness of our proposed methods

\bibliographystyle{abbrv}
\bibliography{references}

@article{kdtruss-def,
author = {Huang, Xin and Lakshmanan, Laks V. S.},
title = {Attribute-Driven Community Search},
year = {2017},
issue_date = {May 2017},
publisher = {VLDB Endowment},
volume = {10},
number = {9},
issn = {2150-8097},
url = {https://doi.org/10.14778/3099622.3099626},
doi = {10.14778/3099622.3099626},
journal = {Proc. VLDB Endow.},
month = {may},
pages = {949–960},
numpages = {12}
}

@article{barbieri2013topic,
  title={Topic-aware social influence propagation models},
  author={Barbieri, Nicola and Bonchi, Francesco and Manco, Giuseppe},
  journal={Knowledge and information systems},
  volume={37},
  number={3},
  pages={555--584},
  year={2013},
  publisher={Springer}
}

@article{chen2015online,
  title={Online topic-aware influence maximization},
  author={Chen, Shuo and Fan, Ju and Li, Guoliang and Feng, Jianhua and Tan, Kian-lee and Tang, Jinhui},
  journal={Proceedings of the VLDB Endowment},
  volume={8},
  number={6},
  pages={666--677},
  year={2015},
  publisher={VLDB Endowment}
}

@article{DBLP:journals/corr/cs-DS-0310049,
  author    = {Vladimir Batagelj and
               Matjaz Zaversnik},
  title     = {An O(m) Algorithm for Cores Decomposition of Networks},
  journal   = {CoRR},
  volume    = {cs.DS/0310049},
  year      = {2003},
  url       = {http://arxiv.org/abs/cs/0310049},
  bibsource = {dblp computer science bibliography, https://dblp.org}
}

@article{SEIDMAN1983269,
title = {Network structure and minimum degree},
journal = {Social Networks},
volume = {5},
number = {3},
pages = {269-287},
year = {1983},
issn = {0378-8733},
doi = {https://doi.org/10.1016/0378-8733(83)90028-X},
url = {https://www.sciencedirect.com/science/article/pii/037887338390028X},
author = {Stephen B. Seidman}
}

@article{cohen2008trusses,
  title={Trusses: Cohesive subgraphs for social network analysis},
  author={Cohen, Jonathan},
  journal={National security agency technical report},
  volume={16},
  number={3.1},
  year={2008},
  publisher={Citeseer}
}

@inproceedings{10.1145/2588555.2610495,
author = {Huang, Xin and Cheng, Hong and Qin, Lu and Tian, Wentao and Yu, Jeffrey Xu},
title = {Querying K-Truss Community in Large and Dynamic Graphs},
year = {2014},
isbn = {9781450323765},
publisher = {Association for Computing Machinery},
address = {New York, NY, USA},
url = {https://doi.org/10.1145/2588555.2610495},
doi = {10.1145/2588555.2610495},
booktitle = {Proceedings of the 2014 ACM SIGMOD International Conference on Management of Data},
pages = {1311–1322},
numpages = {12},
location = {Snowbird, Utah, USA},
series = {SIGMOD '14}
}

@inproceedings{acquisti2006imagined,
  title={Imagined communities: Awareness, information sharing, and privacy on the Facebook},
  author={Acquisti, Alessandro and Gross, Ralph},
  booktitle={International workshop on privacy enhancing technologies},
  pages={36--58},
  year={2006},
  organization={Springer}
}

@book{gibbons1985algorithmic,
  title={Algorithmic graph theory},
  author={Gibbons, Alan},
  year={1985},
  publisher={Cambridge university press}
}

@article{fang2020survey,
  title={A survey of community search over big graphs},
  author={Fang, Yixiang and Huang, Xin and Qin, Lu and Zhang, Ying and Zhang, Wenjie and Cheng, Reynold and Lin, Xuemin},
  journal={The VLDB Journal},
  volume={29},
  number={1},
  pages={353--392},
  year={2020},
  publisher={Springer}
}

@article{ahmedTopicBasedCom,
author = {Al-Baghdadi, Ahmed and Lian, Xiang},
title = {Topic-Based Community Search over Spatial-Social Networks},
year = {2020},
issue_date = {August 2020},
publisher = {VLDB Endowment},
volume = {13},
number = {12},
issn = {2150-8097},
url = {https://doi.org/10.14778/3407790.3407812},
doi = {10.14778/3407790.3407812},
journal = {Proc. VLDB Endow.},
month = {jul},
pages = {2104–2117},
numpages = {14}
}

@article{Gabriel-Graph-Algorithm,
    author = {Gabriel, K. Ruben and Sokal, Robert R.},
    title = "{A New Statistical Approach to Geographic Variation Analysis}",
    journal = {Systematic Biology},
    volume = {18},
    number = {3},
    pages = {259-278},
    year = {1969},
    month = {09},
    issn = {1063-5157},
    doi = {10.2307/2412323},
    url = {https://doi.org/10.2307/2412323},
    eprint = {https://academic.oup.com/sysbio/article-pdf/18/3/259/4595606/18-3-259.pdf}
}

@article{10.1007/s00778-016-0451-4,
author = {Yuan, Long and Qin, Lu and Lin, Xuemin and Chang, Lijun and Zhang, Wenjie},
title = {I/O efficient ECC graph decomposition via graph reduction},
year = {2017},
issue_date = {April     2017},
publisher = {Springer-Verlag},
address = {Berlin, Heidelberg},
volume = {26},
number = {2},
issn = {1066-8888},
url = {https://doi.org/10.1007/s00778-016-0451-4},
doi = {10.1007/s00778-016-0451-4},
journal = {The VLDB Journal},
month = apr,
pages = {275–300},
numpages = {26}
}

@article{10.1145/3768576,
author = {Zhou, Lihua and Wang, Jialong and Song, Yixin and Wang, Lizhen and Chen, Hongmei},
title = {Community Search over Heterogeneous Information Networks: A Survey},
year = {2025},
issue_date = {March 2026},
publisher = {Association for Computing Machinery},
address = {New York, NY, USA},
volume = {58},
number = {4},
issn = {0360-0300},
url = {https://doi.org/10.1145/3768576},
doi = {10.1145/3768576},
journal = {ACM Comput. Surv.},
month = oct,
articleno = {100},
numpages = {37}
}

@article{10.1016/j.ins.2023.119511,
author = {Zhou, Keqi and Xin, Junchang and Chen, Jinyi and Zhang, Xian and Wang, Beibei and Wang, Zhiqiong},
title = {Effective and efficient community search with size constraint on bipartite graphs},
year = {2023},
issue_date = {Nov 2023},
publisher = {Elsevier Science Inc.},
address = {USA},
volume = {647},
number = {C},
issn = {0020-0255},
url = {https://doi.org/10.1016/j.ins.2023.119511},
doi = {10.1016/j.ins.2023.119511},
journal = {Inf. Sci.},
month = nov,
numpages = {14}
}

@inproceedings{10.1145/3132847.3133130,
author = {Ding, Danhao and Li, Hui and Huang, Zhipeng and Mamoulis, Nikos},
title = {Efficient Fault-Tolerant Group Recommendation Using alpha-beta-core},
year = {2017},
isbn = {9781450349185},
publisher = {Association for Computing Machinery},
address = {New York, NY, USA},
url = {https://doi.org/10.1145/3132847.3133130},
doi = {10.1145/3132847.3133130},
booktitle = {Proceedings of the 2017 ACM on Conference on Information and Knowledge Management},
pages = {2047–2050},
numpages = {4},
location = {Singapore, Singapore},
series = {CIKM '17}
}

@article{LI2024111961,
title = {Maximal size constraint community search over bipartite graphs},
journal = {Knowledge-Based Systems},
volume = {297},
pages = {111961},
year = {2024},
issn = {0950-7051},
doi = {https://doi.org/10.1016/j.knosys.2024.111961},
url = {https://www.sciencedirect.com/science/article/pii/S0950705124005951},
author = {Mo Li and Renata Borovica-Gajic and Farhana M. Choudhury and Ningning Cui and Linlin Ding}
}

@inproceedings{zhang2025topr,
  title={Top-r Influential Community Search in Bipartite Graphs},
  author={Zhang, Yanxin and Hua, Zhengyu and Yuan, Long and Chen, Zi},
  booktitle={VLDB 2025 Workshop on Large Scale Graph Data Analytics (LSGDA)},
  year={2025},
  url={https://www.vldb.org/2025/Workshops/VLDB-Workshops-2025/LSGDA/LSGDA25_02.pdf}
}

@article{CSABiGXu2023,
author = {Xu, Zongyu and Zhang, Yihao and Yuan, Long and Qian, Yuwen and Chen, Zi and Zhou, Mingliang and Mao, Qin and Pan, Weibin},
title = {Effective Community Search on Large Attributed Bipartite Graphs},
journal = {International Journal of Pattern Recognition and Artificial Intelligence},
volume = {37},
number = {02},
pages = {2359002},
year = {2023},
doi = {10.1142/S0218001423590024},

URL = { 
    
        https://doi.org/10.1142/S0218001423590024
    
    

},
eprint = { 
    
        https://doi.org/10.1142/S0218001423590024
    
    

}
}

@inproceedings{Wang2006,
author = {Wang, Jun and de Vries, Arjen P. and Reinders, Marcel J. T.},
title = {Unifying user-based and item-based collaborative filtering approaches by similarity fusion},
year = {2006},
isbn = {1595933697},
publisher = {Association for Computing Machinery},
address = {New York, NY, USA},
url = {https://doi.org/10.1145/1148170.1148257},
doi = {10.1145/1148170.1148257},
booktitle = {Proceedings of the 29th Annual International ACM SIGIR Conference on Research and Development in Information Retrieval},
pages = {501–508},
numpages = {8},
location = {Seattle, Washington, USA},
series = {SIGIR '06}
}

@inproceedings{FCS-HGNN-chen-2024,
author = {Chen, Guoxin and Guo, Fangda and Wang, Yongqing and Liu, Yanghao and Yu, Peiying and Shen, Huawei and Cheng, Xueqi},
title = {FCS-HGNN: Flexible Multi-type Community Search in Heterogeneous Information Networks},
year = {2024},
isbn = {9798400704369},
publisher = {Association for Computing Machinery},
address = {New York, NY, USA},
url = {https://doi.org/10.1145/3627673.3679696},
doi = {10.1145/3627673.3679696},
booktitle = {Proceedings of the 33rd ACM International Conference on Information and Knowledge Management},
pages = {207–217},
numpages = {11},
location = {Boise, ID, USA},
series = {CIKM '24}
}

@article{HIC-Zhou-2023,
author = {Zhou, Yingli and Fang, Yixiang and Luo, Wensheng and Ye, Yunming},
title = {Influential Community Search over Large Heterogeneous Information Networks},
year = {2023},
issue_date = {April 2023},
publisher = {VLDB Endowment},
volume = {16},
number = {8},
issn = {2150-8097},
url = {https://doi.org/10.14778/3594512.3594532},
doi = {10.14778/3594512.3594532},
journal = {Proc. VLDB Endow.},
month = apr,
pages = {2047–2060},
numpages = {14}
}

@article{EECS-fang2020,
author = {Fang, Yixiang and Yang, Yixing and Zhang, Wenjie and Lin, Xuemin and Cao, Xin},
title = {Effective and efficient community search over large heterogeneous information networks},
year = {2020},
issue_date = {February 2020},
publisher = {VLDB Endowment},
volume = {13},
number = {6},
issn = {2150-8097},
url = {https://doi.org/10.14778/3380750.3380756},
doi = {10.14778/3380750.3380756},
journal = {Proc. VLDB Endow.},
month = feb,
pages = {854–867},
numpages = {14}
}

@article{PathSim-Sun-2011,
author = {Sun, Yizhou and Han, Jiawei and Yan, Xifeng and Yu, Philip S. and Wu, Tianyi},
title = {PathSim: meta path-based top-K similarity search in heterogeneous information networks},
year = {2011},
issue_date = {August 2011},
publisher = {VLDB Endowment},
volume = {4},
number = {11},
issn = {2150-8097},
url = {https://doi.org/10.14778/3402707.3402736},
doi = {10.14778/3402707.3402736},
journal = {Proc. VLDB Endow.},
month = aug,
pages = {992–1003},
numpages = {12}
}

@INPROCEEDINGS{Temp-bipartite-Li-2024,
  author={Li, Shunyang and Wang, Kai and Lin, Xuemin and Zhang, Wenjie and He, Yizhang and Yuan, Long},
  booktitle={2024 IEEE 40th International Conference on Data Engineering (ICDE)}, 
  title={Querying Historical Cohesive Subgraphs Over Temporal Bipartite Graphs}, 
  year={2024},
  volume={},
  number={},
  pages={2503-2516},
  doi={10.1109/ICDE60146.2024.00197}}

@inproceedings{k-truss-zhang-2019,
author = {Zhang, Yikai and Yu, Jeffrey Xu},
title = {Unboundedness and Efficiency of Truss Maintenance in Evolving Graphs},
year = {2019},
isbn = {9781450356435},
publisher = {Association for Computing Machinery},
address = {New York, NY, USA},
url = {https://doi.org/10.1145/3299869.3300082},
doi = {10.1145/3299869.3300082},
booktitle = {Proceedings of the 2019 International Conference on Management of Data},
pages = {1024–1041},
numpages = {18},
location = {Amsterdam, Netherlands},
series = {SIGMOD '19}
}

@ARTICLE{k-clique-yuan-2017,
  author={Yuan, Long and Qin, Lu and Zhang, Wenjie and Chang, Lijun and Yang, Jianye},
  journal={IEEE Transactions on Knowledge and Data Engineering}, 
  title={Index-Based Densest Clique Percolation Community Search in Networks}, 
  year={2018},
  volume={30},
  number={5},
  pages={922-935},
  doi={10.1109/TKDE.2017.2783933}}

@inproceedings{k-truss-xie-2025,
author = {Xie, Huan and Liu, Qing and Luo, Chengyang and Zhou, Yuhan and Gao, Yunjun},
title = {Truss-based Why-not Community Search},
year = {2025},
isbn = {9798400714542},
publisher = {Association for Computing Machinery},
address = {New York, NY, USA},
url = {https://doi.org/10.1145/3711896.3737167},
doi = {10.1145/3711896.3737167},
booktitle = {Proceedings of the 31st ACM SIGKDD Conference on Knowledge Discovery and Data Mining V.2},
pages = {3309–3320},
numpages = {12},
location = {Toronto ON, Canada},
series = {KDD '25}
}

@article{kpcore-wang-2023,
author = {Wang, Yuxiang and Liu, Jun and Xu, Xiaoliang and Ke, Xiangyu and Wu, Tianxing and Gou, Xiaoxuan},
title = {Efficient and Effective Academic Expert Finding on Heterogeneous Graphs through (k, $P$)-Core based Embedding},
year = {2023},
issue_date = {July 2023},
publisher = {Association for Computing Machinery},
address = {New York, NY, USA},
volume = {17},
number = {6},
issn = {1556-4681},
url = {https://doi.org/10.1145/3578365},
doi = {10.1145/3578365},
journal = {ACM Trans. Knowl. Discov. Data},
month = mar,
articleno = {85},
numpages = {35}
}

@ARTICLE{kcore-fang-2019,
  author={Fang, Yixiang and Wang, Zhongran and Cheng, Reynold and Wang, Hongzhi and Hu, Jiafeng},
  journal={IEEE Transactions on Knowledge and Data Engineering}, 
  title={Effective and Efficient Community Search Over Large Directed Graphs}, 
  year={2019},
  volume={31},
  number={11},
  pages={2093-2107},
  doi={10.1109/TKDE.2018.2872982}}

@INPROCEEDINGS {Scalable-CS-wang-2024,
author = { Wang, Yuxiang and Ye, Shuzhan and Xu, Xiaoliang and Geng, Yuxia and Zhao, Zhenghe and Ke, Xiangyu and Wu, Tianxing },
booktitle = { 2024 IEEE 40th International Conference on Data Engineering (ICDE) },
title = {{ Scalable Community Search with Accuracy Guarantee on Attributed Graphs }},
year = {2024},
volume = {},
ISSN = {},
pages = {2737-2750},
doi = {10.1109/ICDE60146.2024.00214},
url = {https://doi.ieeecomputersociety.org/10.1109/ICDE60146.2024.00214},
publisher = {IEEE Computer Society},
address = {Los Alamitos, CA, USA},
month =May}

@INPROCEEDINGS{ee-want-2021,
  author={Wang, Kai and Zhang, Wenjie and Lin, Xuemin and Zhang, Ying and Qin, Lu and Zhang, Yuting},
  booktitle={2021 IEEE 37th International Conference on Data Engineering (ICDE)}, 
  title={Efficient and Effective Community Search on Large-scale Bipartite Graphs}, 
  year={2021},
  volume={},
  number={},
  pages={85-96},
  doi={10.1109/ICDE51399.2021.00015}}

@ARTICLE{Co-Engaged-haldar-2024-TKDE,
  author={Haldar, Nur Al Hasan and Li, Jianxin and Akhtar, Naveed and Jia, Yan and Mian, Ajmal},
  journal={IEEE Transactions on Knowledge and Data Engineering}, 
  title={Co-Engaged Location Group Search in Location-Based Social Networks}, 
  year={2024},
  volume={36},
  number={7},
  pages={2910-2926},
  doi={10.1109/TKDE.2023.3327405}}

@ARTICLE{Haldar-top-k-2023-TKDE,
  author={Haldar, Nur Al Hasan and Li, Jianxin and Ali, Mohammed Eunus and Cai, Taotao and Chen, Yunliang and Sellis, Timos and Reynolds, Mark},
  journal={IEEE Transactions on Knowledge and Data Engineering}, 
  title={Top-k Socio-Spatial Co-Engaged Location Selection for Social Users}, 
  year={2023},
  volume={35},
  number={5},
  pages={5325-5340},
  doi={10.1109/TKDE.2022.3151095}}

@ARTICLE{NGCS-wu-2022-access,
  author={Wu, Zewen and Xu, Jian and Zhang, Huaixiang and Bao, Qing and Qingsun and Changbengzhou},
  journal={IEEE Access}, 
  title={A Progressive Approach for Neighboring Geosocial Communities Search Over Large Spatial Graphs}, 
  year={2022},
  volume={10},
  number={},
  pages={57012-57024},
  doi={10.1109/ACCESS.2022.3168361}}

@ARTICLE{geo-social-wang-2022-tkde,
  author={Wang, Kai and Wang, Shuting and Cao, Xin and Qin, Lu},
  journal={IEEE Transactions on Knowledge and Data Engineering}, 
  title={Efficient Radius-Bounded Community Search in Geo-Social Networks}, 
  year={2022},
  volume={34},
  number={9},
  pages={4186-4200},
  doi={10.1109/TKDE.2020.3040172}}

@ARTICLE{topk-cs-rai-2023-tkde,
  author={Rai, Niranjan and Lian, Xiang},
  journal={IEEE Transactions on Knowledge and Data Engineering}, 
  title={Top-$k$ Community Similarity Search Over Large-Scale Road Networks}, 
  year={2023},
  volume={35},
  number={10},
  pages={10710-10721},
  doi={10.1109/TKDE.2023.3243177}}

@INPROCEEDINGS {multi-att-cs-guo-2021,
author = { Guo, Fangda and Yuan, Ye and Wang, Guoren and Zhao, Xiangguo and Sun, Hao },
booktitle = { 2021 IEEE 37th International Conference on Data Engineering (ICDE) },
title = {{ Multi-attributed Community Search in Road-social Networks }},
year = {2021},
volume = {},
ISSN = {},
pages = {109-120},
doi = {10.1109/ICDE51399.2021.00017},
url = {https://doi.ieeecomputersociety.org/10.1109/ICDE51399.2021.00017},
publisher = {IEEE Computer Society},
address = {Los Alamitos, CA, USA},
month =apr}

@ARTICLE{group-plianing-ahmed-2022,
  author={Al-Baghdadi, Ahmed and Sharma, Gokarna and Lian, Xiang},
  journal={IEEE Transactions on Knowledge and Data Engineering}, 
  title={Efficient Processing of Group Planning Queries Over Spatial-Social Networks}, 
  year={2022},
  volume={34},
  number={5},
  pages={2135-2147},
  doi={10.1109/TKDE.2020.3004153}}

@inproceedings{NIPS2012_7a614fd0,
 author = {Leskovec, Jure and Mcauley, Julian},
 booktitle = {Advances in Neural Information Processing Systems},
 editor = {F. Pereira and C.J. Burges and L. Bottou and K.Q. Weinberger},
 pages = {},
 publisher = {Curran Associates, Inc.},
 title = {Learning to Discover Social Circles in Ego Networks},
 url = {https://proceedings.neurips.cc/paper_files/paper/2012/file/7a614fd06c325499f1680b9896beedeb-Paper.pdf},
 volume = {25},
 year = {2012}
}

@InProceedings{Epinions,
author="Richardson, Matthew
and Agrawal, Rakesh
and Domingos, Pedro",
editor="Fensel, Dieter
and Sycara, Katia
and Mylopoulos, John",
title="Trust Management for the Semantic Web",
booktitle="The Semantic Web - ISWC 2003",
year="2003",
publisher="Springer Berlin Heidelberg",
address="Berlin, Heidelberg",
pages="351--368",
isbn="978-3-540-39718-2"
}

@inproceedings{dblp_dataset,
author = {Yang, Jaewon and Leskovec, Jure},
title = {Defining and evaluating network communities based on ground-truth},
year = {2012},
isbn = {9781450315463},
publisher = {Association for Computing Machinery},
address = {New York, NY, USA},
url = {https://doi.org/10.1145/2350190.2350193},
doi = {10.1145/2350190.2350193},
booktitle = {Proceedings of the ACM SIGKDD Workshop on Mining Data Semantics},
articleno = {3},
numpages = {8},
location = {Beijing, China},
series = {MDS '12}
}

@InProceedings{cal_road_netword,
author="Li, Feifei
and Cheng, Dihan
and Hadjieleftheriou, Marios
and Kollios, George
and Teng, Shang-Hua",
editor="Bauzer Medeiros, Claudia
and Egenhofer, Max J.
and Bertino, Elisa",
title="On Trip Planning Queries in Spatial Databases",
booktitle="Advances in Spatial and Temporal Databases",
year="2005",
publisher="Springer Berlin Heidelberg",
address="Berlin, Heidelberg",
pages="273--290",
isbn="978-3-540-31904-7"
}

\appendix
{
\section{The Proofs}
\label{appendix:sec:proofs}
\noindent\textbf{Proof of the Lemma \ref{lemma:keyword-based-pruning-u}}
\begin{proof} From our lemma assumption, if  $p.K \cap Q = \emptyset$ holds for all POIs $p \in V_p$ that user $u$ visited, then, for any vertex subset, $V_p'$, of $V_p$ in a community (subgraph) of graph $G_b$, we have $p.K \cap Q = \emptyset$ for all POIs $p \in V_p'$ (as $V_p' \subseteq V_p$). Based on  
Definition \ref{def:awcore}, user $u$ cannot be in the $(\omega, \pi)\mbox{-}keyword\mbox{-}core$, and can thus be safely pruned.
\end{proof}

\noindent\textbf{Proof of the Lemma \ref{lemma:omega-based-pruning}}
\begin{proof}
Since $ub\_f_{sum}(u,V_p')$ is an upper bound of $f_{sum}(u,V_p')$, we have $ub\_f_{sum}(u,V_p') \geq f_{sum}(u,V_p')$. From our lemma assumption, it holds that $ub\_f_{sum}(u,V_p')< \omega$. Hence, by the inequality transition, we can derive that $f_{sum}(u,V_p') < \omega$, and safely prune user $u$ with a low $f_{sum}(u,V_p')$ value. 
\end{proof}

\noindent\textbf{Proof of the Lemma \ref{lemma:pi-based-pruning}}
\begin{proof}
Since $ub\_f_{avg}(u)$ is an upper bound of $f_{u,p}$ for all $p$ visited by $u$, we have $ub\_f_{avg}(u) \geq f_{avg}(V_s',p)$. From our lemma assumption, we have $ub\_f_{avg}(u)< \pi$. Hence, by the inequality transition, we can derive $f_{avg}(V_s',p) < \pi$, and safely prune the user $u$ with a low $f_{avg}(V_s',p)$ value.
\end{proof}

\noindent\textbf{Proof of the Lemma \ref{lemma:structural-cohesiveness-pruning}}
\begin{proof}
Since $ub\_sup(u)$ is the upper bound of $sup(e)$ of all edges between user $u$ and its neighbors, then, according to Definition \ref{def:kdtruss}, we can safely prune $u$ if $ub\_sup(u)< k-2$.
\end{proof}

\section{The Index Pivots Selection for Leaf Nodes}
\label{appendix:sec:index_pivot_selection}

\noindent\textbf{\textit{Bipartite Structure.}} Since each $key_j \in u.K$, has the aggregate $key_j.f_{sum}$ and $key_j.f_{max}$   (as given in Eq.~(\ref{eq:key-f-sum}) and Eq.~(\ref{eq:key-f-max}), respectively). To make these aggregations comparable, we normalize $key_j.f_{sum}$ and $key_j.f_{max}$ into the range [0-1] by dividing it by the largest $f_{sum}$ and $f_{max}$, respectively. Then, for any two users $u,v \in G_s$, we calculate the bipartite structure score $bs\_score(u,v)$ as follows:

\begin{eqnarray}\label{eq:bsscores}
\begin{aligned}
&bs\_score(u,v)=\sum_{\forall key_j\in  \{u.K\cap v.K\}} \\
&\frac{(u.key_j.f_{sum} + v.key_j.f_{sum})}{largest(f_{sum})} + \frac{(u.key_j.f_{max} + v.key_j.f_{max})}{largest(f_{max})}
\end{aligned}
\end{eqnarray}
, where $d=d_{max}$.

The higher the score $bs\_score(u,v)$, the higher the probability of sharing similar keywords, which implies that their interests are more similar. 

\noindent\textbf{\textit{Spatial Structure.}} We calculate the spatial structure score based on the average distance we provide in Eq.~(\ref{eq:avgdist}). Then, for any users $u,v \in G_s$. we calculate $rs\_score(u,v)$ as follows:

\begin{eqnarray}
&rs\_score(u,v)= 
&\frac{\frac{\sum_{\forall p\in v.L}avg\_dist_r(u,p)}{|v.L|}}{largest(rs\_score(.))}
\label{eq:rsscore}
\end{eqnarray}

The lower the value of $rs\_score(u,v)$ between user $u$ and $v$, the closer the distance between them. Similarly to $bs\_score(.)$, we normalize $rs\_score(.)$ to the range [0-1] by dividing it to the largest $rs\_score(.)$.

\noindent\textbf{\textit{Social Structure.}} We must consider $(k,d)\mbox{-}truss$ and topic influences among users to measure social closeness in the social network. Furthermore, we need to consider social distance and social cohesiveness when calculating $(k,d)\mbox{-}truss$, as stated in Definition \ref{def:kdtruss}. Thus, for any two users $u,v \in G_s$, we calculate the social cohesion by finding $sum\_sup(u,v)=ub\_sup(u)+ub\_sup(v)$, where $ub\_sup()$ is given in Eq.~(\ref{eq:ubsup}). 
On the other hand, we use Eq.~(\ref{eq:ub-ISF}) to calculate the influence $ub\_ISF(u,v)$. In both $sum\_sup(.,.)$ and Eq.~(\ref{eq:ub-ISF}), values generally lead to greater cohesiveness. In contrast,  Definition \ref{def:kdtruss} requires a small social distance $dist_s(u,v)$ between users. To make all $sum\_sup(.), ub\_ISF(u,v),$ and $dist_s(u,v)$ comparable, we normalize all to the range [0-1] by dividing it by the largest $sum\_sup(.)$, $ub\_ISF(.)$, $dist_s(.)$, respectively. Finally, to determine the social structure score, we use the following calculation:

\begin{eqnarray}\label{eq:ssscore}
&&\hspace{-2ex}ss\_score(u,v)\\
&\hspace{-2ex}=& \hspace{-2ex}\frac{sum\_sup(u,v)}{largest(sum\_sup(.))} + \frac{ub\_ISF(u,v)}{largest(ub\_ISF(.))} + \frac{(1- dist_s(u,v))}{largest(dist_s(.))}\notag
\end{eqnarray}

Obviously, the higher the $ss\_score(u,v)$ score indicates more coherence, a closer social distance, and a higher degree of influence between $u$ and $v$.

\section{The Index Pivots Selection for Non-Leaf Nodes}
\label{appendix:sec:tree_construction}

\textbf{\textit{Bipartite Structure Score for Non-Leaf Nodes.}} As explained in Section \ref{subsec:indexing_cons}, each node has a set $N.K$, where each $key_j \in N.K$ is associated with $key_j.ub\_f_{sum}$ and $key_j.ub\_f_{max}$  (as given in Eq.~(\ref{eq:key-ub-f-sum}) and Eq.~(\ref{eq:key-ub-f-max}), respectively). To make $key_j.ub\_f_{sum}$ and $key_j.ub\_f_{max}$ comparable, we normalize the values to the range [0-1] by dividing them to the largest $bs\_f_{sum}(.)$ and $bs\_f_{max}(.)$, respectively. Then, for a user index pivot $piv'_i \in \mathbb{P'}_{index}$ ($\mathbb{P'}_{index}\subset \mathbb{P}_{index}$), we calculate the bipartite structure score for a node $N$ as follows: 
\begin{eqnarray}\label{eq:bsscore-node}
&&bs\_score\_node(N,piv'_i)\\
&=&bs\_ub\_f_{sum}(N,piv'_i) +  bs\_ub\_f_{max}(N,piv'_i),\notag
\end{eqnarray}
where $bs\_ub\_f_{sum}(N,piv'_i)$ and $bs\_ub\_f_{max}(N,piv'_i)$ calculated as follows:

\begin{eqnarray}
&&bs\_ub\_f_{sum}(N,piv'_i)\notag\\
&=&\frac{\sum_{\forall key_j\in  \{N.K\cap piv'_i.K\}}
N.key_j.ub\_f_{sum} +  piv'_i.key_j.ub\_f_{sum}}{largest(bs\_f_{sum}(.))}\notag\\ 
\\&&bs\_ub\_f_{max}(N,piv'_i)\notag\\
&=&\frac{\sum_{\forall key_j\in  \{N.K\cap piv'_i.K\}}
N.key_j.ub\_f_{max} +  piv'_i.key_j.ub\_f_{max}}{largest(bs\_f_{max}(.))}\notag
\end{eqnarray}
where $d=d_{max}$.

The higher $bs\_score\_node(N,piv'_i)$ score implies that the interest in visiting similar places is more similar between the users in node $N$ and the pivot user $piv'_i$. 

\textbf{\textit{Social Structure Score for Non-Leaf Nodes.}} The social structure score is mainly dependent on constraints within the social network such as the $(k,d)\mbox{-}truss$ and topic influences among users. We use the parameters of node $N$ to calculate the social structure score between $N$ and a $piv'_i \in \mathbb{P'}_{index}$. We calculate the social structure score for a node $N$ and a $piv'_i \in \mathbb{P'}_{index}$ as follows: 

\begin{eqnarray}\label{eq:ssscore-node}
&&\hspace{-5ex} ss\_score\_node(N,piv'_i) \\
&\hspace{-2ex}=& sum\_sup\_node(N,piv'_i) + max\_ub\_ISF(N,piv'_i)\notag\\
&&+ (1- (lb\_dist_s(N,piv'_i))\notag
\end{eqnarray}

, where the normalized $sum\_sup\_node(N,piv'_i)$ calculated as follows:
\begin{eqnarray}\notag
\begin{aligned}
sum\_sup\_node(N,piv'_i)=\frac{ub\_sup(N) + ub\_sup(piv'_i)}{largest(sum\_sup\_node(.))}
\end{aligned}
\end{eqnarray}

, where $ub\_sup(N)$ and $ub\_sup(piv'_i)$ are given in Eqs.~(\ref{eq:ub-sub-N}) and (\ref{eq:ubsup}), respectively. 

In addition, we calculate the normalized $max\_ub\_ISF(N,piv'_i)$ as follows:

\begin{eqnarray}\notag
\begin{aligned}
max\_ub\_ISF(N,piv'_i)=\frac{ ub\_ISF(piv'_i,N)}{largest(max\_ub\_ISF(.))}
\end{aligned}
\end{eqnarray}
where $ub\_ISF(piv'_i,N)$ is given in Eqs.~ (\ref{eq:ub1-ISF_indel-level}). 

Finally, we compute the distance lower bound $lb\_dist_s(piv'_i,N)$, where $lb\_dist_s(.)$ is given in Eq.~(\ref{eq:lb-dist-indexing}).We normalize $lb\_dist_s(.)$  to the range [0-1].

Based on Definition \ref{def:kcsbssn}, we want to assign nodes to subgroups that maximize both $sum\_sup\_node(.)$ and $max\_ub\_ISF(.)$ and minimize $lb\_dist_s(.)$, which generally leads to greater cohesiveness and a smaller social distance.
}
\end{document}